\documentclass[authoryear,a4paper,11pt]{elsarticle}
\usepackage{natbib}
\usepackage{amsthm}
\usepackage{amssymb}
\usepackage{graphicx,amsmath}
\usepackage{tikz}
\usetikzlibrary{matrix}
\usepackage{mathtools}
\usepackage{natbib}
\usepackage{footmisc} 
\usepackage{booktabs,xcolor}
\usepackage{booktabs}
\usepackage{multirow}
\usepackage{graphics}
\usepackage{enumitem}
\usepackage{subcaption}
\usepackage{appendix}
\usepackage{lscape}
\usepackage{adjustbox}
\usepackage[flushleft]{threeparttable}
\usepackage{multirow}
\usepackage{setspace} 
\usepackage[margin = 2.4cm]{geometry} 
\usepackage{comment}
\usepackage{caption}
\usepackage{setspace}
\usepackage[colorlinks=true,linkcolor=blue, urlcolor=blue, citecolor=blue, breaklinks]{hyperref}
\journal{Games and Economic Behavior}

\newtheorem{proposition}{Proposition}
\newtheorem{corollary}{Corollary}

\newtheorem{lemma}{Lemma}

\theoremstyle{definition}

\newtheorem{example}{Example}

\newtheorem*{lemma*}{Lemma}
\newtheorem{assumption}{Assumption}

\usepackage{etoolbox}
\usepackage{datetime}
\usepackage[flushleft]{threeparttable}
\makeatletter
\patchcmd{\ps@pprintTitle}
{Preprint submitted to}
{Accepted subject to minor revisions to }
{}{}

\makeatother
\patchcmd{\emailauthor}{(#2)}{}{}{}
\patchcmd{\urlauthor}{(#2)}{}{}{}
\usepackage{xpatch}
\xpatchcmd{\MaketitleBox}{\hrule}{}{}{}
\xpatchcmd{\MaketitleBox}{\hrule}{}{}{}

\newcommand{\Ch}{\text{Ch}} 
\newcommand{\sosm}{\textup{\tiny SOSM}} 
\newcommand{\rk}{\textup{rk}} 

\newtheorem*{theorem*}{Theorem}
\theoremstyle{plain} 
\newcommand{\thistheoremname}{}
\newtheorem*{genericthm*}{\thistheoremname}
\newenvironment{namedthm*}[1]
{\renewcommand{\thistheoremname}{#1}%
	\begin{genericthm*}}
	{\end{genericthm*}}

\newcommand\encircle[1]{%
	\tikz[baseline=(X.base)] {
	\node (X) [draw, shape=circle, inner sep=0] {\strut #1};}}
\newcommand\ensquare[1]{%
	\tikz[baseline=(X.base)] {
	\node (X) [draw, shape=rectangle, inner sep=.1cm] {\strut #1};}}

\tikzset{ 
	table/.style={
		matrix of nodes,
		row sep=-\pgflinewidth,
		column sep=-\pgflinewidth,
		nodes={rectangle,text width=3em,align=center},
		text depth=1.25ex,
		text height=2.5ex,
		nodes in empty cells
	},
	row 1/.style={nodes={text depth=0.4ex,text height=2ex}},
	row 6/.style={nodes={text depth=0.4ex,text height=2ex}}
}

\begin{document}

	\begin{frontmatter}
		
		\title{School Choice with Independent versus Consolidated Districts}

\author[add1,add3]{Thilo Klein}		\ead{thilo.klein@zew.de (Corresponding author).}

\author[add2]{Robert Aue}
 
\author[add5]{Josu\'e Ortega}

\address[add1]{ZEW -- Leibniz-Centre for European Economic Research, Mannheim, Germany.}
\address[add3]{Pforzheim University, Germany.}	
\address[add2]{Schleswig-Holstein State Court of Audit, Germany.}
\address[add5]{Queen’s University Belfast, UK.}	

\date{\today}

\begin{abstract}
This paper studies the welfare effects of school district consolidation. Using incomplete rank-ordered lists (ROLs) submitted for admission to the Hungarian secondary school system, we estimate complete ROLs assuming that parents do not use dominated strategies and that the matching outcome is stable. These estimates aid in constructing a counterfactual district-based assignment and discerning the factors driving parents' preferences over schools. We find that district consolidation leads to large welfare gains in Budapest, equivalent to students attending a school five kilometres closer to their residences. These gains offset the additional travel distances incurred in the consolidated assignment. 73\% of matched students benefit from district consolidation, while fewer than 3\% are assigned to a less preferred school. Students from smaller and less under-demanded districts benefit relatively more, as well as those with high academic ability. Using reported preferences instead of estimated ones also yields large gains from district consolidation.

\end{abstract}

\begin{keyword}
	school district consolidation \sep  school choice.\\
	
	{\it JEL Codes:} C78, I21.
\end{keyword}
\end{frontmatter}

\newpage

	\setcounter{footnote}{0}
\setlength{\parskip}{3mm} 
	\newpage

\section{Introduction}
\label{sec:introduction}
	\doublespacing
School district consolidation is the process whereby previously independent school districts are merged so that students can now choose from a greater set of alternative schools, and can be undertaken to reduce administrative costs or to foster integration of racially and economically segregated areas. This phenomenon has taken place in the U.S. for over one hundred years: the number of school districts has fallen from around 125,000 in 1900 to under 15,000 today \citep{brasington1999}. School district consolidations have also occurred in several other countries, including Germany \citep{riedel2010school}, Hungary \citep{bukodi2008hungarian}, Sweden \citep{soderstrom2010school}, and New Zealand \citep{waslander1995choice}. 

However, school district consolidation is rarely a smooth process and is often met with reluctance by some of the independent districts that are to integrate \citep{berry2008growing}. District consolidation even reverses
sometimes: in 128 occasions, communities in the U.S. have attempted to secede from their school district to create their own. These procedures, which sometimes are as simple as a majority vote in the breakaway neighbourhood, have succeeded 73 times \citep{edbuild}.\footnote{Seceding districts are whiter and more affluent on average than the districts from which they seceded. Their secession is associated with significant increases in segregation after adjusting for prior trends \citep{richards}.} One of the many reasons behind the reluctance of districts to merge (or their incentives to secede) is the concern that local students attend worse schools under consolidated school choice \citep{fairman2012school}. This concern is not entirely unwarranted, as district consolidation not only leads to more choice for students, but also to more competition for a place in their preferred schools. Which effect dominates is unclear a priori and depends on many factors, not least on students' characteristics and preferences. In this article, we shed light on the welfare effects of school district consolidation with an empirical analysis of the Hungarian nationwide school assignment system, guided by a simple theoretical model. 

In our model, we study district consolidation as the merger between disjoint Gale-Shapley many-to-one matching markets that differ in their size and their ratio between students and school seats. Students are assigned to schools using the student-optimal stable matching (SOSM) before and after consolidation takes place, but before district consolidation students can only attend schools within their own district.\footnote{The SOSM is the most preferred stable matching for all students. It is consistently implemented in real-life school choice and college admissions in several regions, including Boston \citep{abdulkadirouglu2014elite}, Chile \citep{hastings2013some,correa2019school}, England \citep{burgess2019school}, Hungary \citep{biro2008student}, Paris \citep{hiller2014choix} and Spain \citep{mora2001understanding}.}  
District consolidation can, in some cases, weakly harm all students. 
However, we find that this is unlikely to occur in a random school choice problem. Proposition \ref{thm:prop3} shows that district consolidation generates expected welfare gains for all students, particularly for those who belong to districts that are small or have a smaller over-supply of school seats.

These theoretical predictions are compared to empirical results that are obtained by using the data from secondary school admissions in Hungary, and in particular, from its capital city Budapest during 2015. We focus on Budapest because i) we have data on students' stated preferences over all schools in its 23 districts, as well as schools' priorities over students from the 23 districts; ii) students are assigned using SOSM \citep{biro2008student}; iii) Hungary consolidated primary school districts in 2013 \citep{kertesi2013school}, thus the analysis of the unconsolidated case for secondary schools is particularly meaningful; and iv) we have additional data on students' and schools' characteristics that reveal which school features drive students' preferences, such as schools' previous results in mathematics and Hungarian, distance to the students' home addresses, and socio-economic status. Our empirical strategy is to compare the SOSM in the integrated market to the matching that results in a counterfactual disintegrated market. In order to compute the counterfactual matchings, we construct a complete set of preferences over all market participants, estimating a parametric form of students' preferences over schools, and schools' priorities over students. 

Despite our data being remarkably detailed, an analysis of two-sided markets still needs to overcome two technical problems. The first issue that needs to be addressed concerns estimating students' preferences. In the student-proposing deferred acceptance algorithm, used to compute the SOSM, it is only a \emph{weakly} dominant strategy for students to report their complete rank-order lists (ROLs) of schools truthfully. Therefore, stated ROLs may differ from real ones because students either omit schools which they deem unattainable or truncate their ROLs if they are confident that they will be assigned to more preferred schools.\footnote{A large literature directly uses reported preferences in the deferred acceptance mechanism as true preferences (see \cite{abdulkadirouglu2017welfare}, \cite{che2019efficiency}, and references therein). Our estimation procedure goes one step further by also using estimated preferences.} Both types of omissions have consistently been observed in the field \citep{chen2019} and in the lab \citep{castillo2016truncation}; and both are particularly important for us because the average student in Budapest ranks only four schools, even when they are allowed and encouraged to rank all schools at no cost. The fact that students submit rather short preference lists is the reason why we need a parametric approach to construct their ``true'' complete ROLs. 

A second, closely linked technical complication concerns the estimation of schools' priorities: schools only report their priorities over the set of students who actually apply to them. Priorities over the entire set of students cannot be recovered because the admission rules are either unknown or depend on unobservable criteria.
In Hungary, schools' priorities are based on tests and interviews. These criteria are only observable for schools for which a student applied. Furthermore, the weights with which these criteria determine the priorities are unknown because they are decided by each school, subject to basic governmental guidelines.\footnote{ This is not a special feature of the Hungarian school market. Instead, it is common in many other school choice systems around the world that also use either unknown weights or at least one unobserved criterion. Examples include school choice in Amsterdam \citep{oosterbeek2021preference}, Chicago \citep{pathak2013school}, England \citep{bertoni2020academies,burgess2019school} and Germany \citep{grenet2022preference}. Unobserved criteria include geographical distance, religion, interviews, lottery numbers and non-standardized tests. Geographical distance is often not available in administrative research data because of privacy and data protection issues.}
Therefore, even though \cite{fack2019beyond} have shown how to estimate students' preferences without assuming truth-telling behaviour, we cannot directly apply their two-step and discrete choice estimators in such settings, because they rely on observing complete schools' priorities over students (for example, when schools' priorities are based on a centralised exam). 

To overcome these technical challenges in preference estimation, our empirical strategy builds on two identifying assumptions, namely i) the stability of the final assignment, and ii) that students only use strategies that are undominated. The combination of these two assumptions still allows for students to skip seemingly unattainable schools, or to truncate their ROL at the bottom, and thus it is still weaker than assuming truthful preference revelation from students. Our estimation method is implemented as a Gibbs sampler that imposes bounds on the latent match valuations that are derived from both identifying assumptions. This approach generalises the matching estimator proposed by \cite{Logan2008} and \cite{Menzel2013} from a one-to-one to a many-to-one matching setting. 
We test our proposed estimation method in Monte-Carlo simulations, in which students strategically choose their optimal application portfolio, given their equilibrium beliefs about admission probabilities. We find that the method yields unbiased estimators for students' preferences and schools' priorities.

Our main finding is that the consolidated school market in Budapest is advantageous for the majority of students and yields significant welfare gains when compared to the counter-factual situation in which students only attend schools within their home districts.  This result is robust to whether the counterfactual matching is obtained with reported or estimated students' preferences. The welfare gains from school district consolidation are equivalent to attending a school that is five kilometres closer to the students' home address. In other words, the average student would be willing to incur an additional travel distance of five kilometres to attend their assigned school in the consolidated market, rather than the counterfactual assigned school in their home district.  Only a small fraction of students (less than 3\%) are harmed by district consolidation.

We empirically confirm our theoretical result which states that students who live in smaller districts or districts with less school capacity benefit more from school district consolidation than the average student. Students with high academic ability also benefit relatively more, but we do not find evidence that this extends to students from high socioeconomic status.
The median student incurs a welfare gain that is positive and almost as large as the average welfare gain. 
To explain these large utility gains, we decompose the total gains into a choice effect and a competition effect. We find that the substantial welfare gains are largely due to an enhanced choice set, and that the consolidated market does not lead to significantly increased competitive pressure. 

Furthermore, the parametric specification of students' utility from choosing a school yields insights into the determinants of students' preferences. We find that travel distance is an important factor that determines students' choices, but students also prefer schools with high average academic achievement and average socio-economic status. Our results further imply that students dislike schools which hold additional oral entrance exams, all else equal. Moreover, we find that students have assortative preferences. For instance, students with a high socio-economic background have a stronger preference for schools with a high average socio-economic status than other students. The same holds for students who are particularly strong in mathematics or Hungarian language.

\paragraph{Organisation of the article} This article proceeds as follows. 
Section \ref{sec:literature} discusses the related literature. 
Section \ref{sec:model} presents our theoretical results. 
Section \ref{sec:data} introduces our data. 
Section \ref{sec:strategy} presents the estimation strategy. 
Section \ref{sec:results} discusses our empirical results and performs robustness tests.
Section \ref{sec:conclusion} concludes.

\section{Related literature}
\label{sec:literature}

Although there is an extensive empirical literature studying school district consolidation, the majority of it is unrelated to that of matching markets. This literature has four main findings: i) there is evidence of overall improvement in students' performance after district consolidation, yet these improvements are not uniformly distributed and there may be losses for specific groups of students \citep{berry2005school,berry2008growing,cox2010decade,  leach2010effects};\footnote{There is also a well-established link relating larger school sizes with lower students' performances, which is not the focus of this article.} ii) small and look-alike districts are more likely to merge \citep{brasington1999,gordon2009spatial}; iii) there is empirical evidence of increased fiscal efficiency due to district consolidation \citep{duncombe1995potential,howley2011consolidation}, and iv) district consolidation has diversified the racial composition of schools \citep{albsury2005,siegel2017solidifying}.

More related to our paper is the theoretical and empirical literature studying the consequences of integrated school choice replacing decentralized admissions or neighbourhood choice. The theoretical literature focuses on understanding either residential patterns that emerge in response to school choice \citep{xu2019housing,avery2021distributional,grigoryan2021school}, the inefficiency caused by students receiving multiple admission offers under decentralized admission systems \citep{che2016,manjunath2016,dougan2019unified,hafalir2018jet}, the incentives for schools to join a centralized clearinghouse \citep{ekmekci2019common} or the inter-district policies that can be achieved with consolidation \citep{hafalir2022interdistrict}. Although our model does not share any of those features, it is the first to provide testable hypotheses on the relationship between the gains from district consolidation and the size and imbalance of each individual district.
Nonetheless, our paper agrees with this literature in that centralized systems lead to efficiency gains for most students. The welfare effects are ambiguous for low-income or low-ability students, who may be harmed (e.g. as in \citealt{hafalir2018jet,avery2021distributional}) or benefit (e.g. as in \citealt{xu2019housing,grigoryan2021school}) by integrated school choice, depending on the specifics of the model (we empirically find no statistically significant difference between the gains from integration experiences by students with high versus low socioeconomic status).

There is also an empirical literature analysing the emergence of centralized admission systems that replace decentralized applications in  Brazil \citep{machado2021centralized,mello2022centralized}, Japan \citep{tanaka2020meritocracy}, the US \citep{liu2007diffusion,knight2022reducing}, and in particular, New York City \citep{abdulkadirouglu2017welfare} and Los Angeles \citep{camposqje}. These studies have some common findings: centralized school choice delivers more student mobility and overall large welfare gains (equivalent to the average student attending a school 17 kilometres closer to a student's home address in the case of NYC; we document a smaller welfare gain corresponding to 5 kilometres). Centralized choice also increases students' mobility in all the papers in this literature. However, there is no clear consensus on the effects of centralized choice on diversity or academic scores.

Our paper also contributes to the literature studying how to estimate students' preferences and schools' priorities from observed data. The most common identifying assumption is truth-telling, where under the SOSM, a student is truth-telling if she submits her $k$ most preferred schools. \cite{abdulkadirouglu2017welfare} 
and \cite{che2019efficiency}, for example, follow this truth-telling assumption in their analysis of the New York City high school match. However, truth-telling is only a weakly dominant strategy, even when schools can be listed at no cost. Commonly observed and rationalisable strategies that are inconsistent with truth-telling include skipping ``infeasible'' schools and truncating ROLs after ``safe'' schools. Therefore, we use two weaker identifying assumptions that have been explored in the literature.

The estimator presented in this article is based on combining the stability and undominated strategies assumptions. Our first identifying assumption is that students only use undominated strategies, i.e.\ they do not swap their true preference orderings over schools when submitting a ROL. The second identifying assumption is stability of the observed matching, which implies that a student's assigned school must be the top choice among her ex-post feasible schools. Stability is a more innocuous assumption than truth-telling in that it permits inconsequential mistakes \citep{artemov2017strategic}. We extend the estimator in \cite{fack2019beyond} to a context where students' feasible choice sets are unobserved, and so we extend it to include latent feasible choice sets using a data augmentation approach. The combination of stability and undominated strategies allows us to point identify our parameters of interest, as we show by means of Monte Carlo simulations.\footnote{There is also a literature that focuses on estimating students' preferences using ROLs reported in mechanisms that are not strategy-proof \citep{he2015gaming,calsamiglia2020structural}.}

\section{Model}
\label{sec:model}
We study district consolidation by extending the school choice framework of \cite{abdulkadirouglu2003}. An extended school choice problem (ESCP) is a tuple $(T, S, \mathcal{D}, \succ, \rhd, q)$, where i) $T$ is the set of students, ii) $S$ is the set of schools, iii) $q$ is the number of school seats in each school, iv)  $\mathcal{D} \coloneqq \{D_1,\ldots, D_r\}$ is a partition of $T \cup S$ into districts, v) $\succ_t$ is the strict preference ordering of student $t$ over all schools, and vi) $\rhd_s$ is the strict priority structure of school $s$ over all students. 
The admission policy of each school $s$ is given by a choice rule $\Ch_s:2^{T} \times \{q_s\} \mapsto2^{T}$, which maps every non-empty subset of students $T'$ to a subset $\Ch_s(T',q_s) \subseteq T'$ of size smaller or equal than $q_s$. $\Ch_s(\cdot, q_s)$ is responsive to the priority ranking $\rhd_s$.\footnote{The choice function $\Ch_s$ is responsive to priorities if for each $T'\subseteq T$, $\Ch_s(T',q_s)$ is obtained by choosing the highest-priority students in $T'$ until $q_s$ students are chosen.}

 $T^{D_i}$ and $S^{D_i}$ denote the non-empty sets of students and schools in district $D_i$. A {population} $P$ is the union of some (possibly all) districts. Each district $D_i$ has $q n_i$ students, $n_i+k_i$ schools and $q(n_i+k_i)$ school seats, where $k_i$ is a positive integer that reflects the imbalance between the supply and demand for school seats in each district. We use $K \coloneqq \sum_i^r k_i > 0$, i.e. the society as a whole is under-demanded and the size of its unbalance is $K$.\footnote{This assumption is satisfied in our data. Subsequent work by \cite{kamali2023welfare} studies the scenario in which there are more students that school seats in the integrated market.} We also use $N\coloneqq \sum_i^r n_i$.

Given a population $P$ with students $T^P$ and schools $S^P$, a {matching} $\mu:T^P \cup S^P \mapsto T^P \cup S^P$ is a correspondence such that for each $(t,s)\in T^P \times S^P$, $\mu(t)\in S^P$, $\mu(s)\subseteq T^P$, $\left\vert \mu(s) \right\vert \leq q_s$ and $\mu(t)=s$ if and only if $t \in \mu(s)$. 
A {matching scheme} $\sigma$ is a function that specifies a matching for each district $D_i$, denoted by $\sigma(\cdot, D_i)$, as well as for the society as a whole, denoted by $\sigma(\cdot,\Omega)$.\footnote{As no confusion shall arise, when referring to an arbitrary district, we will simply write $\sigma(\cdot, D)$.} 
The matchings $\sigma(\cdot, D)$ and $\sigma(\cdot, \Omega)$ denote the assignment of students to schools before and after consolidation occurs, respectively.
A matching $\mu$ is {stable} if $\nexists (t,s) \in T^P\times S^P$ such that i) $\mu(t)=t$ and $\left\vert \mu(s) \right\vert < q_s$, or ii) $s \succ_t \mu(t)$ and $t \rhd_s t' \in \mu(s)$. 
The matching $\mu_\sosm$ is the student-optimal stable matching if it is stable and all students weakly prefer over any other stable matching.

We are interested in comparing the number of students who benefit after consolidation occurs against those who become worse off. The sets $T^+(\sigma) \coloneqq \{t \in T:\sigma(t,\Omega) \succ_t \sigma(t,D)\}$ and $T^-(\sigma) \coloneqq \{t \in T:\sigma(t,D)\succ_t \sigma(t,\Omega)\}$ represent the students who benefit and lose from consolidation under the matching scheme $\sigma$. 
For some ESCP, we have that $T^-(\sigma)>T^+(\sigma)$ even when we choose the SOSM before and after consolidation, as in the following example.

\begin{example}
	\label{example1}
	Consider two balanced school districts $D_1$ and $D_2$, the first one with schools $s_1,s_2$ and students $t_1,t_2$, whereas the second one has school $s_3$ and student $s_3$. All schools have capacity one. The preferences and priorities appear below. The SOSM before consolidation occurs appears in squares, whereas the SOSM after consolidation appears in circles.
	\begin{center}
	\begin{tabular}{lll}
		\multirow{2}{*}{$D_1$} & $t_1: \ensquare{$s_2$}\succ \encircle{$s_1$} \succ s_3$  & $\quad s_1: \encircle{$t_1$} \rhd t_3 \rhd \ensquare{$t_2$}$ \\
		& $t_2: \ensquare{$s_1$} \succ \encircle{$s_2$} \succ s_3$  & $\quad s_2: \encircle{$t_2$} \rhd \ensquare{$t_1$} \rhd t_3$ \\
		\midrule
		$D_2$ & $t_3: s_1\succ \ensquare{\encircle{$s_3$}} \succ s_2$  & $\quad s_3: \ensquare{\encircle{$t_3$}} \rhd t_1 \rhd t_2$
	\end{tabular}
	\end{center}	
	Both students from district $D_1$ are harmed by district consolidation, whereas the one student from district $D_2$ retains her initial match. Hence, the number of losers is larger than the number of winners. 
\end{example}

Example \ref{example1} shows how consolidation can be bad for students, and can be generalised to show that, for any ESCP, we can partition the society into districts in such a way that every student is weakly worse off after district consolidation.\footnote{Whenever $\sigma_{\sosm} (t,\Omega)$ is Pareto optimal and different from $\sigma_{\sosm} (t,D)$, at least one student benefits from district consolidation. Conditions on schools' priorities such as acyclicity, guarantee that the SOSM is Pareto optimal for students \citep{ergin2002efficient}. Many other conditions on preferences and priorities have been proposed, see \cite{cantillon2022respecting} and references therein.}
But what can we expect on an average instance of an ESCP? We answer this question next using random markets.

\paragraph{Random Markets} Another way to analyse students' welfare changes is to quantify the gains from district consolidation in terms of the expected rank of their assigned school in random ESCPs, in which the schools' priorities and students' preferences are generated uniformly at random.\footnote{Ranks are a common welfare measure in matching markets \citep{wilson1972, knuth1976, pittel1989,knoblauch2009,che2019efficiency,ortega2023cost}.}
The {\it absolute rank} of a school $s$ in the preference order of a student $t$ (over all potential schools in the society) is defined by $\rk_t(s) \coloneqq \left\vert \{ s' \in S: s' \succcurlyeq_t s \} \right\vert$. Given a matching $\mu$, the {\it students' absolute average rank of schools} can be  defined by
\begin{eqnarray*}
	\rk_T (\mu) &\coloneqq& \frac{1}{\left\vert \overline{T} \right\vert} \sum_{t \in \overline{T}} \rk_t(\mu(t))
\end{eqnarray*}
where $\overline{T}$ is the set of students assigned to a school under matching $\mu$. Then, the welfare gains from consolidation for students of district $D_i$ are defined as $	\gamma_T (\sigma_\sosm) = \rk_T ( \sigma_\sosm(\cdot,D_i) ) - \rk_T ( \sigma_\sosm(\cdot,\Omega) )$. We can approximate the students' expected welfare gains from consolidation as function of $n_i$ and $k_i$, obtaining a set of interesting comparative statics as corollaries.

\begin{proposition}
	\label{thm:prop3}
	In a random ESCP, the expected welfare gains from consolidation for students $\gamma_T (\sigma_\sosm)$ can be approximated by 
$$
 \frac{N+K}{q} \left[\frac{\log \left(\frac{n_i+k_i}{k_i}\right) + (q-1) \log \log \left( \frac{n_i+k_i}{k_i} \right)}{n_i} - \frac{ \log\left(\frac{N+K}{K}\right) + (q-1) \log \log \left(\frac{N+K}{K}\right)}{N} \right]
$$
\end{proposition}

The above expression has three testable implications for empirical studies, namely:
\begin{corollary}
\label{thm:corollary1}

		 	The average gains from consolidation are positive for all districts.\footnote{Under mild restrictions on the parameters, detailed in  \ref{sec:Appendix_A}.}

\end{corollary}

\begin{corollary}
	\label{thm:corollary2}

		  Students from more under-demanded districts benefit less from consolidation.

\end{corollary}

\begin{corollary}
	\label{thm:corollary3}
	
 A smaller district size leads to larger gains from consolidation. 
\end{corollary}

The proof of Proposition \ref{thm:prop3} extends the computation of the average rank in random one-to-one to many-to-one markets using an extension of the coupon collector problem by \cite{newman1960double}, and is postponed to \ref{sec:Appendix_A}. However, we provide here some intuition for the comparative statics above. It is well-known that, in an unbalanced two-sided matching problem the agents in the short side choose whereas the agents in the large side get chosen, a phenomenon that increases as the imbalance between the two sides  of the market grows \citep{ashlagi2017}. Thus, if a local district is highly under-demanded (i.e. if $k_i$ is relatively large), students get assigned to highly ranked schools before consolidation, which makes the gains from consolidation smaller, and vice versa for  students who belong to an almost balanced district (i.e. if $k_i$ is close to zero). This explains our first comparative statics (Corollary 2). 

The second comparative statics (Corollary 3) is due to the relationship between relative and absolute rankings. In small districts, even if students are assigned to some of their preferred schools within their district, it is unlikely that those schools are in the top of their preference list. Thus, in small districts, there is large potential for welfare gains.

\paragraph{Correlated Preferences} We have assumed that preferences and priorities are drawn independently and uniformly, even though these tend to display particular correlation patterns in practice (which we document for Budapest in Section \ref{sec:results}). A reasonable question is whether our three corollaries extend to the case in which preferences are correlated. 
While we are unable to provide a theoretical answer, we address this issue by introducing two parameters, denoted by $\rho_t$ and $\rho_s$, reflecting the correlation on preferences and priorities, respectively (with $0\leq \rho_t, \rho_s \leq 1$ and 1 representing identical preferences). We then conduct a simulation exercise for a large number of correlation structures in \ref{sec:simcorrelation}. 
Correlation on students' preferences further increases the gains from integration (achieving a peak when $\rho_t=0.75$). We continue observing that students from smaller and more under-demanded districts benefit more from consolidation even with correlated preferences and priorities.

\section{Data}
\label{sec:data}
Hungary has a nationwide integrated school market in which every student can apply to any school in the entire country, and uses a centralised assignment mechanism to allocate students to schools. In this system, every student submits a rank-ordered list (ROL) of arbitrary length, ranking the school programmes that he would like to attend. Similarly, each school programme ranks all the students that applied to it according to several criteria such as grades, additional exams and entrance interviews. The specific weighting of these criteria is decided upon by each school but must comply with specific governmental regulations (e.g. the weight of the interview score must not be more than 25\%). School programmes submit a strict ranking of their more preferred students, whereas the remaining students are simply deemed unacceptable and are not ranked against each other. The assignment of students to secondary schools is conducted using the student-proposing deferred acceptance algorithm \citep{biro2008student}. 
This algorithm has been used since 2000 in a fully consolidated fashion.\footnote{See \cite{biro2012mip} for a detailed description of its implementation.} 

For our empirical analysis, we use data from the national centralised matching of students to secondary schools in Hungary, the so-called KIFIR dataset,\footnote{KIFIR stands for \emph{Középiskolai Felvételi Információs Rendszer}, which translates to ``Information System on Secondary School Entrance Exams''.} along with student-level data from the national assessment of basic competencies (NABC), both from the year 2015. Our data encompasses the universe of all students in Hungary who apply to a secondary school programme in 2015 (at an age of 14, with some exceptions). Each secondary school offers general or specialised study programmes with different quotas that are known ex-ante by students. \ref{sec:Appendix_Data} discusses in detail the original data sources. Due to data protection arrangements, access to these data was restricted and our estimation routines were run by officials at the Hungarian Ministry of Education on their local computer.

We restrict our attention to the greater Budapest area which comprises 23 well-defined districts, so as to obtain a realistic setting within which the (un)consolidation of school districts can be studied. Budapest lends itself to this type of analysis because it is a geographically relatively small market that is tightly integrated, and yet the market is large enough to permit a meaningful study of the decomposition of a unified admission system into smaller and well-defined districts. Figure \ref{fig:districtflow} shows the geographical area of Budapest with school district borders, and with arrows between districts that send their students to study to other districts. Figure \ref{fig:districtflow} also shows that there is a considerable amount of inter-district movements, especially in the inner parts of the city.
\begin{figure}[h!]
	\begin{center}
		\includegraphics[clip, trim=0cm 2.5cm 0cm 2.5cm, width=\textwidth]{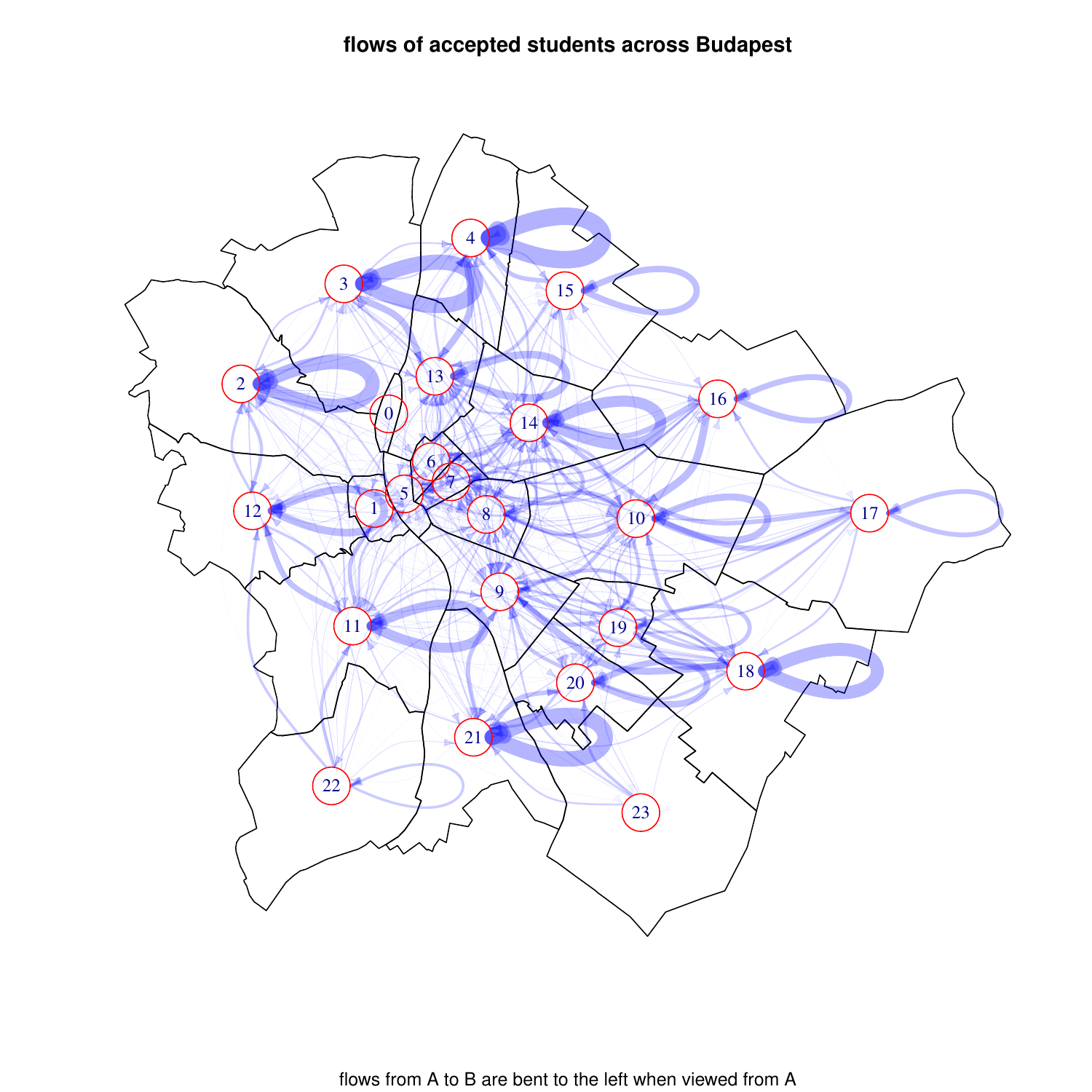}
	\end{center}
	\caption[Flows of accepted students in Budapest.]{Flows of accepted students between school districts in Budapest. Flows from one district $D$ to another district $D'$ are bent to the left when viewed from $D'$. The width of the arrows from $D$ to $D'$ is proportional to the number of students who live in $D$ and who were accepted at a school in $D'$.}
	\label{fig:districtflow}
\end{figure}

We link the application records in the KIFIR database to the corresponding information in the NABC dataset for 10,880 students who applied for a secondary school place in Budapest in 2015.
In order to attain comparable competitive conditions, we adjust the schools' capacities by removing any seats that were assigned to students not in our sample, i.e. outside Budapest. In total, there are 881 school programmes of 246 schools in our sample. A school programme sometimes contains several particular classes in which students specialize on languages or computer science, for instance. Thus, schools can offer multiple programmes within the same age cohort. We aggregate school programmes at the school level in order to reduce the sample size and the associated computational burden, which is not negligible in our context.\footnote{We converted students' ROLs  to the school level by keeping the most preferred school programme of every school. Combining the 246 schools with 10,880 students still leaves us with almost 2.7 million possible student-school combinations to be considered.} 
We focus on three school types -- four-year grammar, vocational secondary, and vocational schools -- to which students apply after having completed eight years of primary education. For all students in the sample, their location of residence is approximated by their zip code, and the Open Source Routing Machine \citep{OSRM} was used to compute travel distances from each of Hungary's zip code centroids to every known school location.
\begin{table}[h!]
	\centering
	\caption{Secondary school applicants in Budapest: Summary statistics.} 
		\label{tab:students_summary}
	\begin{footnotesize}
	\begin{threeparttable}
	\begin{tabular}{lrrrrr} 
		\toprule
		Variable&  \multicolumn{1}{c}{Mean} &\multicolumn{1}{c}{Median} & \multicolumn{1}{c}{SD} & \multicolumn{1}{c}{Min} &\multicolumn{1}{c}{Max}\\ 
		\midrule
		 \multicolumn{6}{l}{\it Panel A. Student characteristics $(N=10,880)$}\\
		birth year & 2,000.1 & 2,000 & 0.550 & 1,996 & 2,002 \\ 
		female & 0.50 & 0 & 0.50 & 0 &  1 \\ 
		grade average & 4.06 & 4.20 & 0.69 & 1.00 & 5.00\\ 
		math score (NABC)\textsuperscript{*} & 0.00 & 0.03 & 1.00 & $-$3.83 & 3.52 \\ 
		hungarian score (NABC)\textsuperscript{*} & 0.00 & 0.05 & 1.00 & $-$4.19 & 3.18 \\ 
		ability\textsuperscript{\textdagger} & 1.47 & 1.55 & 1.40 & -3.66 & 6.01 \\
		SES score\textsuperscript{*}& 0.00 & 0.12 & 1.00 & $-$4.11 & 1.65\\ 
		ROL length & 4.09 & 4 & 1.80 & 1 & 24 \\ 
		applies to home district & 0.68 & 1 & 0.47 & 0 & 1 \\ 
		ROL length within home district & 1.05 & 1 & 0.97 & 0 & 7   \\ 
		&&&&&\\	
		
		 \multicolumn{6}{l}{\it Panel B. Attributes of first-choice school ($N=10,880$)}\\
		distance (km) & 7.10 & 6.35 & 4.63 & 0.11 & 36.65 \\ 
		ave. math score (enrolled students) & 0.32 & 0.37 & 0.72 & $-$1.97 & 1.75  \\ 
		ave. hungarian score (enrolled students) & 0.35 & 0.41 & 0.70 & $-$2.01& 1.69 \\ 
		ave. SES score (enrolled students) & 0.09 & 0.10 & 0.58 & $-$1.89 & 1.21 \\ 
		&&&&&\\	
		
		 \multicolumn{6}{l}{\it		\it Panel C. Attributes of assigned school ($N=9,783$)}\\
		match rank & 1.48 & 1 & 0.92 & 1.00 & 11.00 \\ 
		matched to first choice & 0.71 & 1 & 0.45 & 0.00& 1.00  \\ 
		distance (km) & 7.06 & 6.22 & 4.65 & 0.11 & 36.65\\ 
		assigned to home district & 0.30 & 0 & 0.46 & 0.00 & 1.00 \\ 
		ave. math score (enrolled students) & 0.20 & 0.25 & 0.69 & $-$1.97 & 1.75 \\ 
		ave. hungarian score (enrolled students) & 0.23 & 0.24 & 0.67 & $-$2.01 & 1.69 \\ 
		ave. SES score (enrolled students) & $-$0.01 & -0.04 & 0.57 & $-$1.89 & 1.21 \\ 
		\bottomrule
	\end{tabular} 	
	\begin{tablenotes}
		\tiny \item Variables indicated with an asterisk are z-normalized. The 2015 Hungarian and math test scores are taken by the students as part of the admissions process. \textsuperscript{\textdagger} ability is the first principal component of the joint distribution of students' grades, their math, and their hungarian scores. Socioeconomic status is a composite measure which includes, amongst other variables, the number of books that the household has, or the level of parental education.
	\end{tablenotes}
\end{threeparttable}

	\end{footnotesize}
\end{table}

Table \ref{tab:students_summary} shows student-level summary statistics of our data. Panel A shows that most students were born in 2000, and that there are as many girls as boys, as one would expect. The students' mean grade average in the previous school year is four (on a scale from one to five, where five is the highest grade in the Hungarian grading system). Their math, Hungarian, and SES scores from the NABC\footnote{Where these scores were missing in our data, we imputed the missing values using predictive mean matching, as implemented in the package \texttt{mice} in R \citep{MICE}; see \ref{sec:Appendix_Data}.} were standardised because their absolute numbers have no meaning. The variable measuring students' socio-economic status (SES) is a composite measure that includes, amongst other variables, the number of books that the household has, or the level of parental education. This indicator was also standardized.\footnote{The SES measure is the one used by the Hungarian Ministry of Education for their reporting and is slightly different from the OECD measure used in \cite{horn2013diverging}, \cite{shorrer2023dominated} and \cite{shorrer2017obvious}.} As the students' grade average, their math, and their Hungarian NABC scores are highly correlated, we created a composite measure that we call ``ability'' and which is constructed as the first principal component of these variables. Table \ref{tab:students_summary} shows that the students from Budapest in our sample file applications to about four schools, on average.\footnote{Actually, students apply for course programmes, many of which may be offered by the same school. Thus, the actual length of the students' ROLs is larger than this.} Roughly seventy percent of the students apply to at least one school in their home district, and on average, students include only one school from their home district in their submitted rank order list.

Panel B shows some attributes of students' first choice school, and panel C describes attributes of the students' actual assigned school. Panel C shows that the average match rank is 1.48 (with 1 being the most preferred school), with more than seventy percent of all students being assigned to their top choices.\footnote{Because we aggregated programmes at the school level, the average rank of programmes may be larger than 1.48. For example, if a student applies to two programmes within the first choice school and is assigned to her second best programme, this would be recorded in our data as being assigned to her first choice. The aggregation of programmes at the school level is done for practical reasons, see footnote 15.} This is partially due to the fact that there is excess capacity: the schools in the sample have more seats than there are students.\footnote{The schools' excess capacity has been confirmed in conversation with officials from the Hungarian ministry of education on several occasions. In subsection \ref{sec:Appendix_D}, we replicate our analysis perform several robustness tests reducing schools' capacities under different counterfactuals. The students' welfare gains remain positive and large in four out of five specifications, becoming zero if the demand and supply for school seats are exactly balanced in each district.} The distribution of the number of programmes the students apply to, and of the actual match rank in the 2015 matching round, are shown in Figure \ref{fig:kifir2015}. This diagram confirms that most students submit rather short ROLs, and the majority of students are assigned to their submitted top choice.
\begin{figure}[h!]
	\begin{subfigure}[c]{0.49\textwidth}
		\includegraphics[width=\textwidth]{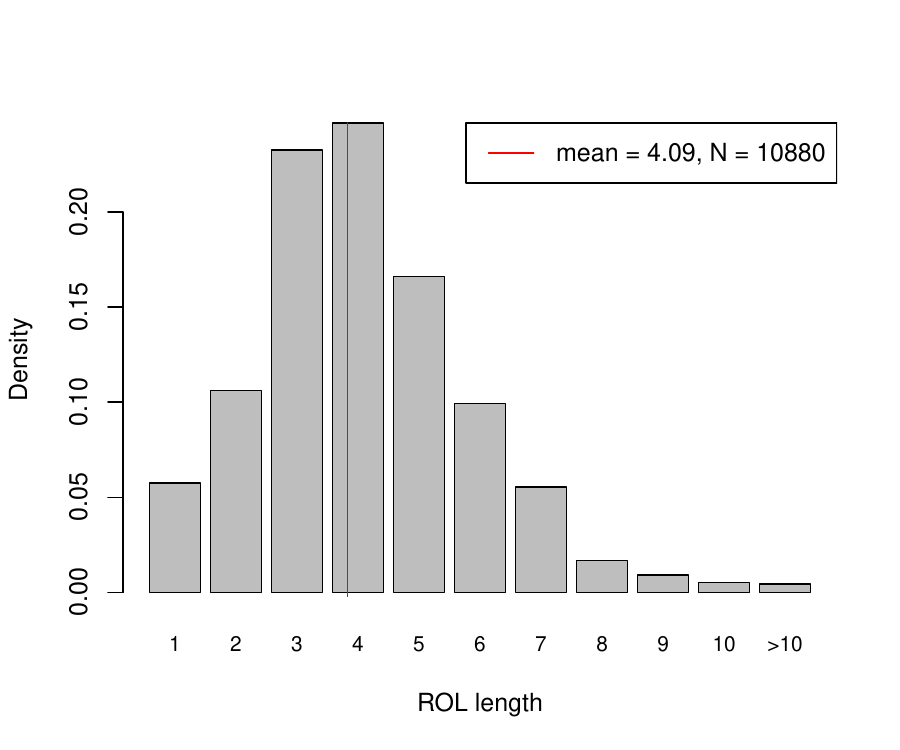}
	\end{subfigure}
	\begin{subfigure}[c]{0.49\textwidth}
		\includegraphics[width=\textwidth]{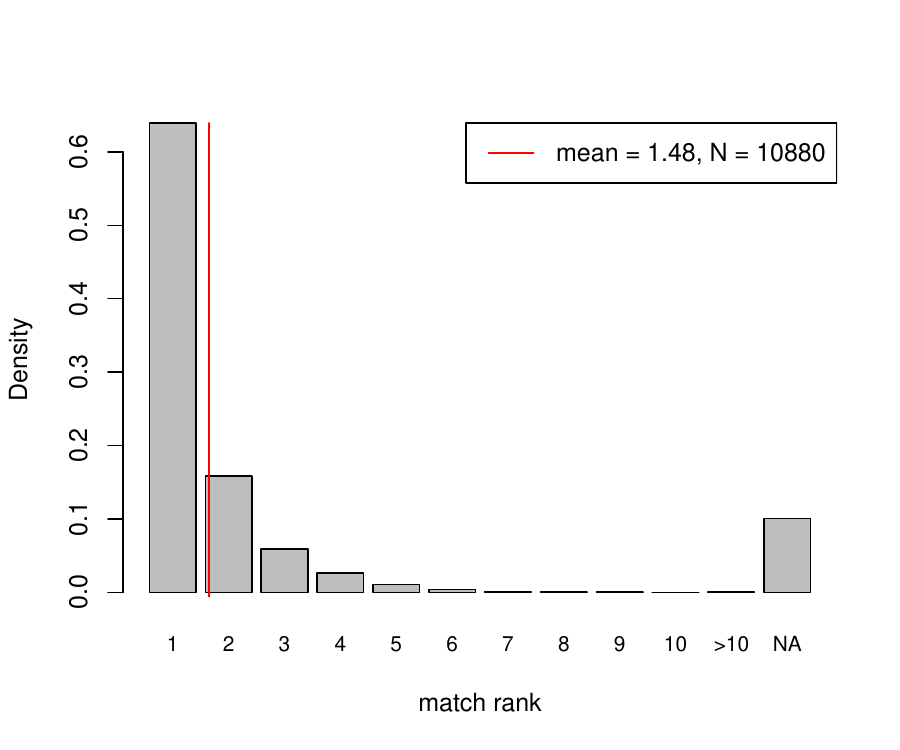}
	\end{subfigure}
	\caption{Distribution of the length of students' ROLs and of their realised match rank.}
	\label{fig:kifir2015}
\end{figure}

Table \ref{tab:schools_summary} shows the school-level summary statistics. School programmes in Budapest are very attractive so that many students from outside Budapest rank a school in Budapest as their top choice. Therefore, students from Budapest face stiff competition in their ``domestic'' school  market, and restricting the attention to students from Budapest will likely lead to a more relaxed assignment problem. In order to circumvent this problem, we subtracted the number of admitted students from outside Budapest from the schools' capacity so as to maintain the original ``tightness'' of the market -- this is the adjusted capacity that is used throughout our analysis. 
\begin{table}[h!]
	\centering
	\caption{Secondary schools in Budapest: Summary statistics ($N=246$).}
		\label{tab:schools_summary}
	\begin{footnotesize}
\begin{tabular}{@{\extracolsep{5pt}}lccccc} 
\toprule

	Variable &  \multicolumn{1}{c}{Mean} & \multicolumn{1}{c}{Median} & \multicolumn{1}{c}{SD} & \multicolumn{1}{c}{Min} & \multicolumn{1}{c}{Max}\\ 
\midrule
	capacity  & 137.10 & 120 & 96.31 & 6 & 502 \\ 
	adj. capacity & 116.45 & 94.5 & 90.57 & 6 & 498 \\ 
	applications & 411.20 & 229.5 & 456.93 & 7 & 2,392 \\ 
	ROL1 applications & 41.18 & 29 & 41.35 & 0 & 234 \\ 
	acceptable applications  & 126.29 & 89 & 119.20 & 0 & 654 \\ 
	assigned students & 39.77 & 34 & 31.50 & 0 & 157 \\ 
	ave. match rank  & 47.23 & 42.38 & 34.01 & 2.25 & 187.30 \\ 
	math (enrolled students)  & $-$0.15 & $-$0.12 & 0.79 & $-$2.01&  1.75 \\ 
	hungarian (enrolled students)  & $-$0.11 & $-$0.10 & 0.75 & $-$2.03 & 1.67 \\ 
	SES (enrolled students)  & $-$0.19 & $-$0.19 & 0.65 & $-$1.91 & 1.22 \\ 
	math (assigned students)  & $-$0.28 & $-$0.32 & 0.68 & $-$2.40 & 1.66 \\ 
	hungarian (assigned students) & $-$0.29 & $-$0.40 & 0.70 & $-$2.15 & 1.45 \\ 
	SES (assigned students)  & $-$0.13 & $-$0.19 & 0.63 & $-$1.66 & 1.52 \\ 
\bottomrule
\end{tabular} 

	\end{footnotesize}
\end{table}

The average school receives over four hundred applications, of which only 130 are deemed acceptable. About forty students are assigned to each school on average. The relatively small number of acceptable applications could indicate that it is quite costly for schools to rank all their applicants consistently, and so they focus on only ranking those students which are most likely to be admitted to the school. Our estimation approach assumes that schools submit their priority lists truthfully, i.e. that every student who is labelled ``unacceptable'' is really less preferred than any other applicant that is actually ranked by the school. This assumption could be violated if schools strategically choose to omit high achieving students who they anticipate will attend a different school. However, we think that this is probably a minor problem and schools are overall truth-telling.

We also collected data on whether a school holds an additional entrance interview, and we found that about forty percent of all schools do so.\footnote{This information was manually collected from the website \hyperlink{https://felvizsga.eu/felvi.php}{https://felvizsga.eu/felvi.php} which provides information about admission procedures at different Hungarian schools. Last accessed on 11 November 2019.} Table \ref{tab:schools_summary} also summarises the school-level averages of admitted and currently enrolled students. The standard deviation of these school-level averages is more than two-thirds of the total variance across students, which is normalised to one. Thus, there is evidence of a substantial amount of sorting by ability and socio-economic status.

\section{Empirical strategy}
\label{sec:strategy}
Our empirical strategy to estimate the gains from district consolidation is to compute the SOSM in an unconsolidated school market and compare it to the SOSM in the consolidated school market. 
In a first pass, we use the submitted ROLs to obtain an \emph{ad hoc} measure of the consolidation gains. This approach, deferred to \ref{sec:reportedpref}, has an important shortcoming:
thirty percent of all students have not included any school from their home district in their submitted ROLs,  and as a result, they remain unmatched in a counter-factual, disintegrated school market.

To circumvent this issue, we estimate the complete ROLs of all market participants. This allows us to compute a more complete SOSM in the unconsolidated market, and also to compare utility outcomes, rather than ordinal ranks. Figure \ref{fig:empirical_strategy} summarises our empirical strategy. In the end, as we describe in Section \ref{sec:results}, both approaches lead to very similar conclusions regarding the effect of district consolidation, but estimating the preferences allows us to understand such gains in terms of students' utilities.

\begin{figure}[h!]
	\center
	\includegraphics[width=0.7\textwidth]{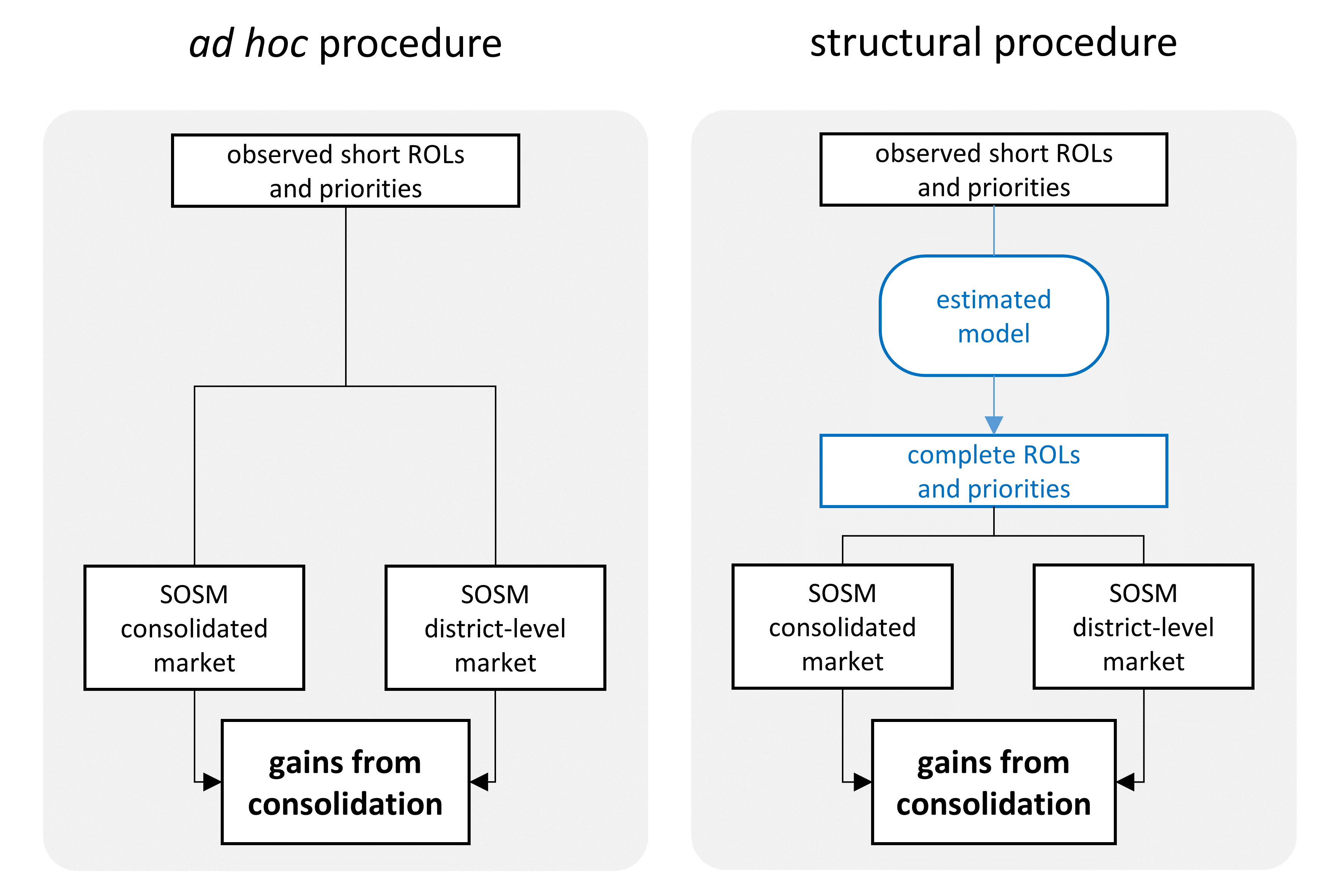}
	\caption{An overview of our empirical strategy (the {\it ad hoc} procedure is presented in \ref{sec:reportedpref}).}
	\label{fig:empirical_strategy}
\end{figure}

\subsection{Estimation methodology}
\label{sec:strategy:prefest}
We observe a school choice market with a set of students ($T$) and a set of schools ($S$). We write students' utilities over the set of schools $U_t(s)$, and schools' valuations over the set of students $V_s(t)$ as\begin{eqnarray}
	U_{t}(s) &=& U_{t0} + \mathbf{X}_{ts}\beta + \epsilon_{ts} \label{Eqn:StudentPref} \\
	V_{s}(t) &=& V_{s0} + \mathbf{W}_{st}\gamma + \eta_{st} \label{Eqn:SchoolPref}
\end{eqnarray}

where $\mathbf{X}_{ts}$ and $\mathbf{W}_{st}$ are observed characteristics that are specific to the school-student match $st$. $\mathbf{X}_{ts}$ could, for instance, include a school fixed effect or the travel distance from $t$ to $s$. The terms $U_{t0}$ and $V_{s0}$ are the outside utilities of not being matched to any student or school. These are assumed to be zero, so that the latent utilities represent the net utility of being matched. The match valuations $U_t(s)$ and $V_s(t)$ are treated as latent variables that are to be estimated along with the structural parameters $\beta$ and $\gamma$. 

We denote by $\mathbf{U}_t$ the vector of student $t$'s utilities over the entire set of schools, and by $\mathbf{V}_s$ school $s$'s valuations over the entire set of students. 
The elements of vector $\mathbf{Z}$ that do not refer to the student-school pair $ts$ are denoted by $\mathbf{Z}_{-ts}$, i.e. $\mathbf{U}_{-ts}$ denotes the entire set of utility numbers but for $U_t(s)$.
We further assume that the structural error terms $\epsilon_{ts}$ and $\eta_{st}$ are independent across alternatives, and normally distributed with unit variance.

We observe students' submitted ROLs over schools, $\mathbf{rk}$, and schools' submitted partial priority orderings over students, $\mathbf{pr}$. 
We denote the observed ROLs of student $t$ as $L_t = (s_t^1, s_t^2, \ldots, s_t^{K_t})$, where $s_t^k\in \mathcal{S}$ is some school. Denote the rank that student $t$ assigns to school $s$ as $rk_t(s)$, with $1 \leq rk_t(s) \leq K_t$ if $s \in L_t$ and $rk_t(s) = \emptyset$ else. The observed ROLs $\mathbf{rk}$ encompass all individually observed rankings $rk_t(s)$. Similarly, denote the set of students who apply to school $s$ as $L_s$, and let the priority number that school $s$ assigns to student $t$ be $pr_s(t)$. Priority numbers are like ranks, in that they take discrete values, and a lower priority number means higher priority. Schools are required to prioritise all students who apply to them, but they may rank some students as unacceptable. We say that $pr_s(t)=+\infty$ if student $t$ is unacceptable to school $s$, and $pr_s(t)=\emptyset$ if student $t$ did not apply at school $s$. Thus, $pr_s(t) \in \{1,2,\ldots,|L_s|,\infty,\emptyset\}$.

Given the specification of the error terms and the observed rankings, equations \eqref{Eqn:StudentPref} and  \eqref{Eqn:SchoolPref} can be regarded as representing two distinct rank-ordered probit models \citep[p.181]{Train2009}. However, the complications outlined in the introductory part of this section imply that an estimation as such is unlikely to obtain the true preference parameters. Because schools only rank students who apply to them, and geographical distance is not an admission criterion, we cannot follow the approach of \cite{Burgess2015} to construct the feasible choice set of each student in order to identify her true preferences. For the same reason, the construction of the stability-based estimator that is proposed in \cite{fack2019beyond} cannot be applied. Still, we follow their idea in that we use a combination of identifying assumptions to identify the model parameters. These are described in turn. 

We chose a Bayesian data augmentation approach, owing to its flexibility and because it allows us to directly estimate the latent variables $\mathbf{U}$ and $\mathbf{V}$ which are our prime objects of interest.\footnote{Similar approaches have been used by \cite{Logan2008} and \cite{Menzel2013} in one-to-one matching markets.} Following \citet{Lancaster2004}, we simulate draws from the posterior density of the structural preference parameters $p(\beta,\gamma|data)$ by considering the component conditionals $p(\mathbf{U}|\beta,\gamma, \mathbf{V}, data)$, $p(\mathbf{V}|\beta, \gamma, \mathbf{U}, data)$, $p(\beta|\gamma,\mathbf{U},\mathbf{V},data)$ and $p(\gamma|\beta,\mathbf{U},\mathbf{V},data)$. We assume a flat prior for the structural preference parameters $\gamma$ and $\beta$.\footnote{Details of the conditional posterior distributions are spelled out in \ref{sec:posterior}.}
Our \textit{data} comprises the co-variates $\mathbf{X}$ and $\mathbf{W}$, of the assignment $\mu$ and of the submitted rank order and priority lists. The Gibbs algorithm to sample from the posterior density repeats the following steps $N$ times:
\begin{enumerate}
	\item $\forall\, t,s$: draw $U_t(s)$ from $p(U_t(s)|\beta,\gamma, \mathbf{U}_{-ts},\mathbf{V}, data) = N(\mathbf{X}_{is} \beta,1)$, truncated to $[\underline{U}_{t}(s),\overline{U}_t(s)]$
	\item $\forall\, s,t$: draw $V_s(t)$ from $p(V_s(t)|\beta,\gamma, \mathbf{V}_{-st},\mathbf{U}, data) = N(\mathbf{W}_{st} \gamma,1),$ truncated to $[\underline{V}_s(t),\overline{V}_s(t)]$
	\item draw $\beta$ from $p(\beta|\gamma,\mathbf{U},\mathbf{V},data) = N\left(b, (\mathbf{X}'\mathbf{X})^{-1}\right)$, with $b=(\mathbf{X}'\mathbf{X})^{-1} \mathbf{X}' \mathbf{U}$
	\item draw $\gamma$ from $p(\gamma|\beta,\mathbf{U},\mathbf{V},data) = N\left(g, (\mathbf{W}'\mathbf{W})^{-1}\right)$, with $g=(\mathbf{W}'\mathbf{W})^{-1} \mathbf{W}' \mathbf{V}$
\end{enumerate}
Key to our estimation methodology are the truncation intervals for $U_t(s)$ and $V_s(t)$. These intervals are functions of the data and the latent variables in the model, and they are specific to the particular identifying restrictions used. We describe possible identifying restrictions below, and outline how they can be used to construct these truncation intervals; a detailed derivation of the truncation intervals is deferred to \ref{sec:Appendix_B}.

\subsection{Identification restrictions}
\paragraph{Weak truth-telling (WTT)}
Weak truth-telling requires that the student truthfully submits her top $k_t$ choices, and that any unranked alternative is valued less than any ranked alternative. Formally, this implies that $U_i(s) \geq U_i(s')$ if (but not only if) $rk_t(s) < rk_t(s')$ or $s' \notin L_t$. A similar reasoning can be applied to schools' priorities over students, with the difference that a school $s$ cannot rank a student $t$ unless $t$ applies to $s$. However, a school can label a student as ``unacceptable'' which implies that all students labelled in this manner are valued less than any other ranked student. So we can bound $V_s(t) \geq V_s(t')$ if $s \in L_t \cap L_{t'}$ and $pr_s(t) < pr_s(t')$ or $pr_s(t')=+\infty$. Taken together, these bounds pin down the truncation intervals and the component conditionals in steps 1 and 2 above. WTT is at odds with students strategically omitting top schools in their ROLs.

\paragraph{Undominated strategies (UNDOM)}
The assumption of undominated strategies is similar to that of WTT, but is restricted to the submitted rank order lists, i.e. there is no inference between a ranked and an unranked school. Thus we can bound $U_t(s) \geq U_t(s')$ if $s,s' \in L_t$ and $rk_t(s) < rk_t(s')$. The bounds for the school's valuation over students are the same as in WTT because a school cannot decide not to rank a student; it must at least decide whether the student is acceptable or not. Therefore, UNDOM is a weaker condition than WTT.

\paragraph{Stability (STAB)}
If we assume that the matching of students to schools is \emph{stable} in the sense outlined in Section \ref{sec:model}, a different set of bounds can be applied to the latent valuations. Denote the observed matching as $\mu$ such that $\mu(t)=s$ and $t\in\mu(s)$ if student $t$ is assigned to school $s$. Stability implies that there is no pair of a student $t$ and a school $s$ such that $V_s(t) > \min_{t' \in \mu(s)} V_s(t')$ and $U_t(s) > U_t(\mu(t))$. These conditions imply that we can bound the realization of $U_t(s)$ conditional on the matching $\mu$, and on the match valuations $\mathbf{U}_{-ts}$ and $\mathbf{V}_{-ts}$. Analogous bounds can  be placed on $V_s(t)$ with straightforward extensions for cases where schools are not operating at full capacity. These bounds are spelt out in \ref{sec:Appendix_B} in greater detail.

\paragraph{STAB + UNDOM}
The combination of undominated strategies and stability is our preferred set of identification assumptions to estimate students' and schools' preference parameters. By assuming that the resulting assignment is stable and that students submit undominated ROLs, the bounds for the latent valuations can be tightened, imposing whichever bound is stricter. One may be concerned that the conditions imposed by both assumptions may be incompatible, but as we show in Lemma \ref{thm:lemma1} in \ref{sec:Appendix_B}, the resulting set of latent valuations is non-empty under reasonable conditions. 

The assumption that applicants use undominated strategies allows for partial parameter identification through the inequality bounds \citep{fack2019beyond}. To see the benefit of combining stability with UNDOM, consider the following example. With stability, the preference order over two schools is mainly estimated using the subset of students who are assigned to one of these schools \textit{and} have both schools in their stable choice set (i.e., cleared the latent admissions cutoffs at both schools). UNDOM instead uses the set of all students who ranked both schools, including students who were not assigned there or did not clear the cutoffs. Thus, while STAB uses stability conditions to infer the preference order over ranked and unranked choices, UNDOM uses additional identifying information for ranked choices. 

The results of an extensive Monte-Carlo study, outlined in \ref{sec:appendix:MCsim}, shows that this combination of identification restrictions allows for point identification of preference and priority parameters and is robust to the strategic submission of preference lists.

\section{Empirical results}
\label{sec:results}

\subsection{Preference estimation results}
\label{sec:results:prefs}
 
We now turn to the key building block of our structural approach to computing the gains from consolidation. We assume that students' preferences over schools depend on the geographical distance and on the squared distance between a student's place of residence and the schools' location. To proxy for a school's academic quality, we computed the average of the mean NABC scores in math and Hungarian of students currently enrolled at that school. Also, we computed the average SES score of those students. Finally, we included the interaction terms of the students' math, Hungarian, and socio-economic scores with their respective school-level means in order to test whether there is evidence for assortative matching.
 
To account for any unobserved heterogeneity across schools, we include school dummies, as we have a small set of observable school-level characteristics.\footnote{Because we are essentially estimating a discrete choice model over the set of schools, the preference specification cannot include an intercept, as this would not be identified. For the same reason, the first school dummy was omitted lest an intercept is introduced by means of a linear combination of school dummies. In the empirical specification, it turned out that some multicollinearity problems arose even when excluding one school dummy, possibly due to numerical inaccuracies or the presence of interactions. Thus, some more school dummies had to be excluded. To this end, we chose the following approach: In a first step, all fixed effects for schools numbered 2 through to 246 were used to generate a design matrix $\mathbf{X}$ for the problem at hand. In step $k$, we checked whether the matrix $\mathbf{X}'\mathbf{X}$ had full rank. If not, we dropped one school fixed effect and continued with step $k+1$, else we stopped. This procedure resulted in a set of fixed effects for the schools numbered 2 through to 243. \label{fn:multicoll}} 
We assume that schools select their students based on their gender, math and Hungarian NABC scores, and the SES score.  
We estimated a separate set of coefficients for each tier of the Hungarian school system.
Our Gibbs sampler was initialized with zero values for all parameters and valuations. Because the estimation procedure is rather time consuming, we let it run for only ten thousand iterations and discarded the first five thousand iterations. To reduce the serial correlation, only every tenths estimate of the remaining five thousand iterations was used so that the posterior means are averaged across five hundred iterations. We confirmed that the coefficient estimates had converged to their stationary posterior distribution after about two thousand iterations. 

The posterior means of the parameter estimates for two different identifying assumptions (WTT and STAB + UNDOM) are shown in Table \ref{tab:prefest} below and will be discussed in turn. Notice that our Bayesian estimation approach allows us to directly sample from the posterior parameter distribution, so that we do not need to rely on asymptotic results as in conventional estimation approaches. Therefore, Table \ref{tab:prefest} does not include asymptotic p-values but instead shows the 95\% confidence intervals of the posterior distribution.
\begin{table}[h!]
	\centering
	\caption{Posterior parameter means under two different identifying assumptions.}
	\label{tab:prefest}
	\begin{small}
		\begin{threeparttable}
			\begin{tabular}{lrrrr}
				\toprule
				{Student's selection of schools}      & \multicolumn{2}{c}{STAB + UNDOM} & \multicolumn{2}{c}{WTT} \\
				& $\bar{\beta}$ & \textit{95\% CI} & $\bar{\beta}$ & \textit{95\% CI} \\
				\midrule
				distance (km) & -0.15 & [-0.15;-0.14] & -0.34 & [-0.34;-0.34] \\
				distance (km\textsuperscript{2}) &  0.00 & [ 0.00;  0.00] &  0.01 & [ 0.01; 0.01] \\
				academic quality &  0.75 & [ 0.68;  0.82] &  1.49 & [ 1.46; 1.52] \\
				avg. SES &  1.52 & [ 1.41;  1.65] &  0.46 & [ 0.42; 0.51] \\
				oral entrance exam & -1.46 & [-1.70;-1.24] & -4.44 & [-4.59;-4.29] \\
				math $\times$ ave. math &  0.18 & [ 0.17;  0.20] &  0.19 & [ 0.18; 0.20] \\
				hungarian $\times$ ave. Hungarian &  0.22 & [ 0.21;  0.24] &  0.30 & [ 0.29; 0.32] \\
				SES $\times$ ave. SES &  0.29 & [ 0.30;  0.31] &  0.36 & [ 0.35; 0.37] \\
				\midrule
				{Schools' selection of students} & $\bar{\gamma}$ & \textit{95\% CI} & $\bar{\gamma}$ & \textit{95\% CI} \\
				\midrule
				\emph{gymnazium} &       &       &       &  \\
				~ female & -0.93 & [-0.95;-0.91] & -0.01 & [-0.04; 0.01] \\
				~ math score &  0.05 & [ 0.03;  0.07] &  0.19 & [ 0.17; 0.22] \\
				~ Hungarian score &  0.39 & [ 0.38;  0.41] &  0.22 & [ 0.20; 0.25] \\
				~ SES score &  0.04 & [ 0.02;  0.05] &  0.10 & [ 0.08; 0.12] \\
				\emph{secondary school} &       &       &       &  \\
				~ female & -0.44 & [-0.48;-0.40] &  0.12 & [ 0.09; 0.16] \\
				~ math score &  0.18 & [ 0.16;  0.21] &  0.24 & [ 0.21; 0.27] \\
				~ Hungarian score &  0.29 & [ 0.26;  0.32] &  0.21 & [ 0.20; 0.26] \\
				~ SES score &  0.05 & [ 0.03;  0.07] &  0.10 & [ 0.08; 0.12] \\
				\emph{vocational school} &       &       &       &  \\
				~ female &  0.09 & [ 0.04;  0.16] &  0.05 & [-0.03; 0.13] \\
				~ math score &  0.10 & [ 0.06;  0.14] &  0.08 & [ 0.03; 0.13] \\
				~ Hungarian score &  0.19 & [ 0.15;  0.23] &  0.14 & [ 0.09; 0.20] \\
				~ SES score &  0.01 & [-0.02;  0.04] &  0.02 & [-0.02; 0.05] \\
				\bottomrule
			\end{tabular}%
			\begin{tablenotes}
				\tiny \item Posterior means of preference and priority parameters under two different identifying assumptions. Fixed effects for schools numbered 2 through to 243 were included in students' preference equation and are not reported here. Confidence intervals from the posterior parameter distribution of the Gibbs sampler.
			\end{tablenotes}
		\end{threeparttable}
	\end{small}
\end{table}

First, consider the results of the college selection equation (top panel of Table \ref{tab:prefest}) across the two identifying assumptions. These results are qualitatively similar to each other: students dislike schools that are further away, but the marginal disutility of travelling is \emph{decreasing} because the squared distance term is positive. Students also value academic quality and prefer schools with a higher average SES score, but they dislike the presence of an oral entrance exam. The coefficient for the presence of an entrance exam is much smaller (\emph{i.e.} more negative) in the WTT result: this is an indication that students strategically omit highly competitive schools which hold an entrance exam, so that the WTT estimates of the oral interview are biased downwards, whereas our stability based estimator corrects for this bias. This result confirms the importance of correcting for strategic reporting when estimating students' preferences. 
The interaction terms are all positive, which suggests that there is sorting on both academic ability and on socioeconomic background. Both estimation approaches yield results that are qualitatively quite similar. The variance of the interacted variables is much larger than that of the school-level variables, so that the interaction terms' contribution towards explaining student preferences is actually quite large.

The results of the student selection equation (bottom panel of Table \ref{tab:prefest}) show that students' math and Hungarian scores are important variables that schools condition their choices on. 
Somewhat surprisingly, the female coefficient is negative in the STAB + UNDOM specification, whereas it is minimal in the WTT specification. The large negative estimated coefficients for the female indicator is due to the stability requirement: in the data, female students have higher Hungarian scores than male students.\footnote{See Table \ref{tab:gender_differences} in \ref{sec:Appendix_Data}.} At the same time, the Hungarian score is also a key determinant of the schools' priority decision. But in the aggregate, roughly as many female students as male students are admitted to each school, and so the negative female coefficient is needed to ensure that not too many female students form instabilities with school seats occupied by male students. Hence, we think that the negative female coefficient merely reflects the schools' desire to have a balanced gender composition, but it does not indicate discrimination of female students per se.

All schools except for vocational schools appear to select on the students' socioeconomic status although the coefficient is small compared to the Hungarian score. Yet, in combination with the students' taste for schools with a higher average socio-economic status, and the tendency of students with higher socio-economic backgrounds to prefer schools with a  higher average socio-economic status, these results may be indicative of social sorting patterns that could be interesting in their own right.

\paragraph{Constructing complete preference lists}
In order to obtain complete preference lists for the entire market, we use the estimated coefficients of the student and school selection equations as represented in Table \ref{tab:prefest} and combine them with one set of draws from the distribution of error terms that respect the upper and lower bounds derived from stability and imposed by submitted preference lists. Thus, the estimated utility for student $i$ attending school $s$ is 
\[
	\hat{U}_t(s) = \mathbf{X}_{ts} \bar{\beta} + \hat{\epsilon}_{ts}
\]
where $\hat{\epsilon}_{ts}$ is one particular realization of the latent error distribution such that $\hat{U}_t(s)$ respects the bounds that are imposed by the identifying assumptions. This estimated latent utility comes straight from the Gibbs sampler. Schools' latent match utilities are constructed analogously. These estimates of the latent valuations can then be used to construct, for each market participant, a complete preference ordering of the other market side. Note, however, that every such set of valuations is only one particular draw from an infinite manifold of possible realizations. We only use a single realization of the valuations, and believe that the large market size validates this approach.

\subsection{Gains from consolidation: using estimated preferences} 
\label{sec:disintegration:stated}

This subsection reports our estimates of the gains from consolidation. We only present the results obtained when using estimated preferences. These are very similar to those obtained with reported preferences, which we postpone to \ref{sec:reportedpref}. In summary, we find that 48\% and 73\% of matched students strictly benefit from district consolidation with reported and estimated preferences, respectively; whereas the corresponding share of students harmed by consolidation is 7\% and 2\%.

We compare the outcome of a consolidated city-wide match to the district-level matching scheme. 
In addition to the rank order gains, we computed the average gains in latent utility. For a student $t$, this is defined as the utility difference between attending the assigned school in the consolidated market, $\mu_{\Omega}(t)$, and the assigned school in the district level market, $\mu_D(t)$:
\[
	\Delta U_t \equiv \hat{U}_t(\mu_{\Omega}(t)) - \hat{U}_t(\mu_D(t))
\]
Utility is a unitless quantity which is hard to interpret per se, but our utility specification allows us to express these gains in terms of travel distances
\[
	\Delta U_t^{km} \approx \frac{\Delta U_t}{ \left|\frac{\partial \hat{U}_t(\mu_{\Omega}(t)) }{\partial d_{t \mu_{\Omega}(t)}}\right| },
\]
with $d_{t \mu_{\Omega}(t)}$ being the travel distance between student $t$'s zip code of residence and her assigned school in the consolidated market.\footnote{Because distance travelled enters the utility specification quadratically, it matters in principle whether the partial derivative is evaluated at the district level matching $\mu_D$, or at the integrated matching $\mu_{\Omega}$. However, the estimated quadratic term is very small (see Table \ref{tab:prefest}), which allows us to use the following approximation:
\[
	\left| \frac{\partial \hat{U}_t(\mu_{\Omega}(t)) }{\partial d_{t \mu_{\Omega}(t)}} \right| = \left| -0.148 + 2 (0.002 d_{t \mu_{\Omega}(t)}) \right| \approx 0.148.
\]
Hence, one utility unit is approximately worth seven kilometres of avoided travel distance.}
Because students dislike travel, we use the absolute value in the denominator, so that $\Delta U_t^{km}>0$ corresponds to a positive welfare gain due to market consolidation. Therefore, $\Delta U_t^{km}$ is a measure of the additional travel time that a student would be willing to incur in order to attend the school in the consolidated market, rather than the assigned school in the district level school market.\footnote{As before, we merge district 23 to its neighbouring district number 20. }

Table \ref{tab:fullROL:matchstats} shows summary statistics of the resulting district-level and consolidated, city-wide matchings. It shows that 894 students remain unmatched in the district-level matching. In the consolidated market, all students are matched because there is enough capacity in the aggregate and preference lists are complete.
\begin{table}[h!]
	\centering
	\caption{Matching statistics (with estimated preferences).}
	\label{tab:fullROL:matchstats}
	\begin{footnotesize}
		\begin{tabular}{lr} 
			\toprule
			\it Panel A. Unconsolidated matching&\\
			~~ matched students & 9,986 \\ 
			~~ share top choice match & 0.83 \\ 
			~~ avg. match distance [km] & 3.55 \\  
			\midrule
			\it Panel B. Consolidated matching&\\
			~~ matched students& 10,880 \\ 
			~~ share top choice match & 0.66 \\ 
			~~ share matched in home district & 0.29 \\ 
			~~ avg. match distance [km] & 7.14 \\  
			\bottomrule
		\end{tabular} 
	\end{footnotesize}
\end{table}

Table \ref{tab:district_integration_est} shows that the vast majority of students (67\% of all students, 73\% of matched students) strictly benefit from participating in a consolidated market. Only 2\% of the students are assigned to a less preferred school under the consolidated assignment. We observe that more students benefit from district consolidation than those that are harmed, and this observation holds for every single school district (we also observe this phenomenon when analysing gains from consolidation with reported preferences). 8\% of the students remain unassigned in the unconsolidated market because there are three school districts with more students than school seats.

The first row in Table \ref{tab:fullROL:gainstats} shows summary statistics of the consolidation gains $\Delta U_t$. Because not all students are matched in the unconsolidated market, those gains cannot be computed for all students. The average gains are positive, but some students also lose due to market consolidation. However, the median is positive so that the majority of all students gain. The second row of that table shows the utility gains, converted to distance units $\Delta U_t^{km}$. It shows that the average student's gains are equivalent to saving more than five kilometres in travel distance, even though students actually incur longer travel distances in the consolidated market, as Table \ref{tab:fullROL:matchstats} shows. Accordingly, the utility gains greatly outweigh the additional travel distances that are incurred in the consolidated market. The third row of the table presents utility gains converted into ranks. The rank measure is obtained by converting estimated latent utilities into ranks and computing rank gains for each student. With a mean of 12 and a median consolidation gain of three ranks, the gains on this measure are also large.

\begin{table}[h!]
	\centering
	\caption{Losers ($-$) and winners ($+$) from consolidation (with estimated preferences).}
	\label{tab:district_integration_est}
	\begin{footnotesize}

	\begin{threeparttable}
	\begin{tabular}{lrrrrrrr} 
		\toprule 
		District & Seats & Students & Excess Seats & $-$ & $0$ & $+$ & Unmatched \\ 
		\midrule
		1  & 338   & 95  & 243   & 3  & 9   & 83  & 0   \\
2  & 1,191 & 634 & 557   & 3  & 216 & 415 & 0   \\
3  & 928   & 743 & 185   & 3  & 253 & 487 & 0   \\
4  & 865   & 746 & 119   & 8  & 273 & 465 & 0   \\
5  & 625   & 217 & 408   & 2  & 43  & 172 & 0   \\
6  & 1,243 & 172 & 1,071 & 8  & 20  & 144 & 0   \\
7  & 1,312 & 212 & 1,100 & 3  & 81  & 128 & 0   \\
8  & 2,524 & 290 & 2,234 & 18 & 64  & 208 & 0   \\
9  & 2,116 & 275 & 1,841 & 24 & 71  & 180 & 0   \\
10 & 2,012 & 591 & 1,421 & 59 & 123 & 409 & 0   \\
11 & 1,025 & 713 & 312   & 8  & 153 & 552 & 0   \\
12 & 956   & 359 & 597   & 8  & 134 & 217 & 0   \\
13 & 3,290 & 449 & 2,841 & 4  & 163 & 282 & 0   \\
14 & 2,893 & 796 & 2,097 & 51 & 196 & 549 & 0   \\
15 & 701   & 454 & 247   & 0  & 92  & 362 & 0   \\
16 & 770   & 659 & 111   & 0  & 24  & 635 & 0   \\
17 & 147   & 628 & -481  & 0  & 6   & 141 & 481 \\
18 & 503   & 873 & -370  & 0  & 31  & 472 & 370 \\
19 & 773   & 444 & 329   & 33 & 70  & 341 & 0   \\
20 & 1,643 & 573 & 1,070 & 6  & 166 & 401 & 0   \\
21 & 2,518 & 641 & 1,877 & 0  & 265 & 376 & 0   \\
22 & 273   & 316 & -43   & 0  & 30  & 243 & 43 \\
\midrule
Total & 28,646 & 10,880 & 17,766 & 241 & 2,483 & 7,262&  894  \\ 
\bottomrule

		\end{tabular} 	
		\begin{tablenotes}
			\tiny \item Seats refers to number of seats after removing those given to students from outside Budapest. Excess seats refers to seats minus students. The symbols $-$, $0$ and $+$ denote the number of losers, indifferences and winners from consolidation, respectively. Data obtained using estimated preferences.
		\end{tablenotes}
	\end{threeparttable}

	\end{footnotesize}
\end{table}

\begin{table}[h!]
	\centering
		\caption{Measures of consolidation gains in latent utility changes.}
	\label{tab:fullROL:gainstats}
	\begin{footnotesize}
\begin{tabular}{lrrrrrr}
	\toprule
	& Mean  & SD    & Min   & Median & Max   & N \\
	\midrule
	\textit{total gains} &       &       &       &       &       &  \\
	~~ in latent utility units & 0.819 & 0.916 & -1.895 & 0.600 & 5.799 & 9,986 \\
	~~ in equivalent kilometres & 5.532 & 6.187 & -12.805 & 4.054 & 39.180 & 9,986 \\
	~~ in equivalent ranks &   12.250    &   24.782    &  -8.000     &   3.000    &   232.000    & 9,986  \\
	\bottomrule
\end{tabular}%

	\end{footnotesize}
\end{table}

Next, we examine the relationship between average consolidation gains, measured in latent utility units per district, and two key district characteristics: size and capacity. Figure \ref{fig:fullROL:Vcgains} illustrates a weak negative correlation between the average utility gains and district size, as well as a negative correlation between average gains and district-level excess capacity.
Corresponding empirical tests, based on regressions of district-level average gains on district characteristics, are presented in Table \ref{tab:fullROL:integration_gains_tests}.\footnote{In our analysis, we treat the dependent variables as observations, although they are estimates based on the first regression in Table \ref{tab:prefest} combined with draws from the distribution of error terms. This may introduce measurement error in the dependent variable, affecting the tests in Tables \ref{tab:fullROL:integration_gains_tests}, \ref{tab:explaining_gains1} and \ref{tab:explaining_gains}. If this error is uncorrelated with regressors, the OLS estimates yield consistent but less precise coefficient estimates \citep[as per][pp.\ 77-78]{wooldridge2010econometric}. The reported standard errors are thus conservative upper bounds.} Columns (1) and (2) reveal that both district size (measured in hundreds of students) and district capacity (measured by relative excess capacity) exert a significant negative partial effect on average latent utility and distance gains, respectively. Column (3) conducts a direct examination of the theoretical framework in Section \ref{sec:model}, assessing welfare gains through rank gains. Regressing the average rank gains on size and capacity offers a direct evaluation of the predictions in Corollaries \ref{thm:corollary2} and \ref{thm:corollary3}. The partial effects remain consistent with those found in Columns (1) and (2) in both direction and significance, providing empirical support for the theoretical predictions. 

\begin{figure}[h!]
	\begin{subfigure}[c]{0.49\textwidth}
		\begin{center}
			\includegraphics[width=\textwidth]{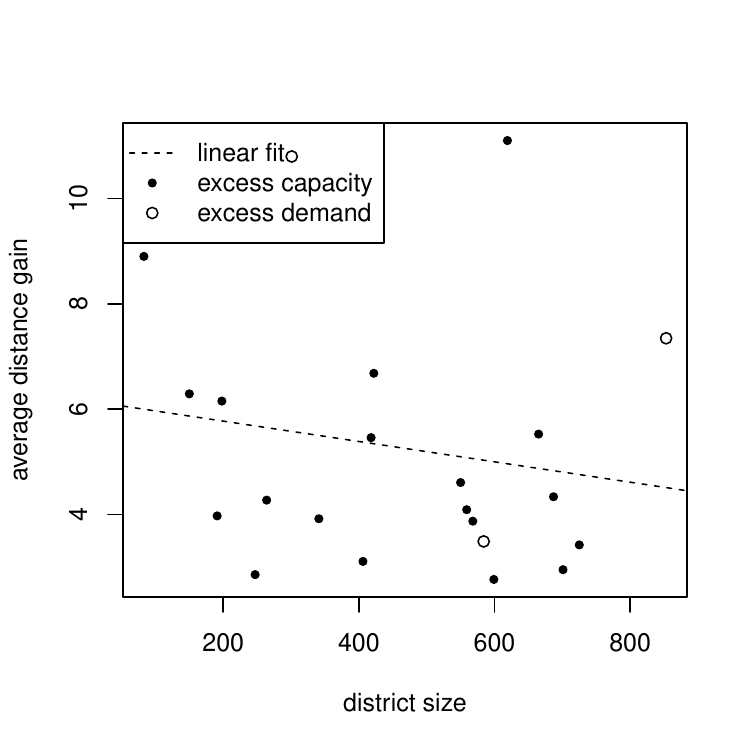}
		\end{center}
		\caption{Average distance gains and district sizes}
		\label{fig:fullROL:Vcgains:size}
	\end{subfigure}
	\begin{subfigure}[c]{0.49\textwidth}
		\begin{center}
			\includegraphics[width=\textwidth]{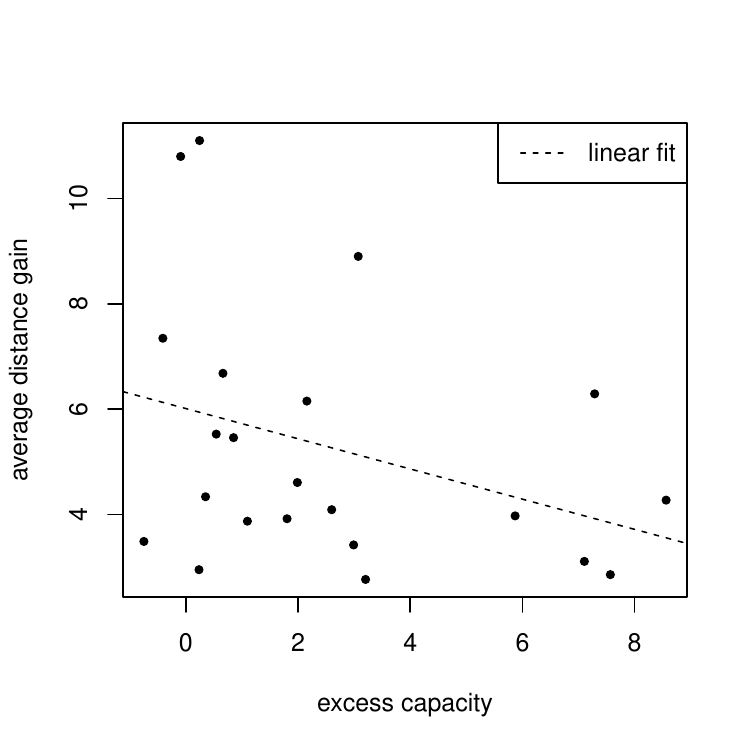}
		\end{center}
		\caption{Average distance gains and excess capacity}
		\label{fig:fullROL:Vcgains:capacity}
	\end{subfigure}
	\caption{Average latent utility gains converted to distance kilometres ($\Delta U_t^{km}$) of an integrated market in Budapest, using inferred complete preference lists. One observation denotes average statistics in one district.}
	\label{fig:fullROL:Vcgains}
\end{figure}

\begin{table}[h!]
	\centering
	\caption{Testing the relationship between consolidation gains and district statistics.}
	\label{tab:fullROL:integration_gains_tests}
	\begin{footnotesize}
		\begin{threeparttable}
			\begin{tabular}{lccc} 
				\toprule
				& \multicolumn{3}{c}{\textit{Dependent variable:}} \\ 
				\cmidrule{2-4} 
				& avg.\ latent utility gain & avg.\ distance gain & avg.\ rank gain \\ 
				\midrule
				district size  & $-$0.0018$^{***}$ & $-$0.01216$^{***}$ & $-$0.0238$^{**}$  \\ 
				& (0.0009)         & (0.0061)         & (0.0111)  \\ 
				& & & \\ 
				excess capacity & $-$0.1454$^{***}$ & $-$0.9824$^{***}$ & $-$2.5403$^{***}$  \\ 
				& (0.0670)         & (0.4527)          & (0.8340)  \\ 
				& & & \\ 
				\midrule
				Observations & 22 & 22 & 22 \\ 
				\bottomrule
			\end{tabular} 
			\begin{tablenotes}
				\tiny \item $^{*}$p$<$0.1; $^{**}$p$<$0.05; $^{***}$p$<$0.01. Standard errors in parentheses, intercept not shown. One observation denotes one district. Statistics obtained using inferred complete preferences.
			\end{tablenotes}
		\end{threeparttable}
	\end{footnotesize}
\end{table}

\subsection{Heterogeneity in utility gains} 
\label{subsec:heterogeneity}

In order to explain the large welfare gains we have documented, we perform two exercises. In a first exercise, we isolate the effects of choice and competition in a decomposition exercise. 
District consolidation has two effects on students' welfare: first, it leads to more choice, which is unambiguously good, and second, it may increase or decrease competition. Increased competition means that it becomes more difficult for a given student to be admitted to his or her favourite schools. Whether competition increases or decreases depends on many factors. If the schools in some sub-market are very attractive, or if this market is not as tight as the aggregate market (from the students' perspectives), then district consolidation will lead to more competition, so that domestic students may be harmed. 
\ref{subsec:compute_decomposition} presents in detail how to construct this decomposition. The results show that the choice effects account for the vast share of total welfare gains, whereas the average competition effects are much smaller in magnitude by a factor ranging between 7 and 22, dependent on the method used.

In a second exercise, we regress the student-level gains on student- and district-level observables. 
Table \ref{tab:explaining_gains1} 
shows the results of this linear regression analysis. The coefficients describe a ``consolidation premium'' that can be ascribed to various observable student characteristics. 
The first column of this table shows that students with a higher socio-economic status (SES) benefit \emph{relatively} more from district consolidation. The italicised adverb is important because students benefit on average, but some students benefit more than others. However, we cannot reject the null hypothesis of there being no effect at all. A similar, but exacerbated pattern can be observed for students with higher academic ability. High-ability students benefit more from district consolidation than average students. This effect is statistically significant.

Students in larger districts or those districts with a lot of excess capacity, benefit significantly less than other students. This is consistent with the predictions of Corollaries \ref{thm:corollary2} and \ref{thm:corollary3} and with the district-level findings reported in Table \ref{tab:fullROL:integration_gains_tests}. Consistent to what we find in the district-level analysis, 
the negative effect of district size on the consolidation gains is also estimated to be significantly different from zero in the student-level analysis.

The results imply that there is a consolidation premium for high-ability students. As Table \ref{tab:students_summary} shows, the explanatory variable SES is standardized and has unit variance, whereas the variance of ``ability'' is about $1.5$. Because the estimated coefficient in Table \ref{tab:explaining_gains1} is also larger for ``ability''  than for SES, it follows that an increase in student ability by one standard deviation increases the consolidation premium by about $1.5 \times 0.0966 \approx 0.15$ kilometres, whereas an increase of the socio-economic status indicator by one standard deviation increases the consolidation premium by only $0.06$ kilometres. So besides being insignificant, the estimated effect of a higher socio-economics status on consolidation gains is also much less relevant.
Thus, it appears that the highly selective consolidated Hungarian school system benefits high-ability students more than those from higher socio-economic background. 

\begin{table}[h!]
	\centering
	\caption[Explaining consolidation gains]{Explaining gains from consolidation with students observables.} 
\label{tab:explaining_gains1}
\begin{footnotesize}
	\begin{threeparttable}
		\begin{tabular}{lc}
			\toprule
			& \multicolumn{1}{c}{\textit{Dependent variable}: distance gains} \\
			\midrule
			\multicolumn{1}{l}{socio-economic } & 0.0574  \\
			status (SES)   & (0.0655)   \\
			&    \\			
			\multicolumn{1}{l}{ability} & 0.0966$^{**}$ \\
			& (0.0473) \\
			&    \\
			\multicolumn{1}{l}{district size } & -0.9514$^{***}$ \\
			   & (0.1162)\\
			&        \\
			\multicolumn{1}{l}{relative excess } & -2.0946$^{***}$  \\
			capacity   & (0.1662)   \\
			&         \\
			\multicolumn{1}{l}{gymnazium} & -0.8703$^{***}$   \\
			& (0.2368) \\
			&         \\
			\multicolumn{1}{l}{secondary} & -0.1946 \\
			& (0.2230) \\
			&        \\
			\multicolumn{1}{l}{constant} & 15.4959$^{***}$  \\
			& (0.5174)  \\
			&      \\
			\midrule
			\multicolumn{1}{l}{district FE} & Yes   \\
			\midrule
			\multicolumn{1}{l}{Observations} & 9,986  \\
			\bottomrule
		\end{tabular} 
		\begin{tablenotes}
			\tiny \item The table shows regression coefficients of students' gains on student observables.	Variables 'district size' and 'relative excess capacity' refer to the students' home districts; the school type refers to the school type of the assigned school in the integrated market. $^{*}$p$<$0.1; $^{**}$p$<$0.05; $^{***}$p$<$0.01. Standard errors in parentheses.
		\end{tablenotes}
	\end{threeparttable}
\end{footnotesize}
\end{table}

However, there are some caveats to the above conclusion. First, the variables measuring SES and student ability are highly correlated ($r=0.47$) and so there will be a large overlap of high-SES and high-ability students among those who benefit a lot from market consolidation. Second, the overall effects are rather small compared to the total variance of the consolidation gains, which is more than five kilometres (see Table \ref{tab:fullROL:gainstats}). On that account, the systematic factors driving the consolidation gains are rather small, and idiosyncratic factors seem to be the most important determinants.\footnote{This finding may in part be due to measurement error in our explanatory variables that is likely to attenuate our parameter estimates towards zero. As was described in Section \ref{sec:data}, we do not exactly observe the students' characteristics which the school can condition their admission choices on. Instead, we must rely on supplementary information from the NABC, and we also make use of imputed data because it is important to have a complete set of students for our empirical approach. Therefore, we may overestimate the contribution of the unobserved idiosyncratic preference and priority shocks to the formation of students' preferences and schools' priorities.}

\subsection{Robustness tests with different excess capacities}
\label{sec:Appendix_D}

The school system in Budapest is characterized by a large overcapacity: in 2015 there were 28,646 school places available for 10,880 students, i.e. 2.63 available seats per student.\footnote{The schools' excess capacity has been confirmed in conversation with officials from the Hungarian ministry of education on several occasions.} A natural question is whether the large gains from consolidation that we have documented crucially depend on this large imbalance between the number of available seats and the number of students. In this subsection, we address this issue by conducting a series of robustness tests in which we reduce the capacity of schools using different approaches that we think are meaningful.

First, we consider the length of the school's submitted ROL as an upper bound of its capacity, i.e. we reduce the capacity of a school to the number of students who were ranked by a school. This bound, that we refer to as ROL-LENGTH, is appropriate because students cannot be assigned to schools who did not rank them. 
A second, tighter upper bound for a school's capacity is the number of accepted students. We refer to this as ACCEPTED. The rationale for this second upper bound is that for an over-demanded school, the accepted students are its de facto capacity because there are students who would prefer to enrol in this school but cannot.  
Both bounds result in more balanced markets: we obtain an excess capacity of 1.75 and 1.21 available seats per student (down from previously 2.63). 
These levels of overcapacity in Budapest are in line with those reported for other countries.\footnote{At the country level, there is evidence that secondary schools in England have had overcapacity of school seats of between 9\% to 20\% of student demand, for all the academic years for which data is available. In 2021-22, there were 464,998 additional school places, i.e.\ 13\% more school seats than students \citep{walker2020goodschools}, latest data from \url{https://explore-education-statistics.service.gov.uk/find-statistics/school-capacity}. 	Overcapacity in secondary schools has also been documented in Northern Ireland 	and Scotland \citep{teelken2005frictions}. In the Netherlands, the city of Amsterdam has 2,434 places for 1,915 students, which corresponds to a 21\% excess capacity \citep[Table 3]{de2023performance}.	Further, at the school level, 24\% of secondary schools in the United States have an overcapacity of at least 25\%, another 36\% report an overcapacity of at least 6\% \citep[page 22, Table 1]{chaney2007public}. Finally, there is also evidence that such underenrollment is desirable from a policy perspective. Chicago Public Schools' (CPS) capacity utilization methodology, for example, states that ``a high school's ideal capacity is 80\% of its maximum capacity''. }

The rationale for reducing capacity parallels the endogenous decision-making of educational administrators. Faced with such choices, administrators may prioritize cutting capacity at schools that are not currently over-demanded. Both capacity bounds preserve the level of competition in the market, because only unassigned seats are removed. We emphasize that the percentage of schools that are over-demanded in our data (72\%) is already in line with the corresponding percentages observed in several US school systems.\footnote{\citet{angrist2020simple} report that 66\% of NYC middle schools, 71\% of NYC high schools, and 82\% of Denver middle schools are over-demanded.} Therefore, to ensure comparability with other cities on this measure, it appears sensible to consider counterfactuals that sustain this level of competition rather than increasing it. With the BALANCED counterfactual, introduced shortly below, we create a scenario where competition intensifies.

Furthermore, we adapt the scenarios ROL-LENGTH and ACCEPTED by guaranteeing that every school district is able to guarantee a local seat for local students.\footnote{To achieve this, for each school, its capacity is multiplied by the district capacity per student in the school's district.} These two cases, which we denote ROL-LENGTH + and ACCEPTED +, are depicted in panels (d) and (e) in Figure \ref{fig:excesscap}. We can see that in our four robustness tests (panels b to e), the school system has a lower excess capacity than the original market, but remaining the same characteristics and capacity distribution as in the original market (panel a). Finally, we perform a final robustness test, which we achieved by scaling the schools' capacities proportionally within each district until the total number of seats equals the total number of students (up to the integer constraint). This is presented in panel (f) in Figure \ref{fig:excesscap}. 
We refer to this specification as BALANCED. Table \ref{tab:district_integration_est_balanced} contains a summary of the excess seats for each setting.

\begin{figure}[h!]
	    \centering 
	\begin{subfigure}{0.25\textwidth}
		\begin{center}
			\includegraphics[width=\textwidth]{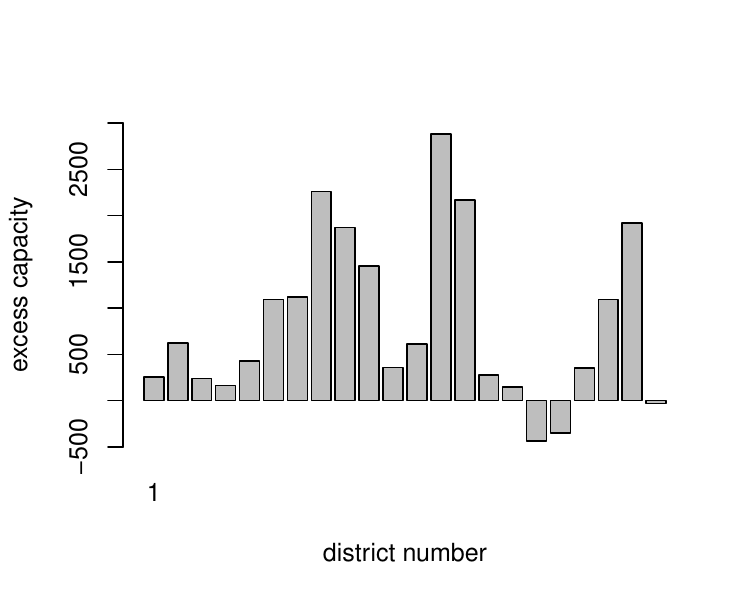}
		\end{center}
		\caption{Reported}
		\label{fig:excesscap:observed}
	\end{subfigure}
	\begin{subfigure}{0.25\textwidth}
		\begin{center}
			\includegraphics[width=\textwidth]{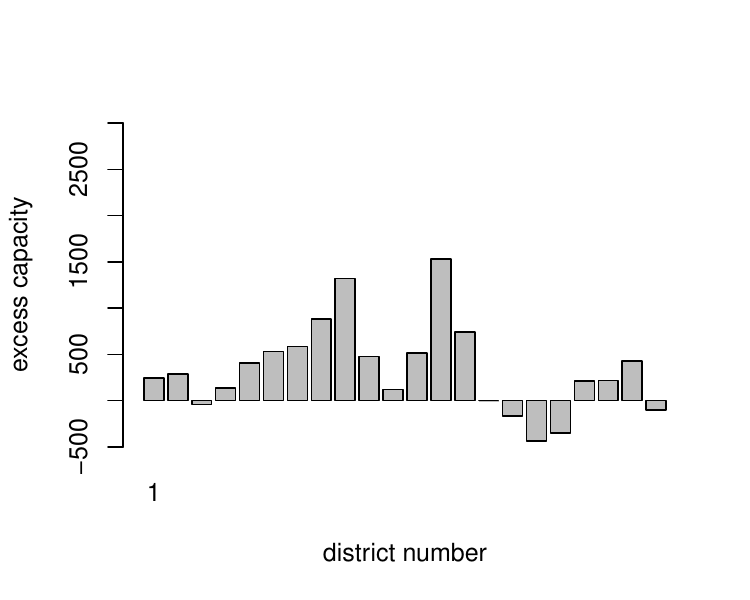}
		\end{center}
		\caption{Using ROL length}
		\label{fig:excesscap:ROL}
	\end{subfigure}
	\begin{subfigure}{0.25\textwidth}
	\begin{center}
		\includegraphics[width=\textwidth]{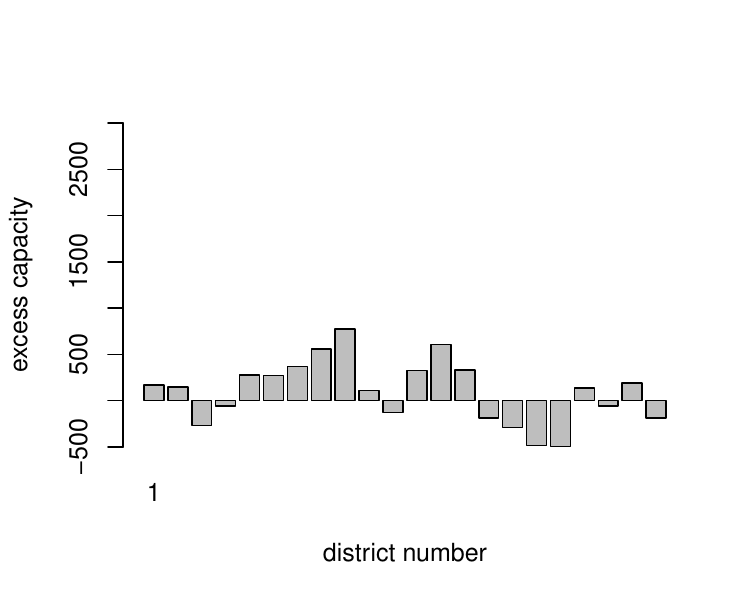}
	\end{center}
	\caption{Using accepted students}
	\label{fig:excesscap:accepted}
\end{subfigure}
\medskip

\centering
	\begin{subfigure}{0.25\textwidth}
		\begin{center}
			\includegraphics[width=\textwidth]{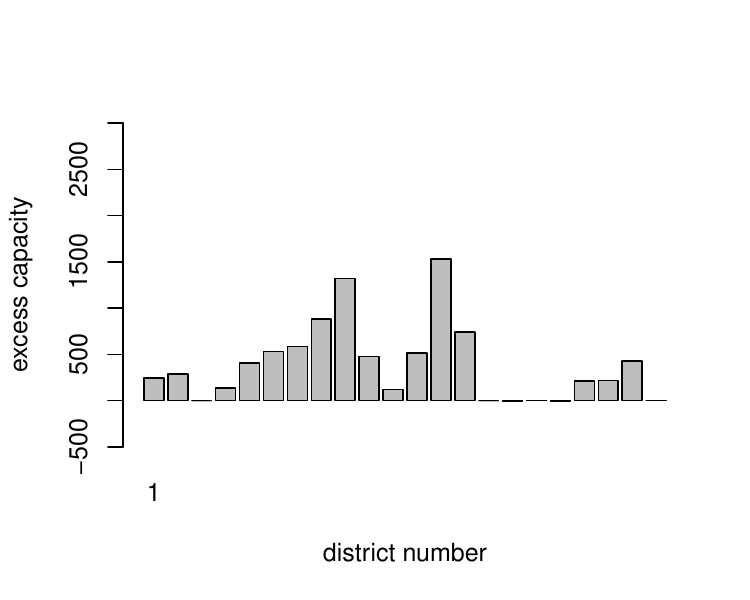}
		\end{center}
		\caption{ROL length +}
		\label{fig:excesscap:ROL-balanced}
	\end{subfigure}
	\begin{subfigure}{0.25\textwidth}
		\begin{center}
			\includegraphics[width=\textwidth]{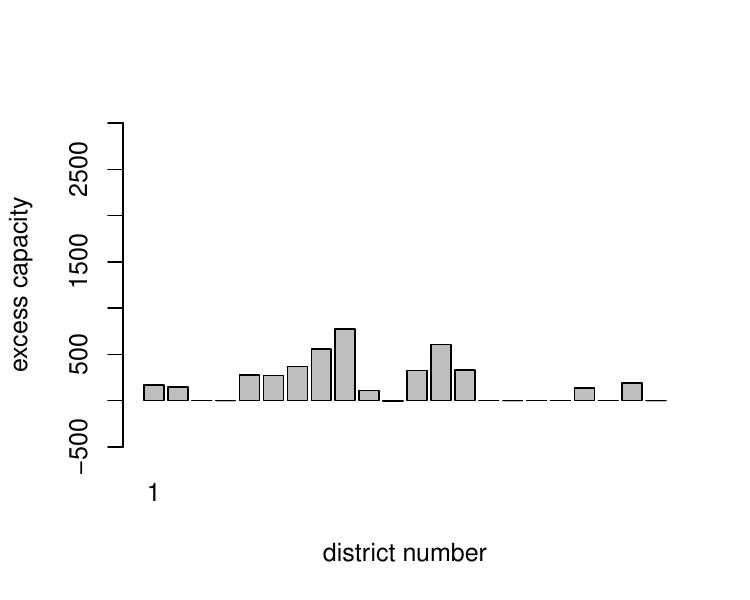}
		\end{center}
		\caption{Accepted students +}
		\label{fig:excesscap:accepted-balanced}
	\end{subfigure}	
	\begin{subfigure}{0.25\textwidth}
	\begin{center}
		\includegraphics[width=\textwidth]{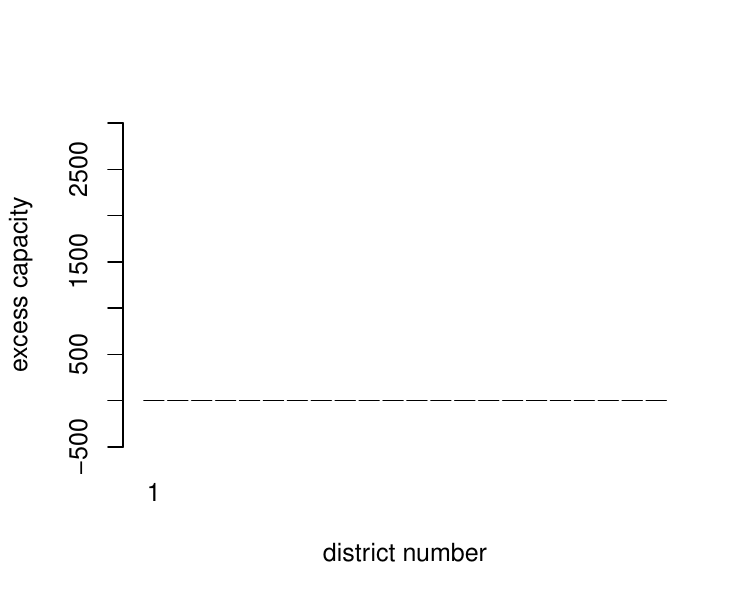}
	\end{center}
	\caption{Fully balanced}
	\label{fig:excesscap:balanced}
\end{subfigure}

	\caption{Excess capacity per district for different specifications of school capacities. + indicates a that excess capacity has been balanced for over-demanded districts, i.e. those with a shortage of places.}
	\label{fig:excesscap}
\end{figure}

We repeat the analysis from the main part of the paper (Table \ref{tab:district_integration_est}) with the different specifications and report results on the share of winners and losers in Table \ref{tab:district_integration_est_balanced}. All figures are reported in percentages to allow for a meaningful comparison between the specifications. The columns with the symbols $-$, $0$ and $+$ denote the share of losers, indifferences and winners from consolidation among the matched students, respectively. With respect to winners and losers, all alternative specifications generate very similar results to those reported in the main part of the paper. The intuition is that in all specifications -- except for the BALANCED one -- the assignment does not change in the consolidated market. But reducing the excess seats reduces welfare in the unconsolidated market. Thus, the reduction of excess seats makes the consolidated market more attractive, producing more winners. This latter effect is relatively small, however, with the share of winners ranging between 73\% and 77\% for all specifications (this share was 73\% in the original specification), except for BALANCED. 

In the consolidated market, all students are assigned to the same school in all specifications because we only drop unused capacities.
In the unconsolidated market, however, we still find that large share of students is assigned to the same school (last column of Table \ref{tab:district_integration_est_balanced}).  We thus expect similar utility gains for all specifications.
This is confirmed in the cardinal utility gains presented in the boxplots in Figure \ref{fig:gains:specifications}. 
The cardinal consolidation gains are the virtually identical in most of our specifications, namely ROL-LENGTH, ACCEPTED, ROL-LENGTH + and ACCEPTED +, both when measured in latent utility units or in equivalent kilometres. 
Only the BALANCED specification, characterised by increased competition for school places, yields gains close to zero. Consequently, welfare gains diminish as a reduction in excess capacity leads to intensified competition within the consolidated market.

\begin{center}
	\begin{footnotesize}
		
	\begin{threeparttable}
		\caption{Share of losers ($-$) and winners ($+$) from consolidation (with estimated preferences),  for different specifications of school capacity.}
		\label{tab:district_integration_est_balanced}
		\begin{tabular}{lr|rrr|rr} 
			\toprule 
			Specification  & Excess seats & $-$ & $0$ & $+$ & Unmatched & Same School \\ 
			\midrule 
			Reported    & 163\%& 2\%  & 25\% & 73\% & 8\%  & 100\% \\
			ROL-LENGTH  & 75\% & 2\%  & 25\% & 73\% & 11\% & 92\% \\
			ROL-LENGTH + & 85\% & 2\%  & 25\% & 73\% & 0\%  & 90\% \\		
			ACCEPTED  & 21\% & 3\%  & 20\% & 77\% & 21\% & 70\% \\
			ACCEPTED + & 42\% & 2\%  & 23\% & 75\% & 0\%  & 74\%  \\
			BALANCED  &  0\% & 41\% & 13\% & 46\% & 0\%  & 40\% \\		
			\bottomrule
		\end{tabular} 	
		\begin{tablenotes}
			\tiny \item Excess seats refers to percentage of excess seats per student. The symbols $-$, $0$ and $+$ denote the share of losers, indifferences and winners from consolidation among the matched students, respectively. The second last column gives the share of unmatched students. The last column gives the share of matched students assigned to the same school as in the original specification in the unconsolidated market. Data obtained using estimated preferences.
		\end{tablenotes}
	\end{threeparttable}

\end{footnotesize}
\end{center}

The robustness tests conducted demonstrate that, as long as competition for school places remains the same, a majority of students will benefit from district consolidation, even when the excess capacity of schools is lower than that indicated by the Hungarian Education Authority. Our findings are expected to generalise to school markets with similar levels of competition to those observed in Budapest. This encompasses markets for middle and high schools in New York City, where the share of oversubscribed schools is between 66\% and 71\% \citep{angrist2020simple} and therefore lower than in Budapest. In such settings, sizeable gains are anticipated, provided there is an overall excess capacity, as indicated in our specifications (a) to (e). Conversely, in markets characterised by a lack of excess capacity, as in our specification (f), gains are likely to approach zero.

\begin{figure}
	\begin{subfigure}[c]{0.49\textwidth}
		\begin{center}
			\includegraphics[width=\textwidth]{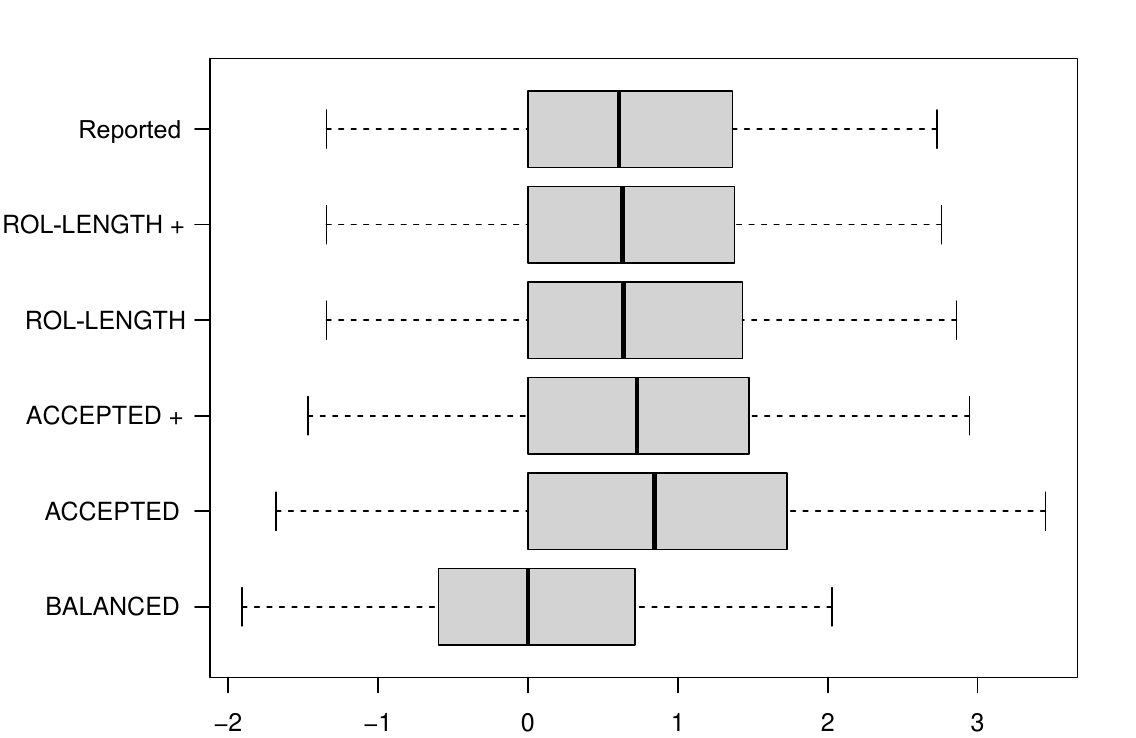}
		\end{center}
		\caption{Total gains in latent utility units}
		\label{fig:gains:utility}
	\end{subfigure}
	\hfill
	\begin{subfigure}[c]{0.49\textwidth}
		\begin{center}
			\includegraphics[width=\textwidth]{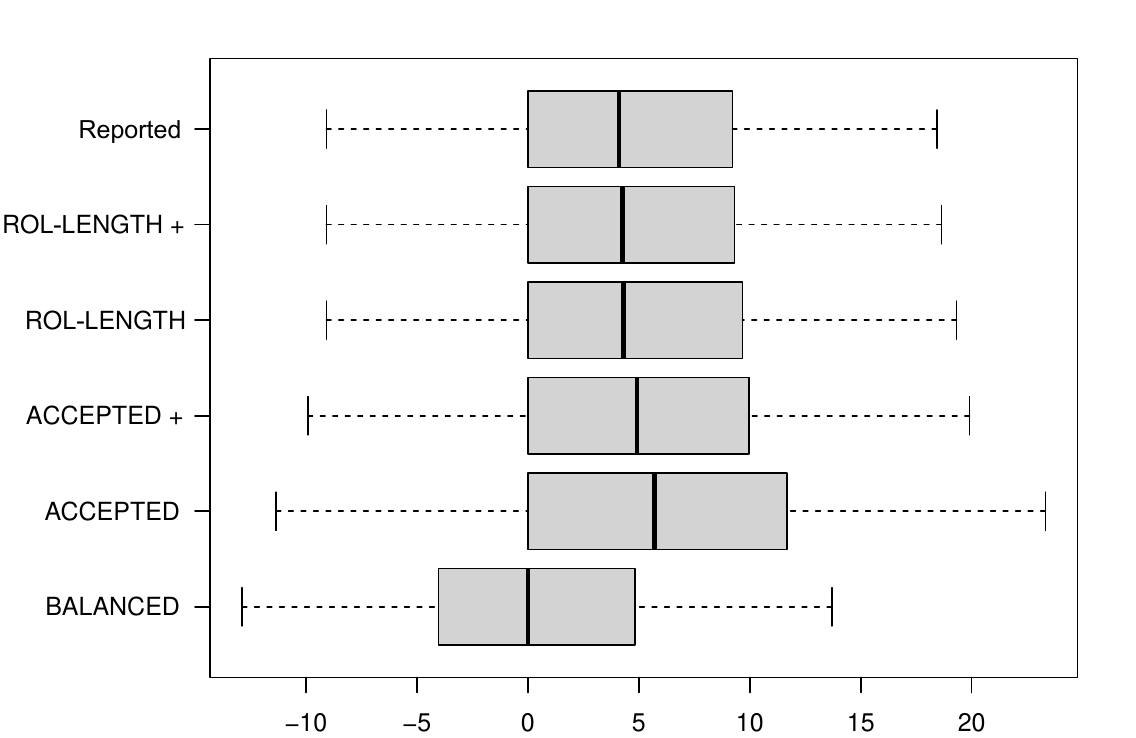}
		\end{center}
		\caption{Total gains in equivalent kilometres}
		\label{fig:gains:km}
	\end{subfigure}
	\caption{Consolidation gains with estimated preferences for different specifications of school capacity.}
	\label{fig:gains:specifications}
\end{figure}

\section{Conclusion}
\label{sec:conclusion}
Our article provides the first study of school district consolidation using both a market design model and empirical evidence from the Hungarian secondary school admission system.  
The theoretical predictions show that market consolidation leads to substantial welfare gains for students, and that students who live in smaller markets, or in markets with fewer available school seats, are expected to have larger welfare gains. Our empirical results confirm that the average student greatly benefits from having a consolidated school market, and that more than half of all students are better off in the consolidated school market.  We find that the gains from consolidation are larger in school districts which have little capacity compared to the number of students, and in smaller districts. By and large, these results are independent of whether we use students' stated preferences or an inferred complete preference ranking. Moreover, our results indicate that high-ability students benefit more from market consolidation than other students. 

As a by-product, we establish a method to consistently estimate students' preferences in school markets with school-specific admission criteria unknown to the researcher. Our estimation approach avoids a bias that is otherwise introduced by students' strategic reporting of their preferences. We show by means of a Monte Carlo study that this method works as intended. We find that students favour nearby schools which have a high academic reputation and peers with a high socio-economic status, but dislike having to sit school-specific entrance exams. We also find that there is evidence for sorting on academic ability, and social status. Schools appear to base their admissions mostly on the students' abilities in Hungarian, with math scores and socio-economic background being less important.

Some interesting questions remain open for future research. Perhaps the main one is that we computed consolidation gains under the assumption that the students' and schools' characteristics remain fixed throughout, whereas only the admission system is changed. But, of course, there are other effects that could be taken into account in future work, such as residential mobility or changes in schools' quality in response to district consolidation. Incorporating these effects would be a thorny issue even if we had data before and after district consolidation has taken place, because of likely anticipatory effects.\footnote{It could appear to the reader that one could estimate the second order consequences of district consolidation by means of an iterative procedure whereby the schools' average academic qualities, and students' preferences, are updated in turns until a ``steady state'' is reached. But in our opinion, such a mechanistic steady state analysis is unlikely to mirror the multitude of individual and institutional responses and would thus be highly speculative. We refrained from this approach, focusing on what we can measure, and not on what we cannot measure.} Our article is transparent about these limitations and instead takes a well-defined partial equilibrium approach to school district (dis-)integration. Thus, our results should be interpreted as measuring the isolated, or partial effect of the admission system on students' welfare. We think that we can accurately describe and measure this partial effect, and that it is a valuable statistic in itself to inform the debate on the merits of centralized assignment mechanisms.

\singlespacing

{\bf Acknowledgements.}
We are particularly grateful to two anonymous reviewers and the associate editor for detailed advice that improved the contents and exposition of this paper. We also acknowledge helpful comments from Inácio Bó, Li Chen, Gabrielle Fack, Ellen Greaves, Karol Mazur, Juan Sebastián Pereyra, Olmo Silva, Bertan Turhan, M. Bumin Yenmez and audiences at the Royal Economic Society conference, the  World congress of the Game Theory Society, the  Matching in Practice workshop, the WZB Berlin workshop on designing and evaluating matching markets, the Public Economic Theory conference and the Irish Economic Association conference. We are thankful to the Hungarian Educational Authority in Budapest for running our code on their data on-site and sharing the aggregate results presented in the paper. We acknowledge funding from the Leibniz Association as part of project K125/2018: Improving school admissions for diversity and better learning outcomes. Sarah Fox and Ashleigh Neill proofread the paper. This  manuscript was previously circulated under the name ``What happens when separate and unequal school districts merge?". Any errors are our own.

\setlength{\bibsep}{0cm}

\bibliographystyle{ecta}

\def\appendixname{Appendix}

\clearpage

\appendix

\renewcommand{\theequation}{A\arabic{equation}} 
\setcounter{equation}{0}
\renewcommand{\thefigure}{A\arabic{figure}} 
\setcounter{figure}{0}
\renewcommand{\thetable}{A\arabic{table}} 
\setcounter{table}{0}

\section{Proofs}
\label{sec:Appendix_A}
\begin{proof}[Proof of Proposition \ref{thm:prop3}.] 
	
	A well-known result in matching markets is that, in a one-to-one Gale-Shapley matching market with $N$ agents on each side, there are approximately $N \log (N)$ proposals made by students in the student-proposing execution of the algorithm. This result is obtained using the analogy of how long will it take to collect $N$ different coupons when a random coupon is collected each day (a problem known as {\it the coupon collector problem}). Because students apply to schools based on their preferences, starting from their most preferred one, this result implies that the average student is matched to a school with expected rank equal to $\log (N)$ \citep{wilson1972,knuth1976,pittel1989}. 
	
	The coupon collector problem has been generalized to study how long will it take to collect $q$ times each of the $N$ different coupons. \cite{newman1960double} show that it takes $N[\log(N)+(q-1)\log\log(N)+o(1)]$ days. This result implies that, in a many-to-one matching market with $N$ schools, each with $q$ available seats, and $qN$ students in total, the student-proposing deferred acceptance algorithm assigns the average student to a school approximately ranked $\frac{1}{q}[\log(N)+(q-1)\log\log(N)]$ when $N$ is large and $q$ is constant. One interpretation for this result is that it takes $\log (N)$ applications for a school to fill its first seat, but only takes $\log \log (N)$ extra applications to fill any other seat besides the first. The theoretical prediction by Newman and Shepp is very close to the average student rank observed in simulations, as depicted in Table \ref{tab:manytoone}.
	
	\begin{table}[h!]
		\centering
		\caption{Average Student Rank (average over 1,000 simulations).}
		\label{tab:manytoone}
		
		\begin{tabular}{llllllll}
			\toprule
			$n \times q$ & Observed & Approximation  & $n \times q$ & Observed & Approximation \\
			\midrule
			200 (100 $\times$ 2) & 3.53 & 3.07  & 200 (40 $\times$ 5)  & 2.15 & 1.78 \\
			400 (200 $\times$ 2) & 4.01 & 3.48  & 400 (80 $\times$ 5)  & 2.37 & 2.06 \\
			1,000 (500 $\times$ 2)  & 4.54 & 4.02  & 1,000 (200 $\times$ 5) & 2.66 & 2.39\\
			\bottomrule
		\end{tabular}
	\end{table}
	
	Going back to the original one-to-one setting with $q=1$, consider that the market now has $N$ students and $N+K$ schools, with $k$ being a positive constant. \cite{ashlagi2017} show that, with high probability, the expected rank for students equals $\frac{N+K}{N}\log(\frac{N+K}{K})$.
	
	Combining both expressions, we can approximate the expected rank of students in our many-to-one setup with market imbalances, as follows:
	$$\rk_T ( \sigma_\sosm(\cdot,\Omega) ) \approx \frac{N+K}{q N} \left[ \log\left(\frac{N+K}{K}\right) + (q-1) \log \log \left(\frac{N+K}{K}\right)\right] $$
	
	The approximation above matches the expected rank observed in simulations remarkably well, as shown in Table \ref{tab:mainapprox}.
	\begin{table}[h!]
		\centering
		\caption{Average Student Rank (average over 1,000 simulations).}
		\label{tab:mainapprox}
		\begin{tabular}{lllllll}
			\toprule
			$N \times q$ & \multicolumn{2}{c}{$K=2$} & \multicolumn{2}{c}{$K=5$} &\multicolumn{2}{c}{$K=10$}\\
			&		Obs & Approx  &  Obs & Approx & Obs & Approx   \\
			\midrule
			200 (100 $\times$ 2) &	2.9 	& 3.1	 	&	2.5 	&	 2.7	&	2.1 	& 2.2\\
			200 (40 $\times$ 5) &	1.7 	& 1.8	 	&	1.5 	&	 1.6	&	1.3 	& 1.2\\
			400 (200 $\times$ 2) &	 3.2	& 	3.5 	&	 2.9	&	3.1 	&	2.5 	& 2.6\\
			400 (80 $\times$ 5) &	1.9 	& 2.1	 	&	1.7 	&	1.8 	&	1.5 	& 1.5\\
			1,000 (500 $\times$ 2)  &	3.8 	& 4.0	 	&3.4	 	&	3.6 	&	3.0 	& 3.1\\
			1,000 (200 $\times$ 5)  &	2.1 	& 2.4	 	&	2.0 	&	2.2 	&	1.7 	& 1.8\\
			\bottomrule
		\end{tabular}
	\end{table}

	To compare the gains from consolidation, we only need to approximate $\rk_T ( \sigma_\sosm(\cdot,D) )$. To do this, we define the {\it relative rank} of a school $s$ in the preference order of a student $t \in T^{D_i}$ (over potential schools in within his own district) as $\hat \rk_t(s) \coloneqq \left\vert \{ s' \in S^{D_i}: s' \succcurlyeq_t s \} \right\vert$. Given a matching $\mu$, the {\it students' relative average rank of schools} is defined by
	\begin{eqnarray*}
		\hat \rk_T (\mu) &\coloneqq& \frac{1}{\left\vert \overline{T} \right\vert} \sum_{t \in \overline{T}} \hat \rk_t(\mu(t))
	\end{eqnarray*}
	where $\overline{T}$ is the set of students assigned to a school under matching $\mu$.
	
	In a district with $qn_i$ students, $q(k_i+n_i)$ school seats and with $k_i>0$, our previous approximation tells us that the relative rank equal:
	\begin{equation}
		\label{eq:ap1}
		\hat \rk_T ( \sigma_\sosm(\cdot,D) ) \approx \frac{n_i+k_i}{q n_i} \left[\log \left(\frac{n_i+k_i}{k_i}\right) + (q-1) \log \log \left( \frac{n_i+k_i}{k_i} \right)  \right]
	\end{equation}

	The final step in the proof follows the proof of Proposition 3 in \cite{ortega2018}. To relate the students' relative average rank of schools before consolidation to the absolute ranking, suppose that a school is ranked $h$ among all $n_i+k_i$ schools in its district. A random school from another district could be better ranked than school 1, between schools 1 and 2, ..., between schools $h-1$ and $h$, ..., between schools $n_i+k_i-1$ and $n_i+k_i$, or after school $n_i+k_i$. Therefore, a random school from another district is in any of those gaps with probability $1/(n_i+k_i+1)$ and thus has $h/(n_i+k_i+1)$ chances of being more highly ranked than our original school with the relative rank $h$. There are $N+K-n_i-k_i$ schools from other districts. On average, $\frac{h(N+K-n_i-k_i)}{n_i+k_i+1}$ schools will be ranked better than it. Furthermore, there were already $h$ schools in its own district better ranked than it. This implies that his expected ranking is $h+\frac{h(N+K-n_i-k_i)}{n_i+k_i+1} =\frac{h(N+K+1)}{n_i+k_i+1}  \approx \frac{h(N+K)}{n_i+k_i}$. Substituting $h$ for \eqref{eq:ap1}, we obtain students' relative average rank of schools before consolidation, which equals
	$$		\rk_T ( \sigma_\sosm(\cdot,D) ) \approx \frac{N+K}{q n_i} \left[\log \left(\frac{n_i+k_i}{k_i}\right) + (q-1) \log  \log \left( \frac{n_i+k_i}{k_i}  \right) \right]$$
	
	After some algebra, it follows that 
	\begin{equation}
		\label{eq:main}
		\gamma_T (\sigma_\sosm) \approx \frac{N+K}{q} \left[\frac{\log \left(\frac{n_i+k_i}{k_i}\right) + (q-1) \log \log \left( \frac{n_i+k_i}{k_i} \right)}{n_i} - \frac{ \log\left(\frac{N+K}{K}\right) + (q-1) \log  \log \left(\frac{N+K}{K} \right)}{N} \right]
	\end{equation}
\end{proof}

Note that expression \ref{eq:main} is positive if:
$$ \left[ \frac{\log \left(\frac{n_i+k_i}{k_i}\right)}{n_i} -  \frac{\log\left(\frac{N+K}{K}\right)}{N} \right] + (q-1)\left[ \frac{\log \log \left( \frac{n_i+k_i}{k_i} \right)}{n_i} -  \frac{\log  \log \left(\frac{N+K}{K} \right)}{N} \right] >0$$

In particular, expression \ref{eq:main} is positive when both expressions \refeq{eq:part1} and \refeq{eq:part2} below hold
\begin{equation}
\label{eq:part1}
 \frac{\log \left(\frac{n_i+k_i}{k_i}\right)}{n_i} >  \frac{\log\left(\frac{N+K}{K}\right)}{N} 
\end{equation} 

and 

\begin{equation}
\label{eq:part2}
 \frac{\log \log \left( \frac{n_i+k_i}{k_i} \right)}{n_i} >  \frac{\log  \log \left(\frac{N+K}{K} \right)}{N} 
\end{equation}

Fixing $\alpha_i=N /  n_i$, $\beta_i=K/k_i$, $x=\frac{N+K}{K}$ and $y_i=\frac{n_i+k_i}{k_i}$, we can rewrite expressions \refeq{eq:part1} and \refeq{eq:part2} as
  \begin{equation}
	\alpha_i>  \frac{\log(x)}{\log (y_i)} 	
	\text{ and }
	\alpha_i>  \frac{\log\log(x)}{\log \log(y_i)} 	
\end{equation}

Assuming $\log\log(y_i)>1$, if the first expression holds (i.e. $\alpha_i>  \frac{\log(x)}{\log (y_i)}$), the second one must hold as well (i.e. $	\alpha_i>  \frac{\log\log(x)}{\log \log(y_i)}$), because
\begin{equation}
	\alpha_i>\log(\alpha_i)+1>\frac{\log(\alpha_i)+\log\log(y_i)}{\log\log(y_i)}=\frac{\log[\alpha_i \cdot \log(y_i)]}{\log\log(y_i)}>\frac{\log\log(x)}{\log\log(y_i)}
\end{equation}

Where the first inequality holds because $\alpha_i>1$ and the last inequality uses the assumption of $\alpha_i>  \frac{\log(x)}{\log (y_i)}$, which implies  $\log[\alpha_i \log(y_i)]>\log\log(x)$ because the $\log$ function is strictly increasing. 

Therefore, to check that expression \ref{eq:main} is positive, we can simply show that $\alpha_i>\log(x)/\log(y_i)$. 
Provided that $\alpha_i>\beta_i$, this condition holds because
\begin{eqnarray}
	\label{eq:expl}
	\alpha_i&>&\log(\alpha_i)+1\nonumber \\
	&>&\log\left(\frac{\alpha_i}{\beta_i}\right)+1=\log\left(\frac{\alpha_i}{\beta_i}\right)+\frac{\log \left(\frac{n_i+k_i}{k_i}\right) }{\log \left(\frac{n_i+k_i}{k_i}\right)}\nonumber \\
	&>&\frac{\log\left(\frac{\alpha_i}{\beta_i}\right)+\log \left(\frac{n_i+k_i}{k_i}\right) }{\log \left(\frac{n_i+k_i}{k_i}\right)}=\frac{\log \left(\frac{\alpha_i n_i\, +\, \alpha_i k_i}{\beta_i k_i}\right) }{\log \left(\frac{n_i+k_i}{k_i}\right)}\nonumber \\
	&>&\frac{\log \left(\frac{\alpha_i n_i\, +\, \beta_i k_i}{\beta_i k_i}\right) }{\log \left(\frac{n_i+k_i}{k_i}\right)}
\end{eqnarray}
 where the first inequality holds because $\alpha_i>1$ by definition, the second one holds because $\beta_i>1$ by definition, the third inequality assumes that $\log(y_i)>1$ (recall that previously we further assumed that $\log\log(y_i)>1$), and the last inequality uses our assumption that $\alpha_i>\beta_i$.
 
 We note that the conditions imposed are sufficient to guarantee that the gains from consolidations are positive, but not necessary. In particular, we do not need to assume that $\alpha_i>\beta_i$. The same conclusion holds can be obtained as long as $\alpha_i>\log(\beta_i)+1$, following almost the same steps as in equation \ref{eq:expl}.%

\newpage
\section{Correlated preferences and priorities}
\label{sec:simcorrelation}

The simulations are conducted for two markets, each with $q=2$ seats per school. 
In the market with two districts, district size is held constant with $n_1=n_2=100$ students per district but overcapacity varies with $q\cdot k_1=2\cdot 5 = 10$ seats in district one and balancedness ($q\cdot k_2 = 2\cdot 0 = 0$) in district two. 
In the market with three districts, overcapacity is kept fixed with $k_1=k_2=k_3=1$ but district sizes vary with $n_1=40$, $n_2=60$ and $n_3=100$ students.

For each of these two markets, the gains from consolidation are obtained for 1,000 simulated draws with different correlation parameters. We start with a randomly drawn master list of identical school preferences and school priorities. The parameter $\rho_t$ then determines the share of schools that are kept fixed in the students' rank order lists. The remaining share of ($1 -\rho_t$) is randomly perturbed. $\rho_t=0$ is equivalent to a random market. $\rho_t=1$ corresponds to all students submitting the same rank order list over schools. Correspondingly, the parameter $\rho_s$ determines the share of students kept fixed in schools' rank order lists.

Table \ref{tab:correlated_markets_disaggregate} reports the average of the rank gains for all students in the two district market and the three district market, respectively. Panel A gives the aggregated results over all districts. In the original iid random market, $\rho_t=\rho_s=0$ and the average rank gain is 2.19 in the two districts market, and 4.03 in the three districts market. The results indicate that the gains from become even larger when students' preferences become more correlated (except for the corner case $\rho_t=1$, where student preferences are perfectly correlated) and remain largely unchanged for all parameters on priorities' correlation. Importantly, the overall rank gains remain positive for any parameter constellation. 

Panel B presents the rank gains disaggregated by district. The results show that our comparative statics hold: the gains from integration are still larger for smaller and more over demanded districts. 
To see the size effect, consider the three-district case, with district sizes $n_1=20$, $n_2=30$ and $n_3=50$ and $q=2$ (i.e. 40, 60, and 100 students). In the random market, the average rank gains are largest in the smallest district (6.88), followed by the second largest district (4.65). The same ordering holds for any other combination of correlation parameters (except for the corner case  $\rho_t=1$).

To see the imbalance effect, consider the random market in the two-district case with an overcapacity of 10 students for district 1 ($k_1=5,q=2$) and zero for district 2 ($k_2=0$). Gains are larger for the balanced district 2 (average gains of 2.97) than for the under-demanded district 1 (gains of 1.41). For any other parameter constellation, it holds that the more under-demanded district 1 exerts smaller gains from integration.

\begin{table}[!h]
	\begin{center}
		\centering
			\caption{Rank gains from consolidation in markets with correlated preferences. Panel A presents overall rank gains. Panel B displays gains disaggregated by districts.}
		\label{tab:correlated_markets_disaggregate}
		\begin{small}
			\begin{tabular}{rrrrrrrrrrrrrr}
				\toprule
				\multicolumn{14}{l}{\textit{Panel A: Overall rank gains}}\\
				\multicolumn{14}{l}{ }\\
				&&& \multicolumn{5}{c}{Two district market} &       & \multicolumn{5}{c}{Three district market} \\
				\cmidrule{4-8}\cmidrule{10-14} \multicolumn{3}{r}{$\rho_t$} & 0 & 0.25 & 0.5 & 0.75 & 1 &  & 0 & 0.25 & 0.5 & 0.75 & 1 \\
				\midrule
				$\rho_s$ & districts & & \multicolumn{11}{c}{ } \\
				\cmidrule{1-2}
				0 & all && 2.19 & 2.70 & 3.87 & 7.13 & 0.1 & & 4.03 & 4.88 & 6.93 & 11.69 & 0.04 \\
				0.25 & all && 2.18 & 2.73 & 3.79 & 7.01 & 0.1 & & 4.00 & 4.83 & 6.77 & 11.50 & 0.04 \\
				0.5 & all && 2.15 & 2.70 & 3.78 & 6.80 & 0.1 & & 4.00 & 4.76 & 6.67 & 11.26 & 0.04 \\
				0.75 & all && 2.15 & 2.62 & 3.68 & 6.55 & 0.1 & & 3.93 & 4.63 & 6.48 & 10.58 & 0.04 \\
				1 & all & & 2.07 & 2.47 & 3.34 & 5.56 & 0.1 &  & 3.77 & 4.43 & 5.95 & 9.33 & 0.04 \\
				\multicolumn{14}{l}{ }\\
				\midrule
				\multicolumn{14}{l}{\textit{Panel B: Rank gains by district}}\\
				\multicolumn{14}{l}{ }\\
				&&& \multicolumn{5}{c}{Two district market} &       & \multicolumn{5}{c}{Three district market} \\
				\cmidrule{4-8}\cmidrule{10-14} \multicolumn{3}{r}{$\rho_t$} & 0 & 0.25 & 0.5 & 0.75 & 1 &  & 0 & 0.25 & 0.5 & 0.75 & 1 \\
				\midrule
				$\rho_s$ & district & & \multicolumn{11}{c}{ } \\
				\cmidrule{1-2}			
				0 & 1 && 1.41 & 1.68 & 2.41 & 4.43 & -1.67 & & 6.88 & 7.97 & 11.16 & 17.43 & -0.88 \\
				0 & 2 && 2.97 & 3.72 & 5.33 & 9.84 & 1.87 & & 4.65 & 5.73 & 7.96 & 13.43 & 0.19 \\
				0 & 3 && -- & -- & -- & -- & -- & & 2.52 & 3.13 & 4.61 & 8.36 & 0.32 \\
				\midrule
				0.25 & 1 && 1.38 & 1.70 & 2.38 & 4.38 & -1.71 & & 6.86 & 7.95 & 10.90 & 17.18 & -1.24 \\
				0.25 & 2 && 2.97 & 3.75 & 5.21 & 9.64 & 1.91 & & 4.61 & 5.56 & 7.78 & 13.30 & -0.14 \\
				0.25 & 3 && -- & -- & -- & -- & -- & & 2.49 & 3.14 & 4.51 & 8.15 & 0.66 \\			
				\midrule
				0.5 & 1 && 1.39 & 1.71 & 2.40 & 4.20 & -1.77 & & 6.80 & 7.85 & 10.62 & 16.75 & -1.07 \\
				0.5 & 2 && 2.92 & 3.69 & 5.16 & 9.39 & 1.98 & & 4.60 & 5.50 & 7.87 & 12.92 & -0.02 \\
				0.5 & 3 && -- & -- & -- & -- & -- & & 2.52 & 3.08 & 4.36 & 8.06 & 0.52 \\
				\midrule
				0.75 & 1 && 1.39 & 1.67 & 2.37 & 4.20 & -1.82 & & 6.70 & 7.76 & 10.50 & 15.71 & -0.92 \\
				0.75 & 2 && 2.92 & 3.57 & 5.00 & 8.89 & 2.02 & & 4.51 & 5.33 & 7.53 & 12.28 & -0.16 \\
				0.75 & 3 && -- & -- & -- & -- & -- & & 2.46 & 2.96 & 4.25 & 7.52 & 0.54 \\
				\midrule
				1 & 1 && 1.40 & 1.68 & 2.34 & 4.01 & -1.60 &  & 6.51 & 7.60 & 10.00 & 15.12 & -0.79 \\
				1 & 2 && 2.74 & 3.25 & 4.34 & 7.10 & 1.80 &  & 4.39 & 5.10 & 6.88 & 10.87 & -0.31 \\
				1 & 3 &&-- & -- & -- & -- & -- &  & 2.31 & 2.76 & 3.77 & 6.10 & 0.58 \\
				\bottomrule
			\end{tabular}%
		\end{small}
	\end{center}

\end{table}

\clearpage

\section[Computation of the valuation bounds]{Explicit computation of the bounds on latent valuations}
\label{sec:Appendix_B}
The estimation procedure relies on imposing upper and lower bounds on the latent valuations. This section describes explicitly how these bounds can be computed at every step of the estimation procedure, under various identifying restrictions. For convenience, we repeat the notation that is used to describe students' and schools' ordinal preferences and priorities here. 

We denote the observed rank order list of student $i$ of length $L_t$ as $\mathcal{L}_t = (s_t^1, s_t^2, \ldots, s_t^{L_i})$, where $s_t^k\in S$. Denote the rank that student $t$ assigns to school $s$ as $rk_t(s)$, with $1 \leq rk_t(s) \leq L_t$ if $s \in \mathcal{L}_t$ and $rk_t(s) = \emptyset$ else. Collect all observed ranks into $\mathbf{rk} = \{rk_t(s)\}_{t \in T, s \in S}$. The preference orderings that is induced by these observed ranks are a subset of students' unobserved strict preference ordering $\mathbf{\succ} = \{\succ_t\}_{t \in T}$, i.e. $ rk_t(s) < rk_t(s') \Rightarrow s \succ_t s'$ but not vice versa, because students may find it optimal to not rank all schools if the application procedure is costly. This is the ``skipping at the top'' and ``truncation at the bottom'' problem that was discussed in the main text and that precludes the application of standard revealed preference arguments to estimate a reduced-form model of students' preferences.

Similarly, denote the set of students who apply to school $s$ as $\mathcal{L}_s$, and let the the priority number that school $s$ assigns to student $t$ be $pr_s(t)$. Priority numbers are like ranks, in that they take discrete values and a lower priority number means higher priority. Schools are required to prioritize all students who apply to them, but they may rank some students as ``unacceptable''. We say that $pr_s(t)=+\infty$ if student $t$ is unacceptable to school $s$, and $pr_s(t)=\emptyset$ if student $t$ did not apply at school $s$. Furthermore, denote the set of ranked students that are acceptable to school $s$ as $\ell_s = \{t \in \mathcal{L}_s: pr_s(t)<\infty\}$ and define the largest priority number of any school $s$ as $\overline{pr}_s = \max\{pr_s(t): t \in \mathcal{L}_s\} \in \{|\ell_s|, \infty\}$. Thus, $pr_s(t) \in \{1,2,\ldots,|\ell_s|,\infty,\emptyset\}$. The set of all observed priority rankings is given by $\mathbf{pr} = \{pr_s(t)\}_{t \in T, s \in S}$. Again, the priority structure induced by $pr_s$ is a subset of the unobserved true priority ordering $\mathbf{\rhd} = \{\rhd_s\}_{s \in S}$.

Because the bounds depend on the observed ranks and priorities, but also on the latent valuations of students and schools, they must be computed anew in every iteration of the Gibbs sampler. More concretely, the vector of latent utilities at the current iteration step $k$ is constructed as
\begin{equation*}
	\mathbf{U}^{(k)}_{ij} = 
	\begin{cases}
		\mathbf{U}^{(k)}_{ij} & \text{if the pair } ij \text{ has been visited in iteration } k \\
		\mathbf{U}^{(k-1)}_{ij} & \text{else.}
	\end{cases}
\end{equation*}

An analogue updating scheme is used to construct the vector of latent valuations $\mathbf{V}$. This Gauss-Seidel style updating scheme ensures that, at any point in the iteration scheme, the upper and lower bounds are satisfied for the entire vector of latent utilities and valuations, but it comes at a higher computational burden. The alternative would be to compute upper and lower bounds once in every iteration $k$, using only the last estimates of the latent utilities $\mathbf{U}^{(k-1)}_{ij}$. In what follows, we will omit the index of the current iteration round $k$, and assume that any reference to $\mathbf{U}_{is} = U_i(s)$ is made with respect to the most recent available estimate of $U_i(s)$, either from iteration $k$ or from iteration $k-1$.

Lastly, we will in the following exposition use the order $>$ on the set of ranks, or priorities. As either a rank $rk_t(s)$ or a priority $pr_s(t)$ can take the value $\emptyset$, it is necessary to define the behaviour of this operator with respect to $\emptyset$: we will assume that the statement $a>\emptyset$ is false for all values of $a$, whereas $a\geq\emptyset$ is true if, and only if, $a=\emptyset$. Also, as a convention, the minimum of an empty set returns $\infty$ and the maximum of an empty set returns $-\infty$.

\subsubsection*{Weak truth-telling (WTT)}
Having clarified the notation, we now turn to describe how upper and lower bounds implied by the weak truth-telling assumption (WTT) are constructed. WTT posits that, on the side of the students, any unranked alternative school $s: rk_t(s)=\emptyset$ is worse than any ranked alternative $s'$ with $rk_t(s')\neq \emptyset$. Given latent valuations $\mathbf{U}_{-it}$, and observed ranks $\mathbf{rk}$, the upper bounds for utility $U_t(s)$ can be expressed as follows
\begin{equation*}
	\overline{U}_t(s) = \begin{cases}
		+\infty &  rk_t(s) = 1 \\
		\min_{s' \in \mathcal{L}_t}\{U_t(s'): rk_t(s') < rk_t(s) \} & rk_i(s)>1 \\
		\min_{s' \in \mathcal{L}_t}\{U_t(s')\} & rk_i(s) = \emptyset
	\end{cases} 
\end{equation*}
and the lower bounds are given as
\begin{equation*}
	\underline{U}_t(s) = \begin{cases}
		\max_{s' \in \mathcal{L}_t}\{U_t(s'): rk_t(s') > rk_t(s) \} & rk_t(s)<L_t \\
		\max_{s' \notin \mathcal{L}_t}\{U_t(s')\} &  rk_t(s) = L_t < |S| \\
		-\infty &  rk_t(s) = \emptyset \wedge rk_t(s) = |S| \\
	\end{cases} 
\end{equation*}

In our setting, schools only get to see those students who apply to them and hence, $pr_s(t)=\emptyset$ does not imply that the school $s$ considers student $t$ worse than any or all of their ranked students $t'\in \mathcal{L}_s$ that showed up their application list. Therefore, WTT does not allow us to infer anything about the upper and lower valuation bounds for those students that did not apply at school $s$. Schools are required to prioritize all students that apply to them, but if school $s$ deems student $t \in \mathcal{L}_s$ unacceptable, it assigns $pr_s(t)=\infty$ to that student, which implies that this student $t$ is less preferred than any other ranked student $t' \in\mathcal{L}_t: pr_s(t')<\infty$. This, however, does not allow us to infer anything about how school $s$ priorities student $t$ relative to other students that are equally unacceptable. Hence, the upper bounds for school $s$'s valuation of student $t$, $V_s(t)$, conditional on $\mathbf{V}_{-st}$ and observed priorities $\mathbf{pr}$ are given by
\begin{equation*}
	\overline{V}_s(t) = 
	\begin{cases}
		+\infty 	& pr_s(t) \in \{1, \emptyset\} \\
		\min_{t' \in \mathcal{L}_s}\{V_s(t'): pr_s(t') < pr_s(t) \} & 1<pr_s(t) \leq \overline{pr}_s \\
	\end{cases}
\end{equation*}
and the lower bounds by
\begin{equation*}
	\underline{V}_s(t) = 
	\begin{cases}
		-\infty &  pr_t(s) \in \{ \overline{pr}_s, \emptyset\} \\
		\max_{t' \in \mathcal{L}_s}\{V_s(t'): pr_s(t') > pr_s(t) \} & 1 \leq pr_s(t)< \overline{pr}_s \\
	\end{cases}
\end{equation*}

\subsection*{Undominated Strategies (UNDOM)}
Under undominated strategies (UNDOM), unranked alternatives are not assumed to be worse, from the students' perspective. Therefore, UNDOM imposes fewer restrictions than WTT. Given latent valuations $\mathbf{U}_{-it}$, and observed ranks $\mathbf{rk}$, the upper bounds for utility $U_t(s)$ can be expressed as follows
\begin{equation*}
	\overline{U}_t(s) = \begin{cases}
		+\infty &  rk_t(s) \in \{1, \emptyset \}\\
		\min_{s' \in \mathcal{L}_t}\{U_t(s'): rk_t(s') < rk_t(s) \} & rk_i(s)>1
	\end{cases}
\end{equation*}
and the lower bounds are given as
\begin{equation*}	
	\underline{U}_t(s) = \begin{cases}
		-\infty &  rk_t(s) \in \{L_t, \emptyset \}\\
		\max_{s' \in \mathcal{L}_t}\{U_t(s'): rk_t(s') > rk_t(s) \} & rk_t(s)<L_t
	\end{cases} 
\end{equation*}
Because schools are cannot choose to intentionally not rank a student who applies there, the upper and lower bounds under UNDOM are exactly the same that were derived under WTT.

\subsection*{Stability}
Finally, consider an observed matching $\mu$ where $\mu(s)$ denotes the set of all students that are assigned to school $s$, and $\mu(t)$ denotes the assigned school of student $t$ (a student can only be assigned to one school at once). If student $t$ is unassigned, $\mu(t)=t$. Every school can accommodate at most $q_s$ students, so we define the convenience function 
\[
	\chi(s) = \mathbf{1}\left( |\mu(s)| = q_s \right)
\]
that indicates whether a school is at full capacity or not. Further, define the \emph{feasible set} of student $t$ as the set of schools that do not classify student $t$ as unacceptable or have not ranked student $t$, and that either have some vacant seats, or would favour student $t$ over one of their currently admitted students:
\[
	\mathcal{F}_t = \left\{s \in S: \left( pr_s(t)<\infty \vee pr_s(t) = \emptyset \right) \wedge \left(\neg\chi(s) \vee V_s(t)>\min_{t'\in\mu(s)} V_s(t') \right)\right\}.
\]
This feasible set of student $t$ is unobserved (it is a latent set) because it depends on the latent valuations $\mathbf{V}$. 

We now outline conditions on the valuations and utilities that, if satisfied, guarantee that the observed matching $\mu$ is stable. \cite{Logan2008} have used similar conditions to estimate the parameters of a one-to-one marriage market model, and we adapt their setting to a many-to-one matching market. Before we proceed, we introduce the following assumption:
\begin{assumption}[Non-wastefulness]
	The matching $\mu$ is non-wasteful: all schools operate at full capacity ($|\mu(s)|=q_s$) or no student is unmatched ($\mu(t)\neq t$).
\end{assumption}

This assumption is convenient in order to ensure that one can always find utilities and valuations that are consistent with a stable matching and it is also the approach that was taken by \citet[p.2732]{Sorensen2007}. Without this assumption, it would be necessary to specify outside options for agents, which would complicate the analysis, but pose no substantial challenges to it. Conditional on the latent set $\mathcal{F}_t$, stability requires that student $t$'s utility for any school in this latent set be less than that of her currently assigned school. Therefore, the upper bound for a student $t$'s valuation of school $s$ is given by
\begin{equation*}
	\overline{U}_t(s) = 
	\begin{cases}
		U_t(\mu(t)) & \mu(t) \notin \{ s, t\} \wedge s \in \mathcal{F}_t \\
		+ \infty  & \text{else}
	\end{cases}
\end{equation*}
Similarly, the lower bounds are given by
\begin{equation*}
	\underline{U}_t(s) = 
	\begin{cases}
		\max_{s'\in\mathcal{F}_t\setminus \{s\}}\{ U_t(s')\} & \mu(t) = s \\
		- \infty & \text{else}
	\end{cases}
\end{equation*}
Note that we assume that all schools are acceptable to the student. This implies that if student $t$ is unmatched ($\mu(t)=t$), then we cannot bound her utility for any school, be it in her feasible set or not. Instead, stability requires that her feasible set be empty. This, places bounds on the schools' valuations for student $t$ which will be described shortly.

We define school $s$'s feasible set as the set of students who are acceptable to school $s$, and who would prefer going to school $s$ than to their current school, or are unassigned under the matching $\mu$. We chose to include only students that are acceptable to school $s$ in this set because it simplifies the notation below. Thus, the feasible set is given by
\[
	\mathcal{F}_s = \left\{t \in T: pr_s(t)<\infty \wedge \large( U_t(s)>U_t(\mu(s)) \vee \mu(t)=t \large) \right\}.
\]
Again, this is a latent set that depends on the latent student utilities $\mathbf{U}$.
Then, upper and lower bounds of school $s$'s valuation of student $t$ can  be constructed if school $s$ is at full capacity, i.e. if $\chi(s)$ is true:
\begin{equation*}
	\overline{V}_s(t) = 
	\begin{cases}
		\min_{t'\in\mu(s)}\{V_s(t')\} & \chi(s) \wedge t\notin\mu(s) \wedge t\in \mathcal{F}_s \\
		+\infty & \text{else.} 
	\end{cases}
\end{equation*}
Similarly, the lower bounds are given by
\begin{equation*}
	\underline{V}_s(t) = 
	\begin{cases}
		\max_{t' \in \mathcal{F}_s \setminus \mu(s)}\{ V_s(t')\} & \chi(s) \wedge t \in \mu(s) \\
		- \infty & \text{else.}
	\end{cases}
\end{equation*}
In general, the upper and lower bounds on utilities and valuations are interdependent, and are not unique. 

\subsection*{Combination of UNDOM and Stability}
The combination of the two assumptions that students and schools play undominated strategies, and that the assignment is stable, allows us to tighten the bounds. For instance, let $[ \underline{U}^{rk}_t(s), \overline{U}^{rk}_t(s) ]$ be the bound that is imposed by the assumption of undominated strategies on the valuation $U_t(s)$, and let $[ \underline{U}^{\mu}_t(s), \overline{U}^{\mu}_t(s) ]$ be the bounds that follow from the requirement that the observed matching $\mu$ be stable. An obvious way to combine these two bounds is to simply set
\begin{align*}
	\underline{U}_t(s) &= \max\left\{ \underline{U}^{rk}_t(s),\underline{U}^{\mu}_t(s) \right\} \\
	\overline{U}_t(s)  &= \min\left\{ \overline{U}^{rk}_t(s), \overline{U}^{\mu}_t(s) \right\}
\end{align*}
and for $V_t(s)$ in an analogous manner. Now, the question is whether so truncation intervals that are constructed in this way are non-empty, \emph{i.e.} whether $\underline{U}_t(s) \leq \overline{U}_t(s)$. We will show that, for any given stable matching $\mu$, observed priories $\mathbf{pr}$ and preference ranks $\mathbf{rk}$, there is at least one set of preferences $\mathbf{U}$ and valuations $\mathbf{V}$ such that the assumptions UNDOM and stability are satisfied:

\begin{lemma} \label{thm:lemma1}
	Consider any given non-wasteful stable matching $\mu$ that is derived from the observed partial rankings $\mathbf{rk}$ and priority structures $\mathbf{pr}$. Then, there exists a complete preference structure $\succ$ and priority ordering $\rhd$ such that
	\begin{enumerate}
		\item $\succ$ and $\rhd$ are consistent with $\mathbf{rk}$ and $\mathbf{pr}$, respectively and
		\item $\mu$ is stable under $\succ$ and $\rhd$.
	\end{enumerate}
	Thus, the set of utilities $\mathbf{U}$ and valuations $\mathbf{V}$ that satisfies the bounds imposed by \emph{UNDOM} and \emph{stability} is non-empty for any observed matching $\mu$.
\end{lemma}
\begin{proof}
The first point is obvious: fix an arbitrary set of utility numbers $\{U_t(s): s \in \mathcal{L}_t \}$ and valuation numbers $\{V_s(t): t \in \mathcal{L}_s\}$ that respect the ordering implied by the observed ranks $\mathbf{rk}$ and priorities $\mathbf{pr}$; there will always be such numbers. For the second point, note that we can equivalently express students' preferences and schools' priorities in terms of their partial rank and priority order lists, or in terms of their utilities and valuations. As the observed matching $\mu$ is stable under the former, it must also be stable under the latter representation and so, any set of utility and priority numbers that respects the bounds imposed by \emph{UNDOM} also satisfies the bounds that are imposed by \emph{stability}. Next, we need to show that there are always utility and valuation numbers for the remaining non-ranked pairs such that there are no \emph{blocking pairs}. Consider any such pair $t,s$ such that $s\notin \mathcal{L}_t$. Under the student-proposing deferred acceptance mechanism, no student can be assigned to a school that she did not include in her stated rank order list $\mathbf{rk}_t$, and hence $s \neq \mu(t)$. Then there are four remaining cases to consider:
\begin{description}
	\item[Case 1] Student $t$ is not unmatched, and school $s$ is at full capacity, i.e. $\mu(t) \neq t$ and $|\mu(s)|=q_s$.
	Stability is satisfied if $U_t(s)<U_t(\mu(t))$ or $V_s(t) < \min_{t' \in \mu(s)} V_s(t')$, or both.
	\item[Case 2] Student $t$ is not unmatched, and school $s$ has spare capacity.
	Stability is satisfied for all $U_t(s)<U_t(\mu(t))$ and $V_s(t) \in \mathbb{R}$.
	\item[Case 3] Student $t$ is unmatched, and school $s$ is at full capacity.
	Stability is satisfied for all $V_s(t) < \min_{t' \in \mu(s)} V_s(t')$ and $U_t(s) \in \mathbb{R}$.
	\item[Case 4] Student $t$ is unmatched, and school $s$ has spare capacity. This case is ruled out under the assumption that $\mu$ is non-wasteful.
\end{description}
Hence, if the matching $\mu$ is non-wasteful, it will always be possible to find utilities and valuations that respect both the partially observed rank and priority structures, and stability properties.
\end{proof}

However, we observe in our dataset that roughly ten percent of all students are not assigned to a school in the first matching round (\emph{c.f.} Table \ref{tab:students_summary}) so that the allocation is not non-wasteful in the sense outlined above, and the last case of the proof does not go through.\footnote{In the Hungarian school choice system, the main matching round is followed by a subsequent round in which any unmatched students are assigned to the closest feasible school.} This could appear to be a problem for our estimation approach, because the existence of an unmatched student $t$ and a school that has spare capacity $s$ necessarily leads to instability in our estimation approach. The solution would be to endogenously determine ``latent'' unacceptable students, to exclude such students from the sample, or to artificially label them as being ``unacceptable'', neither approach of which is very attractive. Instead, we note that if there exists a student $t$ who is unmatched, and a school $s$ with spare capacity, it must either be that $t$ did not apply to $s$, in which case the bounds on the latent utility and on the latent valuation are $\pm \infty$, or that student $t$ did rank school $s$, but school $s$ ranked student $t$ as unacceptable, in which case the valuation and utility bounds are well defined. Only the former case represents a case of true instability, whereas the latter case is well covered by our estimation approach. Most importantly, if such a case of true instability should occur, it will not affect the parameter estimates in either direction, because the utilities and valuations are not restricted and simply add some white noise to the parameter updates.

\clearpage
\section{Posterior distributions}
\label{sec:posterior}

\noindent The Bayesian estimator uses the data augmentation approach \citep[proposed by][]{albert1993bayesian} that treats the latent valuation variables as nuisance parameters. This section describes the components conditionals of the Gibbs sampler that is used to sample from the posterior distribution of the parameters of interest $\beta$ and $\gamma$, $p(\beta,\gamma|data)$ where the $data$ are observed co-variates, and possibly rank and priority structures or matching information.

\subsubsection*{Conditional distribution of utilities and valuations}
Recall that it is assumed that $\epsilon_{ts},\eta_st \sim N(0,1)$, as is customary and necessary in the discrete choice literature. Then, the component conditionals for the unobserved latent utilities and valuations are given by
\begin{footnotesize}
\begin{align*}
	p(U_t(s)|\beta,\gamma, \mathbf{U}_{-ts},\mathbf{V}, data) &\propto \exp\left\{ \frac{-(U_t(s) - \mathbf{X}_{ts}\beta)^2}{2} \right\} \mathbf{1}(U_t(s) \in [ \underline{U}_t(s) \overline{U}_t(s) ] ) \\
	p(V_s(t)|\beta,\gamma, \mathbf{V}_{-st},\mathbf{U}, data) &\propto  \exp\left\{ \frac{- (V_s(t) - \mathbf{W}_{st}\gamma)^2}{2} \right\} \mathbf{1}(V_s(t) \in [\underline{V}_s(t), \overline{V}_s(t)])
\end{align*}
\end{footnotesize}
Note that, although the error terms are uncorrelated and independent across alternatives, the utilities are not because their truncation intervals are endogenously determined. For example, if we observe a student's ranking across three different schools $A$, $B$, and $C$ such that $rk_t(A)<rk_t(B)<rk_t(C)$, this implies that $U_t(A)>U_t(B)>U_t(C)$. Therefore, the distribution of utilities across schools is not iid normal, but rather a multivariate normal distribution subject to a system of linear inequality constraints. Commonly known techniques for sampling from these distributions with linear constraints are rather slow when the number of alternatives is very large, as is the case in our setting with thousands of students, and hundreds of schools.\footnote{The function \texttt{rtmvnorm2} in the R package tmvtnorm (\url{https://cran.r-project.org/package=tmvtnorm}, version 1.4-10) does provide such a a method} Instead, we embed the sampling from this intractable distribution into our Gibbs sampler. However, we found that this procedure is rather slow to converge, and also exhibits very strong serial correlation so that a sufficiently large number of Gibbs samples must be drawn.

\subsubsection*{Conditional distribution of utility and valuation parameters}
We assume a flat prior for the structural parameters $\beta$ and $\gamma$ which, together with the assumption that the error terms have unit variance, implies that the posteriors of $\beta$ and $\gamma$ follow a normal distribution.\footnote{A flat, uninformative, or vague prior assumes that the density is constant on the entire parameter space. As the parameter space is unbounded, the prior is ``improper'' because it does not integrate to any positive constant. Nonetheless, the posterior density of a linear model is well defined even though the prior is improper \citep[p.120]{Lancaster2004}.}. Also, we note that the scale and the location of the utilities and valuations are not identified, as in any discrete choice model. Our assumption that the idiosyncratic errors have unit variance pins down the scale of utility, and the assumption that these errors are zero in expectation pins down the location of utilities. 
Hence the component conditional distribution of the utility parameter is given by
\[
	p(\beta|\gamma,\mathbf{U},\mathbf{V},data) = p(\beta|\mathbf{U},data) = N\left(b, (\mathbf{X}'\mathbf{X})^{-1}\right)
\]
for $b=(\mathbf{X}'\mathbf{X})^{-1} \mathbf{X}' \mathbf{U}$, and similarly, the conditional component for the priority parameter $\gamma$ reads
\[
	p(\gamma|\beta,\mathbf{U},\mathbf{V},data) = p(\gamma|\mathbf{V},data) = N\left(g, (\mathbf{W}'\mathbf{W})^{-1}\right)
\]
for $g=(\mathbf{W}'\mathbf{W})^{-1} \mathbf{W}' \mathbf{V}$.

\clearpage
\section{Monte Carlo results}
\label{sec:appendix:MCsim}

Monte Carlo simulations provide evidence that our estimation methodology can recover the true underlying preference parameters, even if students behave strategically and submit incomplete ROLs. Our data generating process encompasses students who submit strategic reports so that the observed rank order lists are incomplete. In particular, many reports suffer from ``skipping at the top'' and ``truncation at the bottom'' problems already mentioned. Still, our estimation procedure ``\emph{stability} + \emph{UNDOM}'' recovers the true and unbiased preference parameters. Moreover, the estimation error is comparable to the benchmark scenario with completely observed ROLs. 

We further demonstrate that, in the estimation, the share of students affected by a multiplicity of stable matchings converges to zero as the number of students per school grows large. We also provide evidence that our estimation methodology is robust to the strategic submission of ROLs. Our methodology yields unbiased estimates no matter how many submitted ROLs violate the \emph{WTT} assumption, while alternative estimation procedures suffer from a bias that is increasing in the share of ROLs which violate the \emph{WTT} assumption. These results are presented towards the end of this subsection.

The data generating process of the Monte Carlo study is borrowed from \cite{fack2019beyond}.\footnote{Their data generating process is described, and the code is made available, in their online appendix. We  use normally distributed errors on both sides of the market instead of type-I extreme value distributed errors.} However, we depart from their assumption that a student's priority at each school is known to the econometrician. Instead, we only use the ranking of students who actually applied for a particular school. We consider markets with 
200 students, six schools and capacities of $q \in \{0.95, 1.20\}$ with total capacities of $q \cdot 200$ seats. So we have one scenario with slight excess demand of 5\%, as in the original analysis in \citet{fack2019beyond} and one scenario with an under-demand of 20\%, resembling the Budapest context.\footnote{The ACCEPTED specification in \ref{sec:Appendix_D} results in an under-demand of 21\%.} Students' utility over schools is given by
\[
U_{t}(s) = \delta_s - d_{ts} + 3 \cdot (a_t \cdot \bar{a}_s) + \epsilon_{ts}
\]
where $\delta_s$ is a school fixed effect, $d_{ts}$ is the distance from student $t$ to school $s$, $a_t$ is the students' grade and $\bar{a}_s$ is the average grade of all students at school $s$ (or put differently, the schools' academic quality). Hence, the true preference parameter in the data generating process is a vector $\beta_0 = (1,-1,3)'$. 
$\epsilon_{ts}$ follows a standard normal distribution. For the exposition, we assume that $\delta_s$ is known to the econometrician and therefore enters the estimation as an additional co-variate. The schools' valuation over students (which translates into the students' priorities) is given by
\[
V_s(t) = a_t + \eta_{st}
\]
where $\eta_{is}$ is also standard normally distributed. Here, the true priority parameter $\gamma_0$ is a scalar equal to one. We subsume all preference and priority parameters as $\theta_0 = (\beta_0',\gamma_0)'$.

In the market, students choose their optimal application portfolio, given their equilibrium beliefs about admission probabilities, and a small application cost. This leads some students to skip seemingly unattainable top choices, or to truncate their ROL at the bottom. As a result, the submitted ROLs are likely to violate the assumption of WTT. Based on the simulated submitted ROLs, students and school seats are matched according to the SOSM. We refer the reader to the online appendix of Fack et al. for further details. 

For our Monte Carlo study, we simulated one hundred independent realisations of these markets. In the simulated markets,
a share of 0.69 of the submitted rank order lists satisfied WTT across all simulations.
For every sample $k$, we estimated students' preferences over schools ($\hat{\beta}_k$), and schools' priorities over students ($\hat{\gamma}_k$) using the data augmentation approach described above. We used the following assumptions to compute the truncation intervals based on the submitted ROLs: i) WTT, ii) STAB, iii)	UNDOM, and iv) STAB + UNDOM. As a benchmark, we estimated the model under the assumption of undominated strategies based on true and complete ROLs. We let the Gibbs sampler run for 20,000 iterations, with a burn-in period of 10,000 iterations. To reduce the parameter estimates' serial correlation, we used only every fifth sample, and discarded the rest.

Figures \ref{fig:MCsim} and \ref{fig:MCsim_underdemanded} show box plots of the estimation errors ($\hat{\theta}_k - \theta_0$) across the one hundred realised data sets, for different estimation approaches. The analysis in Figure \ref{fig:MCsim} is based on markets with 190 seats and 200 students, which corresponds to an over-demand of 5\%, as in the original analysis in \citet{fack2019beyond}. In Figure \ref{fig:MCsim_underdemanded}, we use an under-demand of 20\%, which is in line with the observed excess capacity in the Budapest data. The results in both experiments are very similar and we therefore focus on the interpretation of the over-demand case in the following. Table \ref{tab:MCsim200} shows the corresponding mean squared error and bias statistics.

The first three panels of Figure \ref{fig:MCsim} depict the distribution of the estimation errors of students' preference parameters ($\hat{\beta}_k-\beta_0$). As expected, the benchmark case where the complete ROLs are known on both sides allows us to identify the parameters very precisely. Furthermore, the estimates for student preferences that are derived under the assumption of weak truth-telling are biased. This too is to be expected because the assumption of weak truth-telling does not hold in the data generating process.

\begin{figure}[h!]
	\centering
	\includegraphics[width=\textwidth]{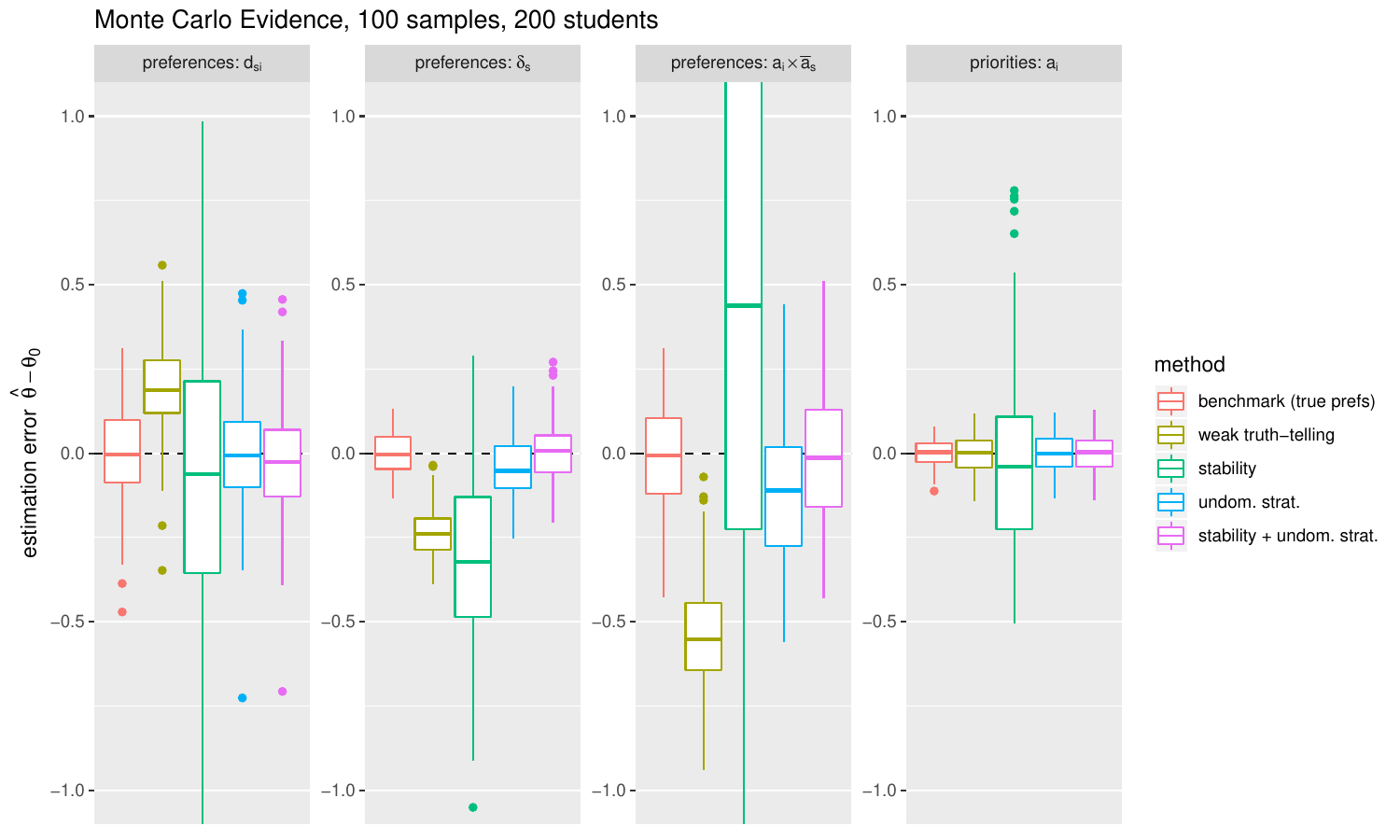}
	\caption[Monte Carlo results]{Box plots of the distributions of estimation errors across one hundred simulated markets (six schools with 190 seats and 200 students).}
	\label{fig:MCsim}
\end{figure}

\begin{figure}[h!]
	\centering
	\includegraphics[width=\textwidth]{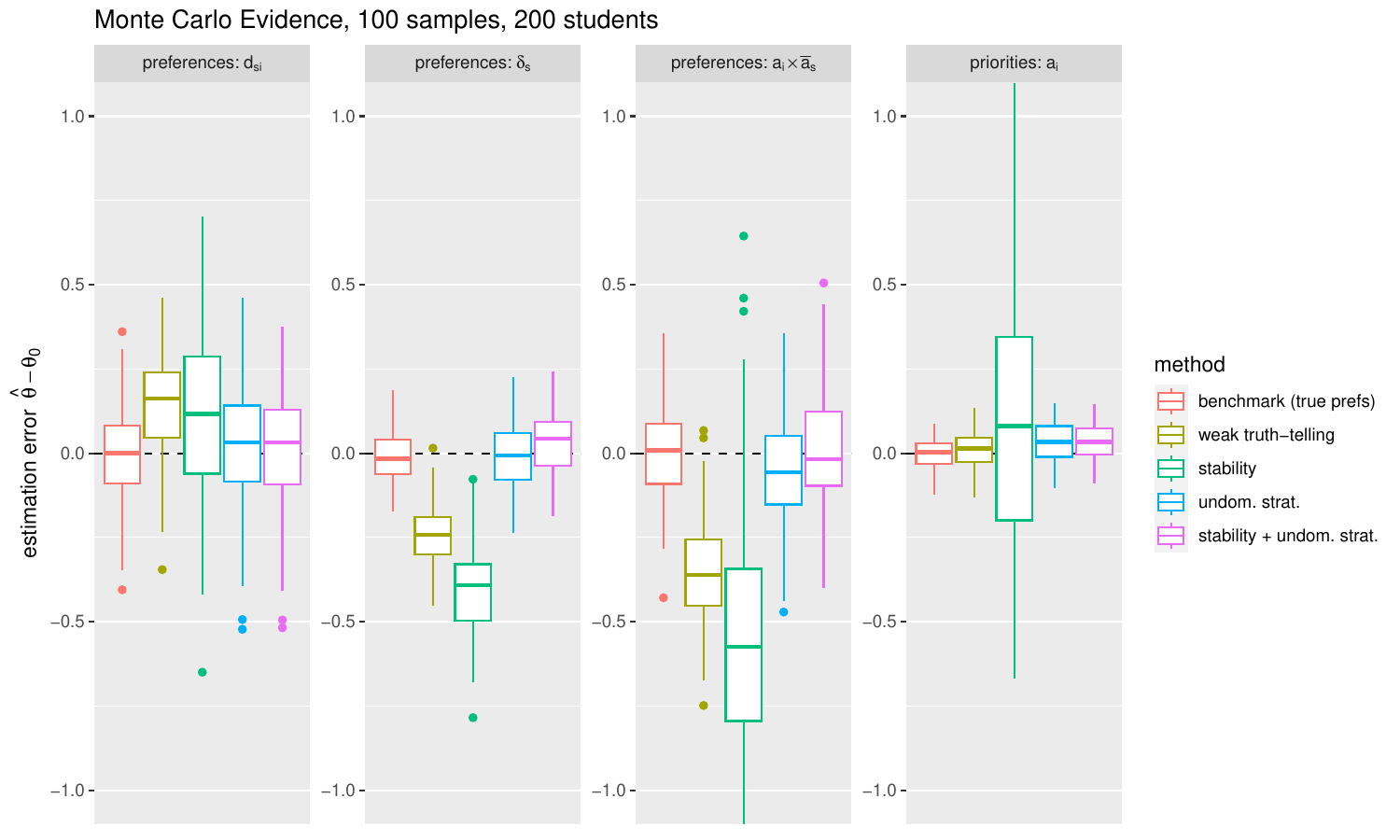}
	\caption[Monte Carlo results]{Box plots of the distributions of estimation errors across one hundred simulated markets (six schools with 240 seats and 200 students).}
	\label{fig:MCsim_underdemanded}
\end{figure}

When the estimation is conducted using only the stability assumption, the results are noisy and biased. Under the stability assumption, the best estimation results are those for the coefficient on travel distances $d_{ts}$, but worse results are obtained for the schools' quality $\delta_s$ and for the interaction parameter.
This is in line with the previous literature on stability based estimators of preferences in small two-sided matching markets. That literature has reached a consensus that the preference parameters are only identified under certain assumptions on the observable characteristics \citep[pp.158-168]{Weldon2016} or certain preference structures such as perfectly aligned preferences \citep{diamond2017latent}, and may not be identified at all in other circumstances. Note that this is not necessarily at odds with \cite{fack2019beyond} who argue that a stability based estimator can be used to point-identify preference parameters, for their stability-based estimator is based on the assumption that students' feasible choice sets are known, whereas we assume that this is not the case. 
\begin{table}[h!]
	\centering
	\caption{MSE and bias statistics on the Monte Carlo Simulation.}
	\label{tab:MCsim200}
	\begin{footnotesize}
		\begin{tabular}{lrrrr}
			\toprule
			Method & \multicolumn{3}{c}{Preferences} & \multicolumn{1}{c}{Priorities} \\
			\cmidrule{2-4}
			& \multicolumn{1}{c}{$d_{is}$} & \multicolumn{1}{c}{$\delta_s$} & \multicolumn{1}{c}{$a_i\cdot\bar{a}_s$} & \multicolumn{1}{c}{$a_i$} \\
			\midrule
			\it Panel A. Mean squared error (MSE) &&&\\
			benchmark (true prefs.) & 0.0187 & 0.0038 & 0.0227 & 0.0016 \\
			WTT & 0.0598 & 0.0581 & 0.3243 & 0.0032 \\
			STAB & 0.2903 & 0.1597 & 4.8612 & 0.0788 \\
			UNDOM & 0.0338 & 0.0103 & 0.0539 & 0.0030 \\
			STAB + UNDOM & 0.0323 & 0.0088 & 0.0448 & 0.0030 \\
			\midrule	\it Panel B. Bias &&&\\
			benchmark (true prefs.) & -0.0066 & -0.0023 & -0.0027 & -0.0009 \\
			WTT & 0.1937 & -0.2302 & -0.5425 & 0.0004 \\
			STAB & -0.1273 & -0.3132 & 0.9949 & -0.0204 \\
			UNDOM & 0.0055 & -0.0421 & -0.1179 & 0.0001 \\
			STAB + UNDOM & -0.0219 & 0.0134 & -0.0183 & 0.0026 \\
			\bottomrule
		\end{tabular}%
	\end{footnotesize}
\end{table}

The estimates that are derived under undominated strategies are much more precise, but also appear to suffer from a slight bias, which could be a result of the small sample size. Finally, when we combine stability and undominated strategies, our estimates are virtually indistinguishable from the benchmark estimates that are derived using the true and complete ROLs.

\paragraph{Set of stable matchings}
The Monte-Carlo simulations allow us to test for the size and convergence of set of stable matchings. In each of the 100 sample markets, we obtain for iterations 2, 5, 10, 20, ... up to 2000 all stable matchings that would result from the latent match valuations of the Gibbs draws for students and schools (\textbf{U} and \textbf{V}).

Figure \ref{fig:equilibria} plots the average number of stable matchings by iteration for markets with a fixed number of six schools and 100, 500 and 2000 students, respectively. The number of stable matchings quickly converges to 2 with increasing iterations of the Gibbs sampler and is invariant with respect to the number of students. The latter result suggests that the share of students affected by the multiplicity of stable matchings converges to zero as the number of students per school grows large. This is confirmed in 
Figure \ref{fig:SOSMvsCOSM}, where the average percentage of students with different student-optimal stable matching (SOSM) and college-optimal stable matching (COSM) is plotted by iteration. 

These results confirm the applicability of our estimator in the context of large school choice markets. If the share of students who are matched to the same school in SOSM and COSM converges to one, 
then there is no need to impose additional uniqueness constraints in the estimation, which would be computationally costly.

\begin{figure}
	\begin{subfigure}[c]{0.46\textwidth}
		\begin{center}
			\includegraphics[width=\textwidth]{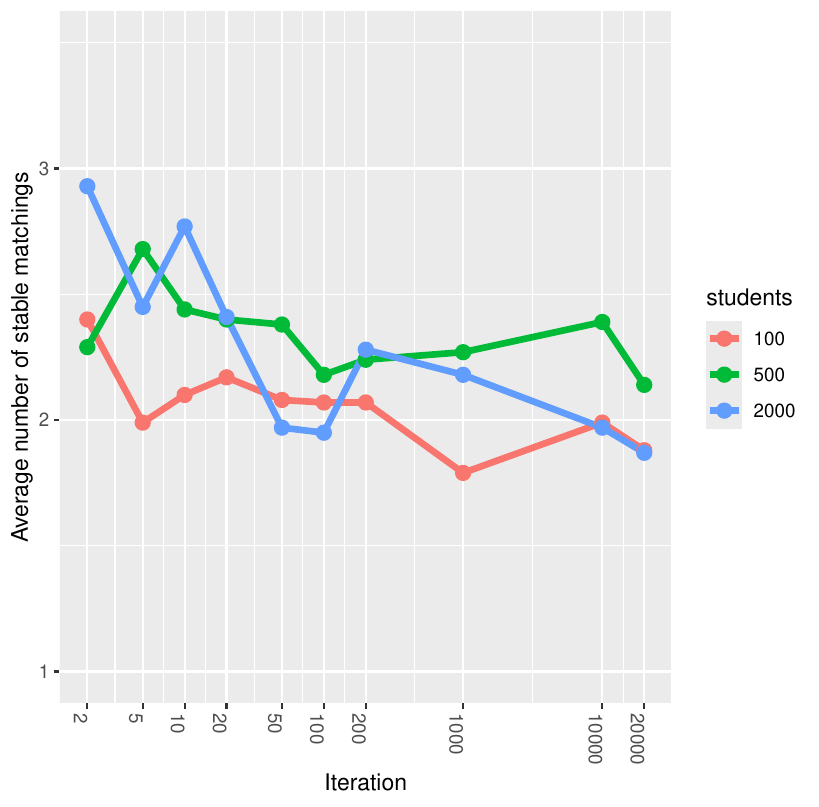}
		\end{center}
		\caption{Average number of stable matchings by Gibbs sampling iteration is invariant with respect to the number of students.}
		\label{fig:equilibria}
	\end{subfigure}
	\quad \quad
	\begin{subfigure}[c]{0.46\textwidth}
		\begin{center}
			\includegraphics[width=\textwidth]{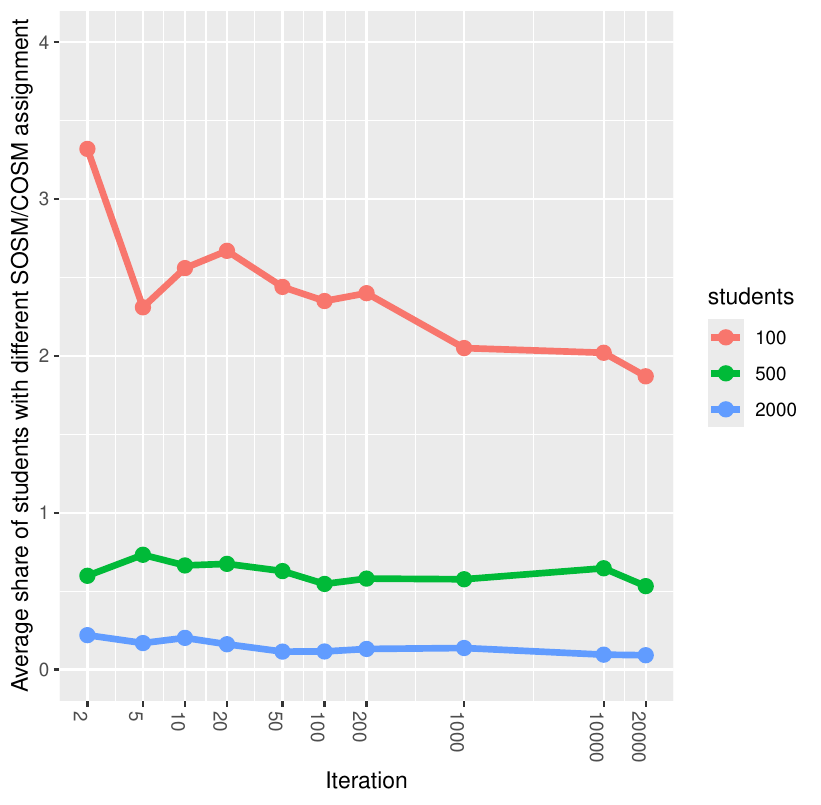}
		\end{center}
		\caption{Average share of students with different student- and college-optimal assignments by Gibbs sampling iteration converges to zero.}
		\label{fig:SOSMvsCOSM}
	\end{subfigure}
	\caption{Convergence results of the share of students involved in multiple stable matchings and the number of stable matchings. Based on markets with 100, 500 and 2000 students and fixed number of schools.}
	\label{fig:uniqueness}
\end{figure}

\clearpage

\begin{figure}
	\begin{center}
		\includegraphics[width=\textwidth]{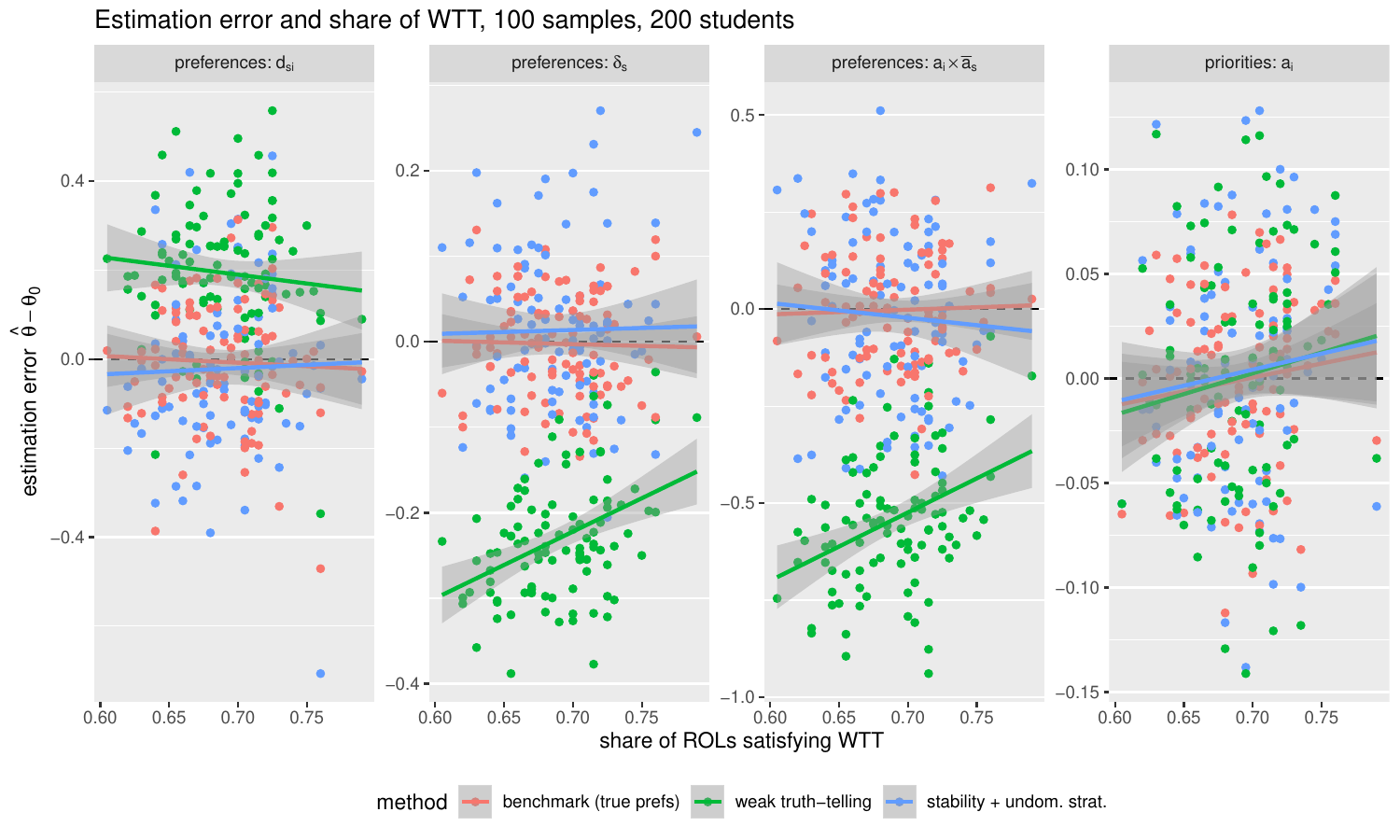}
	\end{center}
	\caption[Estimation errors and WTT]{Dependence of the estimation error in different specifications on the share of submitted ROLs that satisfy the WTT assumption. Every dot represents one parameter estimate in one sample market. One hundred simulated markets, six schools with 190 seats, and 200 students.}
	\label{fig:WTTshare}
\end{figure}

\begin{table}[ht!]
	\begin{center}
		\caption[Robustness to WTT violations]{Robustness of estimation procedures to violations of WTT. }
		\label{tab:WTTshare}
		\begin{footnotesize}
			\begin{threeparttable}
				\begin{tabular}{lr@{}lr@{}lr@{}lr@{}l}
					\toprule
					Method	& \multicolumn{6}{c}{Preferences}               & \multicolumn{2}{c}{Priorities} \\
					\cmidrule{2-6}
					& $d_{is}$ &       & $\delta_s$ &       & $a_i\cdot\bar{a}_s$ &       & $a_i$ &  \\
					\midrule
					benchmark (true prefs.) & -0.12 &       & -0.04 &       & 0.12 &       & 0.13 &  \\
					WTT & -0.40 &       & 0.78 & \textsuperscript{***} & 1.76 & \textsuperscript{***} & 0.20 &  \\
					STAB & 0.49 &       & 0.14 &       & -10.34 & \textsuperscript{*} & 1.67 & \textsuperscript{**} \\
					UNDOM & 0.16 &       & 0.08 &       & -0.31 &       & 0.18 &  \\
					STAB + UNDOM & 0.14 &       & 0.05 &       & -0.38 &       & 0.15 &  \\
					\bottomrule
				\end{tabular}%
				\begin{tablenotes}
					\tiny \item p-values indicated by ${}^*<0.1; {}^{**}<0.05; {}^{***}<0.01$. The table shows the coefficients from separate linear regressions of the estimation error on the share of ROLs satisfying WTT, by estimation approach and parameter. For an estimation approach to be robust to violations of the WTT assumption, the estimation error should not depend on the share of ROLs satisfying WTT.
					
				\end{tablenotes}
			\end{threeparttable}
			
		\end{footnotesize}
	\end{center}
\end{table}

\paragraph{Strategic reporting}
To confirm that the combination of stability and undominated strategies is indeed able to correct the estimation bias due to strategic reporting, we compute the share of submitted ROLs satisfying WTT in each sample market, and plot this share against the parameter estimate in that sample. This is done in Figure \ref{fig:WTTshare}. Each dot in that Figure represents one parameter estimate in one single simulated market. The lines represent the least square estimates for the relation between the share of ROLs that satisfy WTT and the estimation error. 
The corresponding regression coefficients are shown in Table \ref{tab:WTTshare} and asterisks indicate their significance.
The leftmost three panels of that Figure show that the estimation error for students' utility parameters under the WTT assumption decreases in absolute terms as the share of submitted ROLs satisfying WTT increases (green line). On the other hand, the benchmark estimates and the estimates under stability and undominated strategies are not dependent on the share of ROLs that satisfy WTT. For schools' priority parameters, there is no significant relation between either of the estimates and the WTT share, although the point estimates are weakly positive. We conclude that the proposed estimation approach that relies on a combination of undominated strategies and stability is robust to the strategic submission of preference lists.

\clearpage

\section{Data sources}
\label{sec:Appendix_Data}
As the data that we use contain sensitive information, we had no direct access to it but instead submitted our code to the Hungarian Education Authority (HEA) who executed it on their local computes. In order to develop our estimation routines, we were able to use an example dataset that closely resembled the actual data structure. This appendix is intended to provide some more information on the construction of our working data set.

Table \ref{tab:NABC_summary_stats} shows summary statistics of the student-level NABC data. Most students are fifteen years old at the time of the NABC test (in 2015). The NABC scores in Hungarian and mathematics are the results of a standardized test procedure. The socio-economic status (SES) is a composite measure that is based on responses given by students in an accompanying survey, so that this variable has more missing data. Also, the grade average is based on student's own responses and may thus be biased. Therefore, we use the NABC scores as a proxy for student's academic ability.

\begin{table}[!h]
	\begin{center}
			\caption{Summary statistics of the original NABC (2015) data}
		\label{tab:NABC_summary_stats}
	\begin{footnotesize}
\begin{tabular}{@{\extracolsep{5pt}} lccccc} 

\toprule
	& Mean & SD & Min & Max & N \\ 
\midrule
 \multicolumn{6}{l}{\it{	Panel A. Entire country} }\\ 
	Birth year & 2000.1 & 0.58216 & 1996 & 2002 & 88959 \\ 
	Female & 0.494 & 0.5 & 0 & 1 & 88967 \\ 
	Last grade average & 3.9837 & 0.7668 & 1 & 5 & 60843 \\ 
	NABC score Hungarian & 1559.9 & 202.36 & 820.97 & 2199.2 & 82237 \\ 
	NABC score math & 1612.1 & 196.5 & 907.81 & 2307.3 & 82176 \\ 
	Socioeconomic status (SES) & -0.0226 & 1.01 & -3.15 & 1.88 & 64971 \\ 
	 & & & & & \\ 
	 \multicolumn{6}{l}{\it	Panel B. Budapest  }\\ 
	Birth year & 2000.1 & 0.55451 & 1996 & 2002 & 13609 \\ 
	Female & 0.497 & 0.5 & 0 & 1 & 13611 \\ 
	Last grade average & 4.1929 & 0.65506 & 1 & 5 & 8392 \\ 
	NABC score Hungarian & 1634.9 & 190.44 & 820.97 & 2199.2 & 12480 \\ 
	NABC score math & 1685.7 & 190.24 & 947.4 & 2307.3 & 12467 \\ 
	Socioeconomic status (SES) & 0.616 & 0.88 & -3.15 & 1.88 & 9029 \\ 
\bottomrule
\end{tabular} 

	\end{footnotesize}
	\end{center}

\end{table}

Table \ref{tab:gender_differences} shows that there are significant differences in test outcomes and between male and female students. Female students perform much better in Hungarian on average (almost one third of a standard deviation), whereas male students perform better in math on average (one tenth of a standard deviation). Also, female students obtain a slightly better SES index (five percent of a standard deviation) but notice that the SES index is based on self reporting, so it could be due to different reporting behaviour.
In all cases, the differences in means are significant at the one percent level.

\begin{table}[!h]
\begin{center}
\caption[Gender differences in test outcomes.]{Gender differences in test outcomes. Raw NABC (2015) data. Two-sample t-test with equal variance.}
\label{tab:gender_differences}
\begin{footnotesize}
\begin{tabular}{lrrrrrrr}
\toprule
      & \multicolumn{3}{c}{mean} &       & \multicolumn{3}{c}{t-test} \\
\cmidrule{2-4}\cmidrule{6-8}Variable & all   & male  & female &       & diff. & t     & p \\
\midrule
NABC score Hungarian & 1,560 & 1,538 & 1,583 &       & -45.04 & -32.109 & $<$0.001 \\
NABC score math & 1,612 & 1,621 & 1,603 &       & 18.89 & 13.793 & $<$0.001 \\
Socioeconomic status (SES) & -0.023 & -0.037 & -0.008 &       & -0.029 & -3.644 & $<$0.001 \\
\bottomrule
\end{tabular}%
\end{footnotesize}
\end{center}

\end{table}

Table \ref{tab:KIFIR_summary_stats} shows key statistics of the nation-wide matching scheme. The data comprises almost four hundred thousand applications from almost ninety thousand students to over six thousand school programs. Each record corresponds to the application of a student to a school and contains the ranks $rk_t(s)$ and $rk_s(t)$, an indicator whether the school finds the student acceptable, and a match indicator.  On average, each student applies to 4.5 school programs, or to 2.8 different schools. Almost 95\% of all students are assigned to a school, of which three quarters are eventually assigned to their top choice program.\footnote{The admission system ensures that any students who are unmatched at the end of the main matching round are assigned to the nearest school which still has free capacity.} We link this data to a school survey in order to obtain the precise location of each school, and the school's district.

\begin{table}[!h]
	\begin{center}
	\caption{Summary statistics of the original application data (KIFIR)}
\label{tab:KIFIR_summary_stats}
\begin{tabular}{lr} 
\toprule 
Variable & Number\\
\midrule
	students & $88,401$ \\ 
	school programs & $6,181$ \\ 
	schools (OMid-telephely-tipus) & $1,793$ \\ 
	student-school applications & $395,222$ \\ 
	length of submitted ROL (school programs) & $4.471$ \\ 
	length of submitted ROL (schools) & $3.002$ \\ 
	assigned students& $83,482$ \\ 
share top choice & $0.759$ \\ 
average match rank & $1.486$ \\ 
\bottomrule
\end{tabular} 

	\end{center}

\end{table}

The HEA used a confidential concordance table to link records from the KIFIR and NABC datasets. As described in the main text, we restricted the linked sample to students who applied to at least one school from Budapest, which leaves us with 10,880 students.
As Table \ref{tab:NABC_summary_stats} shows, the NABC scores and, in particular, the SES are missing for a substantial share of our sample. Because the computation of the student-optimal stable matching depends on the composition of the student sample, we were reluctant to drop records with missing data, as this would have left us with rather few complete records. Instead, we opted for a data imputation approach and used the R package \texttt{mice} to construct a complete dataset. Missing variables were imputed using predictive mean matching, where missing values are replaced by actual values from other records that resemble the incomplete record, conditional on other observed characteristics. As predictors, we used an extended set of variables that included also some results from the 2017 NABC round (where available), and further student level variables that are not shown  here. This procedure is repeated a few times, until the imputed values converge in expectation. 

The following Table \ref{tab:imputation_stats} shows details of the imputation procedure. It can be seen that the imputed mean of the variables referring to academic ability is lower than in the original data. Our imputation procedure naturally introduces measurement error into the data, which, in a classical regression framework, should lead to estimated coefficients that biased towards zero. We expect that this is also true for our estimation procedure which comprises a data augmentation approach with a linear regression. Nevertheless, it is our opinion that the drawbacks of using an imputed data set are greatly outweighed by the benefit of having a comprehensive set of students for the estimation procedure (which relies on stability considerations, and thus, on the entirety of the student population) and for the counter-factual matches (which are more directly dependent on the entire student population).

\begin{table}[!h]
	\begin{center}
	\caption{Results of the imputation procedure, using predictive  mean matching and ten iterations. The p-value is computed for a two-sided t-test with unequal variances.}
\label{tab:imputation_stats}
	\begin{tiny}
\begin{tabular}{@{\extracolsep{5pt}} lccccccc} 
\toprule

	& N & Mean & SD & N.imp & Mean.imp & SD.imp & Pval \\ 
\midrule
	Birth year & $10,879$ & $2,000.06$ & $0.55$ & 1 & $2,001$ & $$ & $0.09$ \\ 
Sex (1=female,2=male) & $10,880$ & $1.50$ & $0.50$ & 0 & $$ & $$ & $$ \\ 
Last grade average & $6,598$ & $4.12$ & $0.68$ & 4282 & $3.97$ & $0.70$ & $0$ \\ 
NABC score Hungarian & $9,934$ & $1,659.63$ & $183.87$ & 946 & $1,612.10$ & $192.57$ & $0$ \\ 
NABC score math & $9,948$ & $1,607.88$ & $186.60$ & 932 & $1,569.42$ & $189.95$ & $0$ \\ 
Socioeconomic status (SES) & $7,097$ & $0.45$ & $0.87$ & 3783 & $0.41$ & $0.88$ & $0.02$ \\ 
\bottomrule
\end{tabular} 

	\end{tiny}
	\end{center}

\end{table}

In order to ease the interpretation of estimated preference parameters, we decided to standardize the NABC scores and the SES index to have a mean of zero and unit standard deviation. This is shown in Table \ref{tab:students_summary} in the main text.

\clearpage

\section{Gains from consolidation: using reported preferences}
\label{sec:reportedpref}
In this appendix, we approach the problem of estimating the gains from consolidation from a purely descriptive standpoint. To this end, we take the students' submitted rank order lists (ROLs) as given, and re-compute the SOSM under different district consolidation scenarios. As a benchmark outcome, we use the matching in the consolidated market comprising all districts in Budapest. This matching is denoted by $\mu_{\Omega}$ and it is almost identical to the actual matching observed in the KIFIR dataset. This matching is compared to the matching that obtains in a district-level school market ($\mu_D$). For every student, we compare the match rank obtained in the district-level market to the match rank in the benchmark scenario. This difference in match ranks is used as a measure for the consolidation gains. 

There are two major complications with the aforementioned approach: first, a considerable number of students do not include any school from their home district in their submitted rank order list, and second, some individual school districts cannot actually accommodate all domestic students, even though there is much excess school capacity in the aggregate. These problems lead to a large number of students not being matched in the counter-factual matching. We assume that these unmatched students would prefer being matched rather than  being unmatched, and that the option of being unmatched is as good as the school that they ranked last. In doing so, we obtain a lower bound for the consolidation gains. 

Because district number 23 has only one single school, it does not even offer one school for every track (gymnazium, secondary or vocational). Therefore, we merge this district to its neighbouring district number 20. We show some summary statistics of the district-level and consolidated matches in Table \ref{tab:partialROL:matchstats} below.
\begin{table}[h!]
	\centering
	\caption{Matching statistics based on reported preferences.}
	\label{tab:partialROL:matchstats}
	\begin{footnotesize}
		\begin{tabular}{lr} 
			\toprule
			\it Panel A. Unconsolidated matching&\\
			~~ matched students & 6,554 \\ 
			~~ share top choice match & 0.78 \\ 
			~~ avg. match distance [km] & 3.49 \\ 
			\midrule
			\it Panel B. Consolidated matching&\\
			~~ matched students & 10,494 \\ 
			~~ share top choice match & 0.43 \\ 
			~~ share matched in home district & 0.30 \\
			~~ avg. match distance [km] & 7.10 \\ 
			\bottomrule
		\end{tabular} 
	\end{footnotesize}
\end{table}

\clearpage

Table \ref{tab:district_integration} contains a detailed account of the consolidation gains per district. That table shows that the vast majority of students is strictly better off in the consolidated market, either because they are assigned to a more preferred school in the consolidated market (29\%, i.e. 3,129 students) or because they are unmatched in the unconsolidated market (40\%, i.e. 4,326 students). Only 4\% of the students are assigned to a more preferred school in the unconsolidated market. 
Moreover, there is not a single district in which more students would prefer the unconsolidated market over the integrated market in Budapest. 

\begin{table}[h!]
	\centering
	\caption{Losers ($-$) and winners ($+$) from consolidation (reported preferences).}
	\label{tab:district_integration}
	\begin{footnotesize}
		\begin{threeparttable}
			\begin{tabular}{lrrrrrrr} 
				\toprule 
				District & seats & students & excess seats & $-$ & $0$ & $+$ & unmatched \\ 
				\midrule 
				1 & 338 & 95 & 243 & 3 & 9 & 26 & 50 \\ 
				2 & 1,191 & 634 & 557 & 36 & 241 & 190 & 148 \\ 
				3 & 928 & 743 & 185 & 32 & 263 & 227 & 213 \\ 
				4 & 865 & 746 & 119 & 32 & 319 & 241 & 151 \\ 
				5 & 625 & 217 & 408 & 5 & 50 & 32 & 122 \\ 
				6 & 1,243 & 172 & 1,071 & 14 & 24 & 51 & 77 \\ 
				7 & 1,312 & 212 & 1,100 & 14 & 73 & 48 & 70 \\ 
				8 & 2,524 & 290 & 2,234 & 11 & 79 & 77 & 119 \\ 
				9 & 2,116 & 275 & 1,841 & 19 & 73 & 98 & 77 \\ 
				10 & 2,012 & 591 & 1,421 & 45 & 120 & 224 & 194 \\ 
				11 & 1,025 & 713 & 312 & 13 & 181 & 169 & 347 \\ 
				12 & 956 & 359 & 597 & 17 & 142 & 108 & 90 \\ 
				13 & 3,290 & 449 & 2,841 & 44 & 148 & 152 & 100 \\ 
				14 & 2,893 & 796 & 2,097 & 52 & 189 & 247 & 291 \\ 
				15 & 701 & 454 & 247 & 11 & 99 & 120 & 219 \\ 
				16 & 770 & 659 & 111 & 1 & 96 & 162 & 397 \\ 
				17 & 147 & 628 & -481 & 0 & 40 & 107 & 481 \\ 
				18 & 503 & 873 & -370 & 17 & 177 & 245 & 432 \\ 
				19 & 773 & 444 & 329 & 13 & 68 & 120 & 237 \\ 
				20 & 1,643 & 573 & 1,070 & 31 & 157 & 189 & 189 \\ 
				21 & 2,518 & 641 & 1,877 & 14 & 258 & 204 & 157 \\ 
				22 & 273 & 316 & -43 & 7 & 51 & 92 & 165 \\ 
				\midrule
				Total & 28,646 & 10,880 & 17,766 & 431 & 2,857 & 3,129 & 4,326 \\ 
				\bottomrule
			\end{tabular} 	
			\begin{tablenotes}
				\tiny \item Seats refers to number of seats after removing those given to students from outside Budapest. Excess seats refers to seats minus students. The symbols $-$, $0$ and $+$ denote the number of losers, indifferences and winners from consolidation, respectively. Data obtained using stated preferences.
			\end{tablenotes}
		\end{threeparttable}
		
	\end{footnotesize}
\end{table}

\begin{table}[h!]
	\centering
	\caption{Testing the relationship between consolidation gains and district statistics.}
	\label{tab:integration_gains_tests}
	\begin{footnotesize}
		\begin{threeparttable}
			\begin{tabular}{lc} 
				\toprule
				& \multicolumn{1}{c}{\textit{Dependent variable:} average rank gain} \\ 
				\midrule
				district size (\# students) & $-$0.0008 \\ 
				& (0.0005) \\ 
				relative excess capacity & $-$0.1191$^{**}$ \\ 
				& (0.0478) \\ 
				\midrule
				Observations & 22 \\ 
				\bottomrule 
			\end{tabular} 
			\begin{tablenotes}
				\tiny \item $^{*}$p$<$0.1; $^{**}$p$<$0.05; $^{***}$p$<$0.01. Standard errors in parentheses, intercept not shown. One observation denotes one district. Statistics obtained using stated preferences.
			\end{tablenotes}
		\end{threeparttable}
		
	\end{footnotesize}
\end{table}

Motivated by the general insights of Corollaries \ref{thm:corollary2} and \ref{thm:corollary3}, we investigate whether the share of students who strictly gain from consolidation varies along two key dimensions: district size 
and excess capacity. 
This allows us to examine how our theoretical predictions about the distribution of quantitative rank order gains relates to our empirical results. Corollaries \ref{thm:corollary2} and \ref{thm:corollary3} state that the expected gains from consolidation are larger for smaller markets and markets with less capacity. 
Table \ref{tab:integration_gains_tests} contains the estimated coefficients and standard errors from a regression of average rank order gains per district on the size and capacity per district. The table shows that the coefficient for district size is rather small, and also insignificant, whereas the coefficient for district-level capacity is significantly negative. Therefore, we find robust empirical support for Corollary \ref{thm:corollary2}, but we cannot statistically confirm the validity of Corollary \ref{thm:corollary3} using reported preferences.

\pagebreak

\section{Decomposition of consolidation gains}
\label{sec:decomposition}

In order to explain the large welfare gains we have documented, we isolate the effects of choice and competition in a decomposition exercise. The idea is illustrated in Figure \ref{fig:choiceandcompetition}. We start in the North-West corner at point $o$, where all students are assigned to the schools in the unconsolidated market. We then move to the South-West corner in that we keep an individual student $t$ fixed, and assign her to the most preferred feasible school in the global market, given that all other students are still constrained to attend only schools within their district. 

The choice effect (choice II) is the change of student $t$'s welfare as her choice set is expanded to include all schools, keeping the other students' choice sets constant. 
The competition effect (competition II in Figure \ref{fig:choiceandcompetition}) is then the change of student $t$'s welfare as all other students' choice sets are enlarged to include the entire integrated market. This is illustrated in Figure \ref{fig:choiceandcompetition} by moving to the South-East corner at point $x$, where all students are assigned to the schools in the consolidated market. This procedure is repeated for all students, and the results are aggregated.

\begin{figure}[h!]
	\begin{center}
		
		\begin{tikzpicture}
			\filldraw[black] (-2,2.5) node[anchor=west] {Other students' choice set};
			
			\filldraw[black] (-3.8,.3) node[anchor=west] {\rotatebox{90}{Student t's choice set}};
			\matrix (mat) [table]
			{
				& district & global      \\
				district   & $o$  	  &  \\
				global     & 	      & $x$ \\
			};
			\foreach \x in {1,...,3}
			{
				\draw 
				([xshift=-.5\pgflinewidth]mat-\x-1.south west) --   
				([xshift=-.5\pgflinewidth]mat-\x-3.south east);
			}
			\foreach \x in {1,...,3}
			{
				\draw 
				([yshift=.5\pgflinewidth]mat-1-\x.north east) -- 
				([yshift=.5\pgflinewidth]mat-3-\x.south east);
			}    
			\begin{scope}[shorten >=7pt,shorten <= 7pt]
				\draw[->,blue]  (mat-2-2.center) -- (mat-2-3.center) node [midway,above,sloped] (TextNode) {\tiny competition I};
				\draw[->,blue]  (mat-2-3.center) -- (mat-3-3.center) node [midway,above,sloped] (TextNode) {\tiny choice I};
				\draw[->,dashed,red] (mat-2-2.center) -- (mat-3-2.center) node [midway,below,sloped] (TextNode) {\tiny choice II};
				\draw[->,dashed,red] (mat-3-2.center) -- (mat-3-3.center) node [midway,below,sloped] (TextNode) {\tiny competition II};
			\end{scope}
		\end{tikzpicture}
		
	\end{center}
	\caption{Decomposition of the gains from market consolidation into choice and competition effects}
	\label{fig:choiceandcompetition}
\end{figure}

As this figure shows, there are always two ways to measure either the choice, or the competition effect.\footnote{For the type-I competition effect, we constrain student $t$ to attend only schools within her district, while we allow all other students to attend schools from the entire global market. For the type-I choice effect, we also relax the constraint for student $t$.} We refer to the resulting statistics as type-I and type-II effects. \ref{subsec:compute_decomposition} presents in detail how we construct the decomposition of the students' consolidation gains into a choice effect and a competition effect. \ref{subsec:explain_decomposition} presents regression evidence explaining consolidation gains with student- and district-level characteristics.

\subsection*{Computing the decomposition of consolidation gains}
\label{subsec:compute_decomposition}

To compute the decomposition of consolidation gains, we make use of the large market approximation to matching markets by which school-specific cutoff scores play the role of prices that balance the supply of, and the demand for school seats \citep{azevedo2016supply}. The cutoff score at school $s$ under the matching $\mu$ is the lowest valuation among all students who where admitted to that school under $\mu$, or
\[
c_s(\mu) = \min_{t \in \mu(s)} V_s(t)
\]
We assume that the school market consists of relatively few schools and a large number of students so that the addition (or deletion) of a single student has practically no effect on a schools' cutoff score. In order to decompose the total consolidation gains, we compute the school-level cutoff scores under the district wise matching $\mu_D$ and under the integrated matching $\mu_{\Omega}$. 

The effect of increased choice, keeping everything else constant, can then pe computed as the difference between student $t$ being matched to her most preferred feasible school in her own district, and globally, using either the district-level or the city-wide cutoffs. Let the feasible set of student $t$ under the cutoffs $\{c_s(\mu)\}_{s \in S}$ be 
\[
\mathcal{F}_t^\mu = \left\{s \in S: V_s(t) \geq c_s(\mu) \right\}
\]
and denote the set of schools in district $d$ as $S^d$. Then, the choice gain of student $t$ can either be expressed as
\[
\Delta^{ch-I} U_t = \max_{s \in \mathcal{F}_t^{\mu_{\Omega}}} U_t(s) - \max_{s \in \mathcal{F}_t^{\mu_{\Omega}} \cap S^d} U_t(s)
\]
or
\[
\Delta^{ch-II} U_t = \max_{s \in \mathcal{F}_t^{\mu_{D}}} U_t(s) - \max_{s \in \mathcal{F}_t^{\mu_{D}} \cap S^d} U_t(s)
\]
as is illustrated in Figure \ref{fig:choiceandcompetition}. The only difference between $\Delta^{ch-I} U_t$ and $\Delta^{ch-II} U_t$ is the usage of a different baseline scenario to compute the cutoffs -- the global cutoffs $\{c_s(\mu_{\Omega})\}$ for $\Delta^{ch-I} U_t$ and the local cutoffs for $\Delta^{ch-II} U_t$. It is easy to see that the choice gains will always be weakly positive by construction.\footnote{It can also happen that a student is not assigned in one of the counter-factual scenarios. In our empirical application, the choice gains $\Delta^{ch-I} U_t$ are missing for about one quarter of all students, because their set of feasible schools within their home district is empty under the global cutoff scores.} 

In a similar manner, one can compute the change in student $t$'s welfare as the market is opened up to external competition. We call this change a competition gain, but it is not a priori clear whether students actually gain or lose from competition. The competition gain can be computed either as
\[
\Delta^{co-I} U_t = \max_{s \in \mathcal{F}_t^{\mu_{\Omega}} \cap S^d} U_t(s) - \max_{s \in \mathcal{F}_t^{\mu_{D}} \cap S^d} U_t(s)
\]
or as 
\[
\Delta^{co-II} U_t = \max_{s \in \mathcal{F}_t^{\mu_{\Omega}}} U_t(s) - \max_{s \in \mathcal{F}_t^{\mu_{D}}} U_t(s)
\]
Now, $\Delta_t^{co-I}$ differs from $\Delta^{co-II} U_t$ in that student $t$'s choice set is restricted to feasible schools within her home district $d$ in the former, but not in the latter. It is easy to see that the sum of $\Delta^{ch-I} U_t$ and $\Delta^{co-I} U_t$ is identical to the sum of $\Delta^{ch-II} U_t$ and $\Delta^{co-II} U_t$ unless some type-I choice gains are missing. Also, the sum of the choice and competition gains are equal to the total welfare gains.

\subsection*{Decomposition results}

Table \ref{tab:fullROL:gainstats2} below extends the summary statistics of the total gains reported in Table \ref{tab:fullROL:gainstats} by the decomposition into choice and competition effects that are calculated in both ways.\footnote{In general, the sum of the competition and choice effects of either type should be equal to the total welfare effect of consolidation. However, because not all students are assigned to a school in the district level matching (c.f. Table \ref{tab:fullROL:matchstats}), the type-I competition effect and the type-I choice effect cannot be computed for all students. This, however, affects only very few students, and so the average choice and competition effects approximately add up to the total gains.}

\begin{table}[h!]
	\centering
	\caption{Measures of consolidation gains in latent utility changes.}
	\label{tab:fullROL:gainstats2}
	\begin{footnotesize}
		\begin{tabular}{lrrrrrr}
			\toprule
			& Mean  & SD    & Min   & Median & Max   & N \\
			\midrule
			\textit{total gains} &       &       &       &       &       &  \\
			~~ in latent utility units & 0.819 & 0.916 & -1.895 & 0.600 & 5.799 & 9,986 \\
			~~ in equivalent kilometres & 5.532 & 6.187 & -12.805 & 4.054 & 39.180 & 9,986 \\
			&       &       &       &       &       &  \\
			\textit{decomposition (in equ.\ kilometres)} &       &       &       &       &       &  \\
			~~ choice effect I & 5.068 & 5.392 & 0.000 & 3.716 & 33.865 & 10,880 \\
			~~ competition effect I & 0.695 & 3.986 & -24.625 & 0.000 & 36.081 & 9,986 \\
			~~ choice effect II  & 5.844 & 6.067 & 0.000 & 4.480 & 39.180 & 9,986 \\
			~~ competition effect II  & -0.270 & 1.989 & -20.270 & 0.000 & 13.514 & 10,880 \\
			\bottomrule
		\end{tabular}%
	\end{footnotesize}
\end{table}

The results show that the choice effects account for the vast share of total welfare gains, whereas the average competition effects are much smaller in magnitude, and vary in sign. Whereas the average type-II competition effect is small and negative, the type-I competition effect is small and positive. 
Therefore, it remains an open question whether competition is stronger in the consolidated market, or in the district-level markets.\footnote{At first glance, it may seem counter-intuitive that competition could be weaker in the consolidated, aggregate market. But this can be explained by the fact that the school districts are very different. A few districts have a large number of school seats that far exceeds the number of their domestic students (see Table \ref{tab:district_integration}). Whereas the market tightness increases for students in those districts as all districts are integrated, the aggregate market tightness may decrease as a result. Therefore, the majority of students may experience more favourable competition in the aggregate market.}

\subsection*{Heterogeneity in decomposed gains} 
\label{subsec:explain_decomposition}

To further explain the gains from market consolidation, we regress the student-level gains, and the competition and choice effects that were computed above, on student- and district-level observables. Table \ref{tab:explaining_gains} shows the results that extend Table \ref{tab:explaining_gains1} by the decomposition into choice and competition effects. While the results for the type-I and type-II decomposition are similar, there are notable differences. 

\begin{table}[h!]
	\centering
	\caption[Explaining consolidation gains]{Explaining gains from consolidation with students observables.} 
\label{tab:explaining_gains}
\begin{footnotesize}
\begin{threeparttable}
	\begin{tabular}{lccccc}
		\toprule
		& \multicolumn{5}{c}{\textit{Dependent variable}: latent utility gains in equiv.\ kilometres} \\
		&       & \multicolumn{2}{c}{type-I decomposition} & \multicolumn{2}{c}{type-II decomposition} \\
		\cmidrule{3-4} 
		\cmidrule{5-6}
		& total & choice & competition & choice & competition \\
		& (1)   & (2)   & (3)   & (4)   & (5) \\
		\midrule
		\multicolumn{1}{l}{socio-economic } & 0.0574 & 0.0953 & -0.0453 & 0.1264$^{**}$ & -0.0655$^{***}$ \\
		status (SES)   & (0.0655) & (0.0581) & (0.0419) & (0.0642) & (0.0230) \\
		&       &       &       &       &  \\
		\multicolumn{1}{l}{ability} & 0.0966$^{**}$ & 0.2993$^{***}$ & -0.1608$^{***}$ & 0.1801$^{***}$ & -0.0939$^{***}$ \\
		& (0.0473) & (0.0419) & (0.0304) & (0.0466) & (0.0169) \\
		&       &       &       &       &  \\
		\multicolumn{1}{l}{district size } & -0.9514$^{***}$ & -1.4212$^{***}$ & 0.4676$^{***}$ & -0.9747$^{***}$ & 0.0257 \\
		& (0.1162) & (0.1081) & (0.0743) & (0.1142) & (0.0426) \\
		&       &       &       &       &  \\
		\multicolumn{1}{l}{relative excess } & -2.0939$^{***}$ & -0.3084$^{**}$& -1.8378$^{***}$ & -1.9962$^{***}$ & -0.0784 \\
		capacity   & (0.1662) & (0.1493) & (0.1061) & (0.1632) & (0.0591) \\
		&       &       &       &       &  \\
		\multicolumn{1}{l}{gymnazium} & -0.8838$^{***}$ & -1.6804$^{***}$ & 0.5294$^{***}$ & -1.0432$^{***}$ & 0.2486$^{***}$ \\
		& (0.2378) & (0.2101) & (0.1520) & (0.2342) & (0.0831) \\
		&       &       &       &       &  \\
		\multicolumn{1}{l}{secondary} & -0.1946 & -0.4520$^{**}$& 0.0905 & -0.2865 & 0.1270 \\
		& (0.2230) & (0.1973) & (0.1426) & (0.2203) & (0.0777) \\
		&       &       &       &       &  \\
		\multicolumn{1}{l}{constant} & 15.4959$^{***}$ & 12.9729$^{***}$ & 2.8108$^{***}$ & 15.5351$^{***}$ & -0.1453 \\
		& (0.5174) & (0.4662) & (0.3311) & (0.5101) & (0.1842) \\
		&       &       &       &       &  \\
		\midrule
		\multicolumn{1}{l}{district FE} & Yes   & Yes   & Yes   & Yes   & Yes \\
		\midrule
		\multicolumn{1}{l}{Observations} & 9,986 & 10,880 & 9,986 & 9,986 & 10,880 \\
		\bottomrule
	\end{tabular} 
	\begin{tablenotes}
		\tiny \item The table shows regression coefficients of students' gains on student observables.	Variables 'district size' and 'relative excess capacity' refer to the students' home districts; the school type refers to the school type of the assigned school in the integrated market. $^{*}$p$<$0.1; $^{**}$p$<$0.05; $^{***}$p$<$0.01. Standard errors in parentheses.
	\end{tablenotes}
\end{threeparttable}

\end{footnotesize}
\end{table}

For example, as discussed in Section \ref{subsec:heterogeneity}, the first column of this table shows that students with a higher socio-economic status (SES) do not benefit significantly more from district consolidation. The second and third columns reveal that this is because students with a higher SES benefit more from increased choice, but benefit less from the more favourable competitive conditions in the consolidated market. 
While these effects are insignificant for the type-I effects, the fourth and fifth column pertaining to the type-II effects are indeed significantly different from zero.
A similar, but exacerbated pattern can be observed for students with higher academic ability. High-ability students benefit more from district consolidation than average students, and they benefit comparatively more from an enhanced choice set, and less from more relaxed competitive conditions in the aggregate. This effect is statistically significant.

\end{document}